\documentclass[journal, onecolumn]{IEEEtran}

\usepackage{bm}
\usepackage{cite}
\usepackage{amsmath, amssymb, amsthm, mathtools}
\usepackage{graphicx, subfigure}
\usepackage[noend]{algpseudocode}
\usepackage{algorithmicx,algorithm}
\usepackage{color}
\usepackage{enumitem}
\usepackage{overpic}
\usepackage{booktabs}

\definecolor{my_dark_brown}{rgb}{0.547,0.161,0.161}
\definecolor{my_light_green}{rgb}{0.471,0.580,0.251}
\definecolor{my_light_blue}{rgb}{0.298,0.709,0.862}
\definecolor{my_dark_blue}{rgb}{0.075,0.333,0.435}
\definecolor{my_purple}{rgb}{0.412,0.072,0.365}

\providecommand{\mathbold}[1]{\bm{#1}}
\newcommand{\vct}[1]{\mathbold{#1}}
\newcommand{\mtx}[1]{\mathbold{#1}}
\def \R 	{\mathbb{R}}
\def \B 	{\mathbb{B}}
\def \S 	{\mathbb{S}}
\def \P 	{\mathbb{P}}
\def \E 	{\mathbb{E}}

\newcommand{\mA}{\mtx{A}}

\newcommand{\mI}{\mtx{I}}

\newcommand{\vu}{\vct{u}}
\newcommand{\va}{\vct{a}}
\newcommand{\vb}{\vct{b}}

\newcommand{\vy}{\vct{y}}
\newcommand{\vz}{\vct{z}}

\newcommand{\vs}{\vct{s}}
\newcommand{\vt}{\vct{t}}
\newcommand{\vv}{\vct{v}}
\newcommand{\vw}{\vct{w}}
\newcommand{\vd}{\vct{d}}

\newcommand{\vg}{\vct{g}}
\newcommand{\vh}{\vct{h}}
\newcommand{\vO}{\bm{0}}
\newcommand{\vtau}{\vct{\tau}}
\newcommand{\vepsilon}{\vct{\epsilon}}

\newcommand{\vkappa}{\vct{\kappa}}
\newcommand{\vomega}{\vct{\omega}}

\newcommand{\PP}{\mathcal{S}}

\newcommand{\K}{K}
\newcommand{\fdc}{\mathcal{C}}
\newcommand{\fds}{\mathcal{S}}

\newcommand{\fx}{\vct{x}}

\newcommand{\fxs}{\vct{x}^{\star}}

\DeclareMathOperator{\dist}{dist}
\DeclareMathOperator{\aff}{aff}
\DeclareMathOperator{\var}{Var}

\DeclareMathOperator{\interior}{int}
\DeclareMathOperator{\ri}{ri}

\DeclareMathOperator{\nullspace}{null}

\DeclareMathOperator{\closure}{cl}
\DeclareMathOperator{\cone}{cone}

\newcommand{\argmin}{\operatorname*{arg\; min}}

\newcommand{\st}{\operatorname*{s.t.}}

\floatname{algorithm}{Recipe}

\newtheorem{theorem}{Theorem} 
\newtheorem{lemma}{Lemma}

\newtheorem{proposition}{Proposition}
\newtheorem{fact}{Fact}

\newtheorem{corollary}{Corollary}

\newtheorem{definition}{Definition}
\newtheorem{remark}{Remark}

\makeatletter

\newcommand{\Rmnum}[1]{\expandafter\@slowromancap\romannumeral #1@}
\makeatother

\begin{document}

\title{Phase Transition of Convex Programs for Linear Inverse Problems with Multiple Prior Constraints}

\author{Huan~Zhang,
        Yulong~Liu,
	  	and~Hong~Lei}

%

\maketitle

\begin{abstract}
  	A sharp phase transition emerges in convex programs when solving the linear inverse problem, which aims to recover a structured signal from its linear measurements. This paper studies this phenomenon in theory under Gaussian random measurements. Different from previous studies, in this paper, we consider convex programs with multiple prior constraints. These programs are encountered in many cases, for example, when the signal is sparse and its $\ell_2$ norm is known beforehand, or when the signal is sparse and non-negative simultaneously. Given such a convex program, to analyze its phase transition, we introduce a new set and a new cone, called the prior restricted set and prior restricted cone, respectively. Our results reveal that the phase transition of a convex problem occurs at the statistical dimension of its prior restricted cone. Moreover, to apply our theoretical results in practice, we present two recipes to accurately estimate the statistical dimension of the prior restricted cone. These two recipes work under different conditions, and we give a detailed analysis for them. To further illustrate our results, we apply our theoretical results and the estimation recipes to study the phase transition of two specific problems, and obtain computable formulas for the statistical dimension and related error bounds. Simulations are provided to demonstrate our results.
\end{abstract}


%
\IEEEpeerreviewmaketitle

\section{Introduction} \label{sec: introduction}
\IEEEPARstart{T}{he} linear inverse problem refers to the problem of recovering an unknown signal from its linear measurements. It is frequently encountered in many applications, such as image processing \cite{LUSTIG2007}, network data analysis \cite{HAUPT2008} and so on. In practice, we often have less measurements than the dimension of the true signal. As a result, the problem is generally ill-posed. Therefore, to make recovery possible, we may assume that the true signal has low complexity under some structures. Commonly considered structures include sparsity and low rank, and the corresponding recovery problems are known as \textit{compressed sensing} and \textit{matrix completion}.

Given the structures of the signal, a popular approach for recovery is to solve a convex program that enforces the known prior information about the structures. For example, we pursue a sparse recovery through $\ell_1$ norm minimization in the compressed sensing problem, and a low-rank recovery through nuclear norm minimization in the matrix completion problem. This approach is shown to be simple and efficient in many practical applications. 

Meanwhile, a sharp phase transition is numerically observed, when we use convex programs to recover structured signals. The phase transition refers to the phenomenon that for a certain convex program, when the measurement number is greater than some threshold, it succeeds with high probability; while when the measurement number is smaller than another threshold, it fails with high probability. When we say a sharp phase transition, we mean that the transition region is very narrow. This phenomenon has attracted many researchers, and much work has been done to explain it in theory in the past several years. Some exciting results have been obtained since then.

In \cite{DONOHO2005, DONOHO2009, DONOHO2010a, DONOHO2010b}, Donoho and Tanner analyzed the phase transition of the compressed sensing problem in the asymptotic regime. They first demonstrated that the $\ell_1$ minimization approach succeeds if and only if the random projection preserves the structure of faces of \textit{cross-polytope}, and then used the theory of polytope angles to deal with this problem. In \cite{DONOHO2013a, DONOHO2013b, OYMAK2016}, the authors established a connection between the phase transition and the statistical decision theory, and revealed that the phase transition curve coincides with the minimax risk curve of denoising in many linear inverse problems. In \cite{AMEL2014}, Amelunxen \textit{et al.} presented a comprehensive analysis of the phase transition of convex programs in the linear inverse problem. They first formulated the phase transition problem to a geometry problem, then used tools from the theory of conic integral geometry to study this geometry problem. The results show that the phase transition of convex programs occurs at the \textit{statistical dimension} of the \textit{descent cone} of the structure inducing function at the true signal. In \cite{RUDELSON2008}, Rudelson and Vershynin studied the performance of the $\ell_1$ minimization approach using the ``escape from the mesh'' theorem \cite{GORDON1988} in Gaussian process theory. Later, their ideas were extended in the papers \cite{CHAN2012, TROPP2015, AMEL2014}, and the phase transition were identified by incorporating the arguments of Rudelson and Vershynin with a polarity argument. The obtained results are stated in terms of \textit{Gaussian width}, and consistent with the results in \cite{AMEL2014}. In \cite{BAYATI2015}, Bayati \textit{et al.} made use of a state evolution framework, inspired by ideas from statistical physics, and demonstrated that the phase transition of $\ell_1$ minimization is universal over a class of sensing matrices. Recently, in \cite{OYMAK2015}, Oymak and Tropp demonstrated the universality laws for the phase transition of convex programs for linear inverse problems, over a class of sensing matrices.

However, most of the above work focuses on the case when we have no additional prior constraints. But in many practical problems, we do have some additional prior information. For example, in image processing problems, in addition to the structures about texture etc, the fact that the pixel values are non-negative may help to recover the true image. In these cases, we would solve convex problems with (multiple) prior constraints to recover the true signal. While these problems exhibit a sharp phase transition as well, theoretical understanding of the phase transition is far from satisfactory. We mention that in \cite{DONOHO2005, DONOHO2009}, Donoho and Tanner studied the $\ell_1$ minimization problem with an additional non-negativity constraint, and ``weak threshold'' and ``strong threshold'' were obtained in the asymptotic regime, which marks the phase transition. Nevertheless, a comprehensive analysis about the phase transition of this problem does not exist. Furthermore, when the signal has structures other than sparsity, or when we have prior constraints other than non-negativity, it remains an open problem to prove the existence and identify the location of the phase transition.

In this paper, we study the phase transition of convex programs with multiple prior constraints under Gaussian random measurements. In our analysis, we first introduce a new set and a new cone, called the prior restricted set and prior restricted cone, respectively. Next, we give a sufficient and necessary condition for the success of convex programs, which involves the prior restricted cone. It states that convex programs succeed if and only if the intersection of the null space of the sensing matrix and the prior restricted cone contains only the origin. This condition has been well studied by Amelunxen \textit{et al.} in \cite{AMEL2014} using the theory of conic integral geometry. Utilizing their results, we obtain that the phase transition of convex programs with multiple prior constraints occurs at the statistical dimension of the prior restricted cone. Thus, intuitively, the ``dimension'' of the prior restricted cone (i.e., the statistical dimension of this cone) can be seen as a measure of how much we know about the true signal from the prior information, if convex programs are used to recover signals. Moreover, to apply our theoretical results in practice, we present two recipes to accurately estimate the statistical dimension of the prior restricted cone. The two recipes work under different conditions, and we give a detailed analysis for them. To further illustrate our results, we apply our theoretical results and the estimation recipes to study the phase transition of two specific problems: One is the linear inverse problem with $\ell_2$ norm constraints, and the other is the linear inverse problem with non-negativity constraints. We obtain computable formulas for the statistical dimension and related error bounds in either problem. The following simulations demonstrate that our results match the empirical successful probability perfectly.

The rest of the paper is organized as follows: In section \ref{sec: problem formulation}, we give a precise statement of the problems studied in this paper. In section \ref{sec: preliminaries and notations}, some preliminaries and notations are introduced. In section \ref{sec: main results}, we state our main results. In section \ref{sec: applications}, we apply our main results to study the phase transition of two specific problems. In section \ref{sec: simulation results}, simulations are provided to demonstrate our theoretical results. In section \ref{sec: conclusion}, we conclude the paper.

\section{Problem Formulation} \label{sec: problem formulation}

In this section, we provide a precise statement of the problems studied in this paper. In section \ref{subsec: model}, we introduce the linear inverse problem. In section \ref{subsec: approach}, we introduce the convex optimization procedure to recover signals from compressed, linear measurements.

\subsection{Linear Inverse Problem} \label{subsec: model}
In the linear inverse problem, we observe a signal via its linear measurements:
\begin{equation} \label{eq: model}
  \vy = \mA \fx^{\star},
\end{equation}
where $\vy \in \R^m$ is the measurement vector, $\mA \in \R^{m \times n}$ is the sensing matrix, and $\fx^{\star} \in \R^n$ is the unknown signal. Our goal is to recover $\fxs$ given the knowledge of $\vy$ and $\mA$.

\subsection{Convex Optimization Procedure} \label{subsec: approach}
In many applications, we often have compressed measurements, i.e., $m < n$. As a result, to recovery $\fxs$ from $\vy$ and $\mA$ is an ill-posed problem. Hence, to make recovery possible, it is commonly assumed that the signal $\fxs$ is well structured. In this case, a simple yet efficient approach for recovery is to solve a convex program, which forces the solution to have the corresponding structures. Moreover, apart from the assumed structures, we may have some additional prior information about $\fxs$. For example, we may know the $\ell_2$ norm of $\fxs$ beforehand, or the signal $\fxs$ is non-negative. The additional prior information often acts as constraints.

Suppose that $f_0: \R^n \rightarrow \overline{\R}$ is a proper convex function and promotes the structures of $\fxs$, and $f_i: \R^n \rightarrow \overline{\R}, 1 \le i \le k,$ are some proper convex functions and promote the additional prior information of $\fxs$. Then in practice the following convex program is often used to recover the true signal $\fxs$:

\begin{equation} \label{eq: problem with multiple prior}
 		\min f_0(\fx), 
		\quad \st \  \vy = \mA \fx,
  		\   f_i(\fx) \le f_i(\fxs), \ i = 1,\dots,k.
\end{equation}
We say that the convex problem \eqref{eq: problem with multiple prior} \textit{succeeds} if the unique solution $\hat{\fx}$ satisfies $\hat{\fx} = \fxs$; otherwise, we say it \textit{fails}.

In this paper, we study the phase transition of problem \eqref{eq: problem with multiple prior}. The analysis relies on some knowledge from convex analysis and convex geometry. Hence, in the next section, we give a brief introduction about the needed knowledge.

\section{Preliminaries} \label{sec: preliminaries and notations}

In this section, we present some preliminaries that will be used in our analysis.
\subsection{Subgradient}
Suppose $h: \R^n \rightarrow \overline{\R}$ is a proper convex function. Then the \textit{subdifferential} of $h$ at $\vz \in \R^n$ is the set
$$
\partial h(\vz) = \big\{ \vu \in \R^n:h(\vz + \vt) \ge h(\vz) + \left<\vu, \vt\right> \ \textnormal{for all} \ \vt \in \R^n \big\}.
$$
\subsection{Descent Cones and Normal Cones of Convex Functions}
The descent cone of a proper convex function $h: \R^n \rightarrow \overline{\R}$ at $\vz \in \R^n$ is the set of all non-ascent directions of $h$ at $\vz$:
$$
D(h, \vz) = \big\{ \vd \in \R^n: \exists \, a > 0, h( \vz + a \cdot \vd) \le h(\vz) \big\}.
$$
The \textit{normal cone} of a proper convex function $h: \R^n \rightarrow \overline{\R}$ at $\vz \in \R^n$ is the polar of the descent cone of $h$ at $\vz$:
$$
N(h, \vz) = D(h, \vz)^{\circ} = \big\{ \vu \in \R^n: \left< \vu, \vd \right> \le 0 \ \textnormal{for all} \ \vd \in D(h, \vz) \big\}.
$$
Suppose $\partial h(\vz)$ is non-empty, compact, and does not contain the origin, then the normal cone is the cone generated by the subdifferential \cite[Corollary 23.7.1]{ROCK1970}:
\begin{equation*} \label{eq: relation_subdiff_and_normal}
  N(h, \vz) = \cone\big(\partial h(\vz) \big) = \big\{ \vu \in \R^n: \exists \, \tau \ge 0, \vu \in \tau \cdot \partial h(\vz) \big\}.
\end{equation*}

\subsection{Normal Cone to Convex Sets}
Let $C \subseteq \R^n$ be a convex set with $\bar{\fx} \in C$. The \textit{normal cone} to $C$ at $\bar{\fx}$ is
$$
N(\bar{\fx}; C) \coloneqq \{ \vv \in \R^n : \ \left< \vv, \fx - \bar{\fx}\right> \le 0, \ \forall \, \fx \in C \}.
$$

\subsection{Statistical Dimension of Convex Cones}
For a convex cone $\K$, the \textit{statistical dimension} of $\K$ is defined as:
$$
\delta(\K) = \E \Big( \sup_{\vt \in \K \cap \B^n} \left< \vg, \vt \right> \Big)^2, \text{ where }\vg \sim N(\vO,\mI_n).
$$
The statistical dimension of a convex cone has a number of important properties, see \cite[Proposition 3.1]{AMEL2014}. 
Moreover, the statistical dimension satisfies the following additivity property:
\begin{fact} [Additivity of statistical dimension] \label{fact: statistical dimension of sum}
  	Let $\K_1$ and $\K_2$ be two convex cones in $\R^n$. The following holds:
	\begin{enumerate}
		\item
		  	If for any $\va \in \K_1$ and $\vb \in \K_2$, we have $\left< \va, \vb \right> = 0$. Then
			$$
			\delta ( \K_1 + \K_2 ) = \delta( \K_1 ) + \delta (\K_2).
			$$
		\item
		  	If for any $\va \in \K_1$ and $\vb \in \K_2$, we have $\left< \va, \vb \right> \le 0$. Then
			$$
			\delta ( \K_1 + \K_2 ) \ge \delta( \K_1 ) + \delta (\K_2).
			$$
		\item
		  	If for any $\va \in \K_1$ and $\vb \in \K_2$, we have $\left< \va, \vb \right> \ge 0$. Then
			$$
			\delta ( \K_1 + \K_2 ) \le \delta( \K_1 ) + \delta (\K_2).
			$$
	\end{enumerate}
\end{fact}

\begin{proof}
  	See Appendix \ref{sec: proof of statistical dimension of sum}.
\end{proof}

Fact \ref{fact: statistical dimension of sum} generalizes the fact that for two linear subspaces $L_1$ and $L_2$, suppose $L_1 \perp L_2$, then $\dim(L_1+L_2) = \dim(L_1) + \dim(L_2)$, since the statistical dimension extends the dimension of a linear subspace to the class of convex cones \cite{AMEL2014}.

\subsection{Indicator Function of a Convex Set}
Let $C \subseteq \R^n$ be a convex set. Then the \textit{indicator function} of the set $C$ is defined as
\begin{displaymath}
	I_C(\fx) = \left\{
	\begin{array}{cl}
	  	0, & \textnormal{when} \ \fx \in C, \\
	  	\infty, & \textnormal{when} \ \fx \notin C.
	\end{array}
\	\right.
\end{displaymath}
For any $\bar{\fx} \in C$, the subdifferential of $I_C$ is \cite[Example 2.32]{MORD2014}:
\begin{equation} \label{eq: subdifferential of indicator function}
	\partial I_C(\bar{\fx}) = N(\bar{\fx}; C) = \{\vv \in \R^n: \left< \vv, \fx-\bar{\fx}\right> \le 0, \ \forall \, \fx \in C \}.
\end{equation}

\subsection{Prior Restricted Set and Prior Restricted Cone}
We first define the \textit{prior restricted set} of convex problem \eqref{eq: problem with multiple prior}:
\begin{definition} [Prior Restricted Set] \label{def: prior restricted set}
  	For the convex problem \eqref{eq: problem with multiple prior}, suppose $\fxs \in \R^n$ is the true signal, then we define its prior restricted set as the following set:
	\begin{equation*} \label{eq: prior restricted set}
		\PP = \big\{ \vd \in \R^n: f_i(\fxs + \vd) \le f_i(\fxs), \ i =0, 1, \dots, k \big\}.
  	\end{equation*}
\end{definition}
Using this set, we can define the prior restricted cone of problem \eqref{eq: problem with multiple prior}:

\begin{definition} [Prior Restricted Cone] \label{def: prior restricted cone}
  	For the convex problem \eqref{eq: problem with multiple prior}, suppose $\fxs \in \R^n$ is the true signal, then we define its prior restricted cone as the following set:
	$$
	\fdc = \cone(\PP) =  \{ \vu \in \R^n: \exists \, t > 0, f_i(\fxs + t \cdot \vu) \le f_i(\fxs), \ i =0, 1, \dots, k \}.
	$$
\end{definition}

\subsection{Notations}
Throughout, we denote $\R^n_+$ the non-negative orthant in $\R^n$: $\R^n_+ \coloneqq \{ \fx \in \R^n: \fx_i \ge 0 \ \textnormal{for} \ 1 \le i \le n\}$, and $\R^n_{++}$ the positive part: $\R^n_{++} \coloneqq \{ \fx \in \R^n: \fx_i > 0 \ \textnormal{for} \ 1 \le i \le n\}$.

For a set $C \in \R^n$, we use $\interior(C)$ to denote its \textit{interior}:
$$
\interior(C) \coloneqq \{\fx \in C: \B(\fx, r) \subseteq C \ \textnormal{for some} \ r > 0\}.
$$
Denote $\aff (C)$ the \textit{affine hull} of $C$:
$$
\aff (C) = \{ \theta_1 \fx_1 + \dots + \theta_k \fx_k : \, \fx_1,\dots,\fx_k \in C, \ \theta_1 + \dots + \theta_k = 1\},
$$
and $\ri(C)$ the \textit{relative interior} of the set $C$:
$$
\ri(C) = \{\fx \in C: \, \B(\fx, r) \cap \aff (C) \subseteq C \ \textnormal{for some} \ r > 0\}.
$$
The closure of $C$ is denoted by either $\overline{C}$ or $\closure(C)$. 

Given a point $\vu \in \R^n$ and a subset $C \subseteq \R^n$, the distance of $\vu$ to the set $C$ is denoted by $\dist(\vu, C)$:
$$
\dist(\vu, C) \coloneqq \inf_{\fx \in C} \|\vu - \fx\|_2.
$$
We denote $\Pi_{C}(\vu)$ the projection of $\vu$ onto the set $C$:
$$
\Pi_C(\vu) := \{\fx \in C: \|\vu - \fx\|_2 = \dist(\vu, C) \}.
$$
If $C$ is non-empty, convex, and closed, the projection $\Pi_C(\vu)$ is a singleton. In this case, $\Pi_C(\vu)$ may denote the unique point in it, depending on the context.

\section{Main Results} \label{sec: main results}

In this section, we state our main results in this paper. We first give results about the phase transition of problem \eqref{eq: problem with multiple prior} in subsection \ref{subsec: phase transition of convex programs}, and then present two recipes to estimate the statistical dimension of the prior restricted cone in subsections \ref{subsec: calc_statistical_prior_1} and \ref{subsec: calc_statistical_prior_2}.

\subsection{Phase Transition of Convex Programs with Multiple Prior Constraints} \label{subsec: phase transition of convex programs}

In this subsection, we state our main results about the phase transition of problem \eqref{eq: problem with multiple prior}. We begin by a geometry condition which determines the success of problem \eqref{eq: problem with multiple prior}:

\begin{lemma} [Optimality condition] \label{lem: optimality condition}
  	Consider problem \eqref{eq: problem with multiple prior} to recover the true signal $\fxs$. If $f_i$ is a proper convex function for any $0 \le i \le k$, problem \eqref{eq: problem with multiple prior} succeeds if and only if
	$$
	\fdc \cap \nullspace(\mA) = \{ \vO \},
	$$
	where $\fdc$ denotes the prior restricted cone of problem \eqref{eq: problem with multiple prior}.
\end{lemma}

\begin{proof}
  	See Appendix \ref{app: proof of main results}.
\end{proof}

Fig. \ref{fig: noise-free} gives a geometric interpretation of Lemma \ref{lem: optimality condition}. Note that when there is no additional prior constraint, i.e., when we consider problem
\begin{equation} \label{eq: problem no addi_prior}
	\min f_0(\fx), \quad \st  \ \vy = \mA \fx
\end{equation}
to recover $\fxs$, the prior restricted cone is exactly $D(f_0,\fxs)$, the descent cone of $f_0$ at $\fxs$. In this case, our optimality condition, Lemma \ref{lem: optimality condition}, will degenerate to the optimality condition given by Chandrasekaran \textit{et al.} in \cite[Fact 2.8]{CHAN2012}. 

\begin{figure*}
  	\centering
	\subfigure[]{
		\begin{overpic}[scale=0.49]{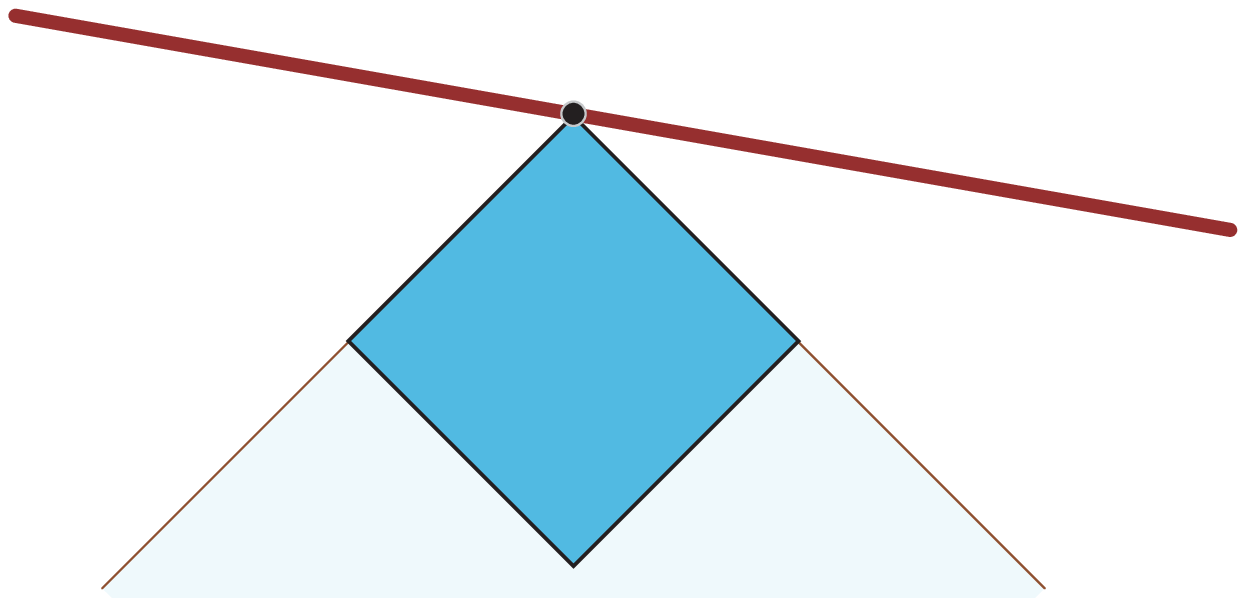}
		  	\put(90,44){\color{my_dark_brown}{$\nullspace(\mA)$}}
			\put(28,25){\color{my_light_blue}{$\fdc$}}
			\put(41,39){\color{my_dark_blue}{$\PP$}}
		\end{overpic}
  	}
	\subfigure[]{
		\begin{overpic}[scale=0.49]{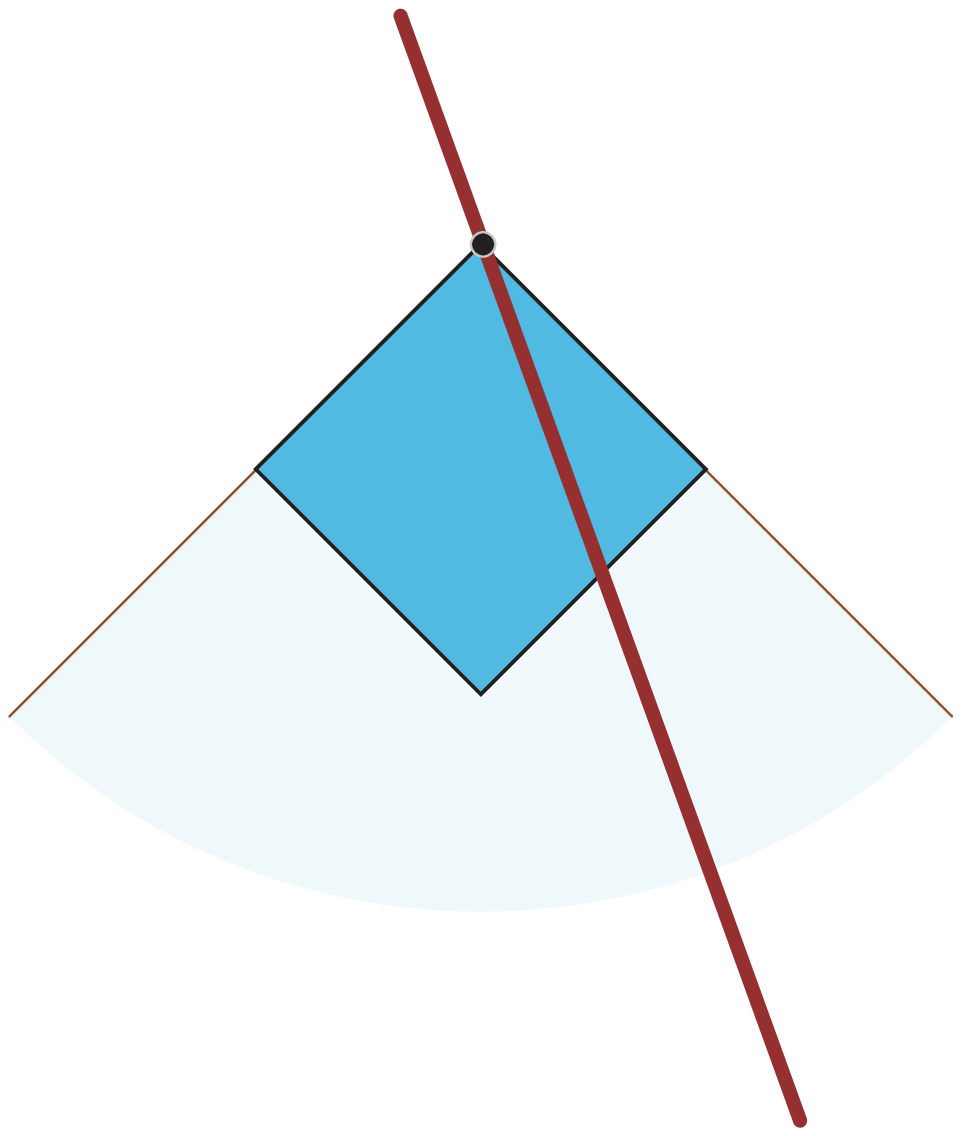}
		  	\put(69,3){\color{my_dark_brown}{$\nullspace(\mA)$}}
			\put(28,25){\color{my_light_blue}{$\fdc$}}
			\put(41,39){\color{my_dark_blue}{$\PP$}}
		\end{overpic}
  	}
	\caption{A geometric interpretation of the optimality condition for success of problem \eqref{eq: problem with multiple prior}, i.e., Lemma \ref{lem: optimality condition}. In both figures, the dark red line denotes the null space of $\mA$, the light blue region denotes the prior restricted cone of problem \eqref{eq: problem with multiple prior} (i.e., $\fdc$), and the dark blue region denotes the prior restricted set of problem \eqref{eq: problem with multiple prior} (i.e., $\PP$). In figure (a), the intersection of $\nullspace(\mA)$ and $\fdc$ contains only the origin. In this case, problem \eqref{eq: problem with multiple prior} succeeds. In figure (b), the intersection of $\nullspace(\mA)$ and $\fdc$ contains a ray. In this case, problem \eqref{eq: problem with multiple prior} fails.}
	\label{fig: noise-free}
\end{figure*}

Using Lemma \ref{lem: optimality condition}, we can study the phase transition of problem \eqref{eq: problem with multiple prior}. For this purpose, we assume that we have random sensing matrix. In particular, we assume that $\mA$ is drawn at random from the standard normal distribution on $\R^{m \times n}$. According to Lemma \ref{lem: optimality condition}, to study the phase transition of problem \eqref{eq: problem with multiple prior}, it is sufficient to answer the following questions:

\begin{itemize}[leftmargin=1em]
  	\item Under what conditions the kernel of $\mA$ intersects the cone $\fdc$ trivially with high probability?
	\item Under what conditions the kernel of $\mA$ intersects the cone $\fdc$ nontrivially with high probability?
\end{itemize}

This questions have been well studied in recent years. We borrow the answer from \cite{AMEL2014}:

\begin{proposition} [\cite{AMEL2014}, Theorem I] \label{prop: probability of intersection}
  	Fix a tolerance $\zeta$. Suppose the matrix $\mA \in \R^{m \times n}$ has independent standard normal entries, and $K$ denotes a convex cone. Then when
	$$
	m \le \delta (K) - a_{\zeta} \sqrt{n},
	$$
	we have $\nullspace(\mA) \cap K = \{ \vO \}$ with probability less than $\zeta$. On the contrary, when
	$$
	m \ge \delta (K) + a_{\zeta} \sqrt{n},
	$$
	we have $\nullspace(\mA) \cap K = \{ \vO \}$ with probability at least $1-\zeta$. The quantity $a_{\zeta} \coloneqq \sqrt{8\log (4/\zeta)}$.
\end{proposition}

Proposition \ref{prop: probability of intersection} is a direct consequence of \cite[Theorem I]{AMEL2014}. The proof involves the theory of conic integral geometry. See reference \cite{AMEL2014} for details. Now combining Lemma \ref{lem: optimality condition} and Proposition \ref{prop: probability of intersection}, we obtain our main results about the phase transition:

\begin{theorem} [Phase transition of convex programs with multiple prior constraints] \label{th: main theorem}
  	Consider convex problem \eqref{eq: problem with multiple prior} to solve the linear inverse problem. If the sensing matrix $\mA$ has independent standard normal entries, the phase transition of problem \eqref{eq: problem with multiple prior} occurs at the statistical dimension of its prior restricted cone. More precisely, for any $\zeta > 0$, when the measurement number $m$ satisfies 
$$
m \le \delta(\fdc) - a_{\zeta}\sqrt{n},
$$
problem \eqref{eq: problem with multiple prior} fails with probability at least $1-\zeta$. On the contrary, when
$$
m \ge \delta(\fdc) + a_{\zeta}\sqrt{n},
$$
problem \eqref{eq: problem with multiple prior} succeeds with probability at least $1-\zeta$. The quantity $a_{\zeta} \coloneqq \sqrt{8\log (4/\zeta)}$.
\end{theorem}


Since the phase transition occurs at the statistical dimension of the prior restricted cone, intuitively, we can see it as a measure of how much we know about the true signal from the prior information.

\begin{remark}
  	We can apply our Theorem \ref{th: main theorem} to analyze the phase transition of problem \eqref{eq: problem no addi_prior}. The prior restricted cone of problem \eqref{eq: problem no addi_prior} is exactly $D(f_0, \fxs)$, the descent cone of $f_0$ at $\fxs$. Thus, in this case, our Theorem \ref{th: main theorem} can be read as: The phase transition of problem \eqref{eq: problem no addi_prior} occurs at the statistical dimension of $D(f_0, \fxs)$. This coincides with the results in \cite[Theorem II]{AMEL2014}.
\end{remark}

\subsection{Statistical Dimension of Prior Restricted Cones: Part I} \label{subsec: calc_statistical_prior_1}

In theory, Theorem \ref{th: main theorem} have revealed that the phase transition of problem \eqref{eq: problem with multiple prior} occurs at the statistical dimension of its prior restricted cone. However, if we want to apply these results in practice, we must find ways to compute the statistical dimension efficiently. For this purpose, in this subsection, we present a recipe that provides a reliable estimate for the statistical dimension of the prior restricted cone, when all the subdifferentials are compact and do not contain the origin. The idea is inspired by the recipe proposed by Amelunxen \textit{et al.} \cite[pp. 244-248]{AMEL2014} for the computation of the statistical dimension of a descent cone.

The basic idea for the recipe is that the statistical dimension of the prior restricted cone can be expressed in terms of its polar, which has a close relation with the normal cones, and furthermore, the subdifferentials, of the functions $f_i$'s, $0 \le i \le k$. Let us first express the statistical dimension in terms of normal cones.

\begin{lemma} \label{lem: calc_statistical_dimension}
  	Consider problem \eqref{eq: problem with multiple prior} to recover the true signal $\fxs$. Let $D(f_i, \fxs),N(f_i, \fxs)$ denote the descent cone and normal cone of $f_i$ at $\fxs$ for $0 \le i \le k$, respectively. Suppose that 
	$$
	\ri\big(D(f_0, \fxs) \big) \cap \ri\big(D(f_1, \fxs)\big) \cap \dots \cap \ri\big(D(f_k, \fxs)\big) \neq \emptyset.
	$$
	Then the polar of the prior restricted cone can be expressed as follows:
	$$
	\fdc^{\circ} = \sum_{i=0}^k N(f_i, \fxs),
	$$
	and the statistical dimension of the prior restricted cone can be expressed as follows:
	$$
	\delta(\fdc) = \E \dist^2 \big(\vg, \sum_{i=0}^k N(f_i, \fxs)\big).
	$$
\end{lemma}

\begin{proof}
  	See Appendix \ref{subsec: statistical dimension in terms of normal cones}.
\end{proof}


Lemma \ref{lem: calc_statistical_dimension} establishes connections between the prior restricted cone and the individual normal cones, but it does not allow us to compute the statistical dimension of the prior restricted cone efficiently. This can be done by incorporating the subdifferential expression for normal cones.

\begin{theorem} [The statistical dimension of the prior restricted cone] \label{th: calc_statistical_dimension_1}
  	Let $\fdc$ be the prior restricted cone of problem \eqref{eq: problem with multiple prior}, and let $\fxs \in \R^n$ be the true signal. Assume that for any $0 \le i \le k$, the subdifferential $\partial f_i(\fxs)$ is non-empty, compact, and does not contain the origin. Assume that the descent cones satisfy
	$$
	\ri\big(D(f_0, \fxs) \big) \cap \ri\big(D(f_1, \fxs)\big) \cap \dots \cap \ri\big(D(f_k, \fxs)\big) \neq \emptyset.
	$$
	Define the function $J: \R^{k + 1}_{+} \rightarrow \R$ to be
	$$
	J(\vtau) \coloneqq \E \Big[ \dist^2 \Big(\vg, \sum_{i=0}^{k} \vtau_i \cdot \partial f_i(\fxs)\Big)\Big],
	$$
	where $\vg \sim N(\vO, \mI_{n})$. Then the statistical dimension of the prior restricted cone has the following upper bound:
	$$
	\delta\big(\fdc\big) \le \inf_{\vtau \in \R^{k + 1}_{+}} J(\vtau).
	$$
	The function $J(\vtau)$ is convex, continuous, and continuously differentiable in $\R^{k+1}_{+}$. It attains its minimum in a compact subset of $\R^{k+1}_+$. Moreover, suppose that
	$$
	\textnormal{the two sets }\sum_{i = 0}^{k} \vtau_i \cdot \partial f_i(\fxs) \textnormal{ and } \sum_{i = 0}^{k} \tilde{\vtau}_i \cdot \partial f_i(\fxs) \textnormal{ are not identical, for any }\vtau \neq \tilde{\vtau} \in \R^{k+1}_+.
	$$
	Then the function $J(\vtau)$ is strictly convex, and attains its minimum at a unique point. For the differential of $J$ at the boundary of $\R^{k+1}_+$, we interpret the partial derivative $\frac{\partial J}{\partial \vtau_i}$ similarly as the right derivative, if $\vtau_i = 0$. 
\end{theorem}

\begin{proof}
  	Since the subdifferential is non-empty, compact, and does not contain the origin, the normal cones is the cone generated by the subdifferential \cite[Corollary 23.7.1]{ROCK1970}. Thus, by Lemma \ref{lem: calc_statistical_dimension},
	\begin{align*}
		\delta(\fdc) & = \E \big[ \dist^2 (\vg, \sum_{i=0}^{k} N(f_i,\fxs))\big] = \E \Big[ \dist^2 \Big(\vg, \sum_{i=0}^{k} \Big( \bigcup_{\vtau_i \ge 0} \vtau_i \cdot \partial f_i(\fxs) \Big)\Big)\Big] \\
		& = \E \Big[ \dist^2 \Big(\vg, \bigcup_{\vtau \in \R_{+}^{k+1}}\Big( \sum_{i=0}^{k} \vtau_i \cdot \partial f_i(\fxs) \Big)\Big)\Big] = \E \Big[ \inf_{\vtau \in \R_{+}^{k+1}} \dist^2 \Big(\vg, \sum_{i=0}^{k} \vtau_i \cdot \partial f_i(\fxs)\Big)\Big] \\
		& \le \inf_{\vtau \in \R_{+}^{k+1}} \E \Big[ \dist^2 \Big(\vg, \sum_{i=0}^{k} \vtau_i \cdot \partial f_i(\fxs)\Big)\Big].
	\end{align*}
	The inequality results from Jensen's inequality. The proof of properties of $J$ appears in Appendix \ref{subsec: distance to sum of subdifferentials} and Appendix \ref{subsec: expected distance to sum of sets}.
\end{proof}

\begin{algorithm}[t]
	\caption{The statistical dimension of the prior feasible descent cone}
  	\label{rec: calc_statistical_dimension_1}
	{\bf Assume} that for $0 \le i \le k$, the function $f_i : \R^n \rightarrow \overline{\R}$ is a proper convex function and $\fxs \in \R^n$. \\
	{\bf Assume} that the intersection of interiors of descent cones are non-empty, i.e., $\ri\big(D(f_0, \fxs) \big) \cap \ri\big(D(f_1, \fxs)\big) \cap \dots \cap \ri\big(D(f_k, \fxs)\big) \neq \emptyset$. \\
	{\bf Assume} that for $0 \le i \le k$, the subdifferential $\partial f_i(\fxs)$ is non-empty, compact, and does not contain the origin. \\
	{\bf Assume} that for any $\vtau \neq \tilde{\vtau} \in \R^{k+1}_+$, the two sets $\sum_{i = 0}^{k} \vtau_i \cdot \partial f_i(\fxs)$ and $\sum_{i = 0}^{k} \tilde{\vtau}_i \cdot \partial f_i(\fxs)$ are not identical.
	\begin{algorithmic}[1]
	  	\State Identify the subdifferential $S_i = \partial f_i(\fxs)$, for $0 \le i \le k$.
		\State For any $\vtau \in \R^{q+1}_+$, find the following set, which is the Minkowski sum of the subdifferentials: $S(\vtau) = \sum_{i = 0}^k \vtau_i \cdot S_i$.
		\State Compute the function $J(\vtau) \coloneqq \E \dist^2 \big(\vg, S(\vtau) \big)$, where $\vg \sim N(\vO, \mI_n)$.
		\State Compute the differential, $\nabla J(\vtau)$, of function $J(\vtau)$.
		\State Find the minimizer $\vtau^{\star}$ of $J(\vtau)$ over $\R^{k+1}_+$, using its differential $\nabla J(\vtau)$.
		\State The statistical dimension of the prior restricted cone has the upper bound $\delta( \fdc ) \le J(\vtau^{\star})$.
	\end{algorithmic}
\end{algorithm}

Theorem \ref{th: calc_statistical_dimension_1} provides an effective way to estimate the statistical dimension of the prior restricted cone, when all the subdifferentials are non-empty, compact, and does not contain the origin. We summarize it in Recipe \ref{rec: calc_statistical_dimension_1}. In subsection \ref{subsec: linear inverse with bounded l2 norm}, we apply Recipe \ref{rec: calc_statistical_dimension_1} to study the phase transition of linear inverse problems with $\ell_2$ norm constraints.

\begin{remark} \label{re: remark for calc}
  	In \cite{AMEL2014}, Amelunxen \textit{et al.} studied the phase transition of problem \eqref{eq: problem no addi_prior}. They proved that the phase transition occurs at the statistical dimension of the descent cone $D(f_0,\fxs)$, and provided a recipe to compute it. Our Recipe \ref{rec: calc_statistical_dimension_1} can be seen as a generalization of this recipe from one function to multiple functions, and our proof idea for Theorem \ref{th: calc_statistical_dimension_1} is inspired by the proof for \cite[Proposition 4.1]{AMEL2014}.
\end{remark}


\subsection{Statistical Dimension of Prior Restricted Cones: Part II} \label{subsec: calc_statistical_prior_2}

Recipe \ref{rec: calc_statistical_dimension_1} gives a reliable estimate of the statistical dimension of the prior restricted cone, when all the subdifferentials are compact and do not contain the origin. However, in many practical applications, we may encounter the case that some of the subdifferentials are unbounded or contain the origin. For example, consider the linear inverse problem with non-negativity constraints, i.e., 
\begin{equation} \label{eq: linear inverse nonnegative}
	\min f_0(\fx), \quad \st \  \vy = \mA \fx, \ \fx \ge \vO.
\end{equation}
Note that $\fx \ge \vO$ is equivalent to $I_{\R^{n}_+}(\fx) \le I_{\R^{n}_+}(\fxs)$. The subdifferential of indicator functions has specific formula \eqref{eq: subdifferential of indicator function}. It is easy to verify that the subdifferential of $I_{\R^{n}_+}$ at $\fxs$ contains the origin, and if $\fxs$ contains zero entries, it is unbounded. Therefore, Recipe \ref{rec: calc_statistical_dimension_1} cannot be used directly, and we have to find other ways to compute the statistical dimension. Actually, an effective way in this case is to express the normal cones via the subdifferentials, only for those functions whose subdifferentials are compact and do not contain the origin.

\begin{theorem} [The statistical dimension of the prior restricted cone] \label{th: calc_statistical_dimension_2}
  	Let $\fdc$ be the prior restricted cone of problem \eqref{eq: problem with multiple prior}, and let $\fxs \in \R^n$ be the true signal. Assume that for $0 \le i \le q$, the subdifferentials $\partial f_i(\fxs)$'s are non-empty, compact, and do not contain the origin, and for $q < i \le k$, the subdifferentials $\partial f_i(\fxs)$'s are non-empty, where $0 \le q < k$ is a natural number. Assume that the descent cones satisfy
	$$
	\ri\big(D(f_0, \fxs) \big) \cap \ri\big(D(f_1, \fxs)\big) \cap \dots \cap \ri\big(D(f_k, \fxs)\big) \neq \emptyset.
	$$
	Assume that for any $\vtau \in \S^{q} \cap \R^{q+1}_{+}$, we have
	\begin{equation*} \label{eq: condition for J_2}
		\vO \notin \closure \big( \sum_{i = q + 1}^{k} N(f_i,\fxs) \big) + \sum_{i = 0}^{q} \vtau_i \cdot \partial f_i(\fxs),
  	\end{equation*}
	Define the function $J: \R^{q + 1}_{+} \rightarrow \R$ by
	$$
	J(\vtau) \coloneqq \E \Big[ \dist^2 \Big(\vg, \sum_{i=0}^{q} \vtau_i \cdot \partial f_i(\fxs) + \sum_{i = q + 1}^{k} N(f_i,\fxs) \Big)\Big],
	$$
	where $\vg \sim N(\vO, \mI_{n})$. Then the statistical dimension of the prior restricted cone has the following upper bound:
	$$
	\delta(\fdc) \le \inf_{\vtau \in \R^{q + 1}_{+}} J(\vtau).
	$$
	The function $J(\vtau)$ is convex, continuous, and continuously differential in $\R^{q+1}_{+}$. It attains its minimum in a compact subset of $\R^{q+1}_+$. Moreover, suppose that for any $\vtau \neq \tilde{\vtau} \in \R^{q+1}_+$, 
	$$
	\textnormal{the two sets} \ \closure\Big( \sum_{i = q + 1}^{k} N(f_i,\fxs) \Big) + \sum_{i = 0}^{q} \vtau_i \cdot \partial f_i(\fxs)  \ \textnormal{and} \ \closure\Big( \sum_{i = q + 1}^{k} N(f_i,\fxs) \Big) + \sum_{i = 0}^{q} \tilde{\vtau}_i \cdot \partial f_i(\fxs) \ \textnormal{are not identical}.
	$$
	Then the function $J(\vtau)$ is strictly convex, and attains its minimum at a unique point. For the differential of $J$ at the boundary of $\R^{q+1}_+$, we interpret the partial derivative $\frac{\partial J}{\partial \vtau_i}$ similarly as the right derivative, if $\vtau_i = 0$.
\end{theorem}

\begin{proof} 
  	We proceed similarly as in Theorem \ref{th: calc_statistical_dimension_1},
	\begin{align*}
	  	\delta(\fdc) & = \E \big[ \dist^2 (\vg, \sum_{i=0}^{k} N(f_i,\fxs))\big] = \E \Big[ \dist^2 \Big(\vg, \sum_{i=0}^{q} \Big( \bigcup_{\vtau_i \ge 0} \vtau_i \cdot \partial f_i(\fxs) \Big) + \sum_{i = q + 1}^{k} N(f_i,\fxs) \Big)\Big] \\
		& = \E \Big[ \dist^2 \Big(\vg, \bigcup_{\vtau \in \R_{+}^{q+1}}\Big( \sum_{i=0}^{q} \vtau_i \cdot \partial f_i(\fxs) + \sum_{i = q + 1}^{k} N(f_i,\fxs) \Big)\Big)\Big] \\
		& = \E \Big[ \inf_{\vtau \in \R_{+}^{q+1}} \dist^2 \Big(\vg, \sum_{i=0}^{q} \vtau_i \cdot \partial f_i(\fxs) + \sum_{i = q + 1}^{k} N(f_i,\fxs) \Big)\Big] \\
		& \le \inf_{\vtau \in \R_{+}^{q+1}} \E \Big[ \dist^2 \Big(\vg, \sum_{i=0}^{q} \vtau_i \cdot \partial f_i(\fxs) + \sum_{i = q + 1}^{k} N(f_i,\fxs) \Big)\Big].
	\end{align*}
	The inequality results from Jensen's inequality. The proof of properties of $J$ appears in Appendix \ref{subsec: dist to sum of unbouded sets} and Appendix \ref{subsec: expected distance to sum of unbounded sets}.
\end{proof}

Theorem \ref{th: calc_statistical_dimension_2} further generalizes our Theorem \ref{th: calc_statistical_dimension_1} and \cite[Proposition 4.1]{AMEL2014} to the case when some of the subdifferentials are unbounded or contain the origin, and the proofs share similar ideas. We summarize it in Recipe \ref{rec: calc_statistical_dimension_2}. In subsection \ref{subsec: linear inverse with non-negativity}, we apply Recipe \ref{rec: calc_statistical_dimension_2} to study the phase transition of linear inverse problems with non-negativity constraints.

\begin{algorithm}[t]
	\caption{The statistical dimension of the prior feasible descent cone}
  	\label{rec: calc_statistical_dimension_2}
	{\bf Assume} that for $0 \le i \le k$, the function $f_i : \R^n \rightarrow \overline{\R}$ is a proper convex function and $\fxs \in \R^n$. \\
	{\bf Assume} that the intersection of interiors of descent cones are non-empty, i.e., $\ri\big(D(f_0, \fxs) \big) \cap \ri\big(D(f_1, \fxs)\big) \cap \dots \cap \ri\big(D(f_k, \fxs)\big) \neq \emptyset$. \\
	{\bf Assume} that for $0 \le i \le q$, the subdifferential $\partial f_i(\fxs)$ is non-empty, compact, and does not contain the origin. \\
	{\bf Assume} that for $q+1 \le i \le k$, the subdifferential $\partial f_i(\fxs)$ is non-empty. \\
	{\bf Assume} that for any $\vtau \in \S^{q} \cap \R^{q+1}_{+}$, the set $\closure \big( \sum_{i = q + 1}^{k} N(f_i,\fxs) \big) + \sum_{i = 0}^{q} \vtau_i \cdot \partial f_i(\fxs)$ does not contain the origin. \\
	{\bf Assume} that for any $\vtau \neq \tilde{\vtau} \in \R^{q+1}_{+}$, the two sets 
	$$
	\closure \big( \sum_{i = q + 1}^{k} N(f_i,\fxs) \big) + \sum_{i = 0}^{q} \vtau_i \cdot \partial f_i(\fxs)  \ \textnormal{and} \ \closure \big( \sum_{i = q + 1}^{k} N(f_i,\fxs) \big) + \sum_{i = 0}^{q} \tilde{\vtau}_i \cdot \partial f_i(\fxs)
	$$
	are not identical.
	\begin{algorithmic}[1]
	  	\State Identify the subdifferential $S_i = \partial f_i(\fxs)$, for $0 \le i \le q$, and the normal cone $N_i = N(f_i, \fxs)$, for $q+1 \le i \le k$.
		\State For any $\vtau \in \R^{q+1}_+$, find the following set, which is the Minkowski sum of the subdifferentials and the normal cones: $S(\vtau) = \sum_{i = 0}^q \vtau_i S_i + \sum_{i = q+1}^k N_i$.
		\State Compute the function $J(\vtau) \coloneqq \E \dist^2 \big(\vg, S(\vtau) \big)$, where $\vg \sim N(\vO, \mI_n)$.
		\State Compute the differential, $\nabla J(\vtau)$, of function $J(\vtau)$.
		\State Find the minimizer $\vtau^{\star}$ of $J(\vtau)$ over $\R^{k+1}_+$, using its differential $\nabla J(\vtau)$.
		\State The statistical dimension of the prior restricted cone has the upper bound $\delta ( \fdc ) \le J(\vtau^{\star})$.
	\end{algorithmic}
\end{algorithm}

\section{Examples and Applications} \label{sec: applications}

In this section, we make use of our theoretical results and the computation recipes to study the phase transition of several specific problems. In subsection \ref{subsec: linear inverse with bounded l2 norm}, we study the phase transition of linear inverse problems with $\ell_2$ norm constraints, and in subsection \ref{subsec: linear inverse with non-negativity}, we study the phase transition of linear inverse problems with non-negativity constraints.

\subsection{Phase Transition of Linear Inverse Problem with $\ell_2$ Norm Constraints} \label{subsec: linear inverse with bounded l2 norm}

In this subsection, we make use of Recipe \ref{rec: calc_statistical_dimension_1} to study the phase transition of linear inverse problems with $\ell_2$ norm constraints. In other words, we study the phase transition of the following convex problem:
\begin{equation} \label{eq: linear inverse bounded norm}
		\min f_0(\fx), \quad \st \  \vy = \mA \fx, \ \|\fx\|_2 \le \|\fxs\|_2.
\end{equation}
Note that for $\fxs \neq \vO$, the subdifferential of $\ell_2$ norm is $\partial \|\fxs\|_2 = \big\{ \frac{\fxs}{\|\fxs\|_2}\big\}$. Therefore, applying Recipe \ref{rec: calc_statistical_dimension_1} directly, we obtain the following results:

\begin{corollary} \label{coro: calc_bounded_norm}
Let $\fdc_1$ be the prior restricted cone of problem \eqref{eq: linear inverse bounded norm}, and let $\fxs \in \R^n$ be the true signal. Assume that the subdifferential $\partial f_0(\fxs)$ is non-empty, compact, do not contain the origin. Assume that the descent cones satisfy
	$$
	\ri\big(D(f_0, \fxs)\big) \cap \ri \big(D(\|\cdot\|_2, \fxs)\big) \neq \emptyset.
	$$
	Define the function $J_1: \R^{2}_{+} \rightarrow \R$ to be
	$$
  	J_1(\vtau) = \E \dist^2 \Big[ \big( \vg, \vtau_0 \cdot \partial f_0(\fxs) + \vtau_1 \cdot \frac{\fxs}{\|\fxs\|_2} \big) \Big],
	$$
	where $\vg \sim N(\vO, \mI_{n})$. Then the statistical dimension of the prior restricted cone of problem \eqref{eq: linear inverse bounded norm} has the following upper bound:
	$$
	\delta(\fdc_1) \le \inf_{\vtau \in \R^{2}_{+}} J_1(\vtau).
	$$
	The function $J_1(\vtau)$ is convex, continuous, and continuously differentiable in $\R^{2}_{+}$. It attains its minimum in a compact subset of $\R^2_+$. Moreover, suppose that
	\begin{equation*} \label{eq: sets not equal-bounded}
	  \textnormal{the two sets }\vtau_0 \cdot \partial f_0(\fxs) + \vtau_1 \cdot \frac{\fxs}{\|\fxs\|_2} \textnormal{ and }\tilde{\vtau}_0 \cdot \partial f_0(\fxs) + \tilde{\vtau}_1 \cdot \frac{\fxs}{\|\fxs\|_2} \textnormal{ are not identical for any }\vtau \neq \tilde{\vtau} \in \R^2_+.
  	\end{equation*}
	Then the function $J_1(\vtau)$ is strictly convex, and attains its minimum at a unique point. For the differential of $J_1$ at the boundary of $\R^{2}_+$, we interpret the partial derivative $\frac{\partial J_1}{\partial \vtau_i}$ similarly as the right derivative, if $\vtau_i = 0$. 
\end{corollary}

\begin{proof}
  	Applying Theorem \ref{th: calc_statistical_dimension_1} to problem \eqref{eq: linear inverse bounded norm} directly, we obtain Corollary \ref{coro: calc_bounded_norm}.
\end{proof}

Corollary \ref{coro: calc_bounded_norm} implies that to study the phase transition of problem \eqref{eq: linear inverse bounded norm}, we need to find the infimum of $J_1$. When $f_0$ is a general proper convex function, the infimum may be attained anywhere in $\R^2_+$. However, when $f_0$ is a norm, an important result is that the infimum of $J_1$ must be attained in $\R_+ \times \{0\}$.

\begin{proposition} \label{prop: minimizer of J in bounded case}
  Consider problem \eqref{eq: linear inverse bounded norm}. Assume that $f_0$ is a norm, and the conditions in Corollary \ref{coro: calc_bounded_norm} hold. Define the function $J_2: \R_+ \rightarrow \R$ as
	$$
	J_2(\tau) = \E \dist^2\big(\vg, \tau \cdot \partial f_0(\fxs) \big) \quad \textnormal{for } \tau \ge 0.
	$$
	The function $J_2(\tau)$ is strictly convex, continuously differentiable in $\R_+$, and attains its minimum at a unique point. Moreover, the unique minimizer $\vtau^{\star} = (\vtau^{\star}_0, \vtau^{\star}_1)$ of $J_1$ satisfies $\vtau^{\star}_1 = 0$ and $\vtau^{\star}_0$ is the unique minimizer of $J_2$, and the minimum of $J_1$ over $\R^2_+$ and that of $J_2$ over $\R_+$ are equal: 
	$$
	\inf_{\vtau \in \R^2_+} J_1(\vtau) = J_1(\vtau^{\star}) = J_2(\vtau^{\star}_0) = \inf_{\tau \in \R_+} J_2(\tau).
	$$
\end{proposition}

\begin{proof}
  	The first part of this proposition, i.e., the properties of $J_2(\tau)$, has been proved in \cite{AMEL2014}. For the proof of the results about $\vtau^{\star}$, please see Appendix \ref{subsec: minimizer of J in bounded case}.
\end{proof}

\begin{remark}
  Let $\fdc_2$ denote the prior restricted cone of problem \eqref{eq: problem no addi_prior}, i.e., $\fdc_2 = D(f_0, \fxs)$. In \cite{AMEL2014}, Amelunxen \textit{et al.} have proved that $\inf_{\tau \in \R_+} J_2(\tau)$ is a reliable estimate of $\delta(\fdc_2)$. Therefore, Proposition \ref{prop: minimizer of J in bounded case} implies that when we use Recipe \ref{rec: calc_statistical_dimension_1} to compute the statistical dimension of the prior restricted cone of problem \eqref{eq: linear inverse bounded norm}, the obtained phase transition point is exactly the same as that of problem \eqref{eq: problem no addi_prior}.
\end{remark}

At the first sight, the above results may be surprising, since we have more prior information, but the obtained phase transition point is the same. From another point of view, this implies that if our Recipe \ref{rec: calc_statistical_dimension_1} can provide an accurate estimation of the statistical dimension, $\delta(\fdc_1)$ and $\delta(\fdc_2)$ must be nearly equal. Actually, we can verify that this is the case.

\begin{proposition} \label{prop: equivalence of statistical dimensions}
  	Let $\fdc_1$ and $\fdc_2$ denote the prior restricted cones of problem \eqref{eq: linear inverse bounded norm} and problem \eqref{eq: problem no addi_prior}, respectively. Assume that $f_0$ is a norm. Then
	\begin{equation} \label{eq: statistical dimension are equal}
	  	\delta(\fdc_1) \le \delta(\fdc_2) \le \delta(\fdc_1) + \frac{1}{2}.
	\end{equation}
\end{proposition}

\begin{proof}
  See Appendix \ref{subsec: proof of equivalence of statistical dimensions}.
\end{proof}

\begin{remark}
  Proposition \ref{prop: equivalence of statistical dimensions} implies that in the case when $f_0$ is a norm, the additional $\ell_2$ norm constraint $\|\fx\|_2 \le \|\fxs\|_2$ has little effect on the phase transition of linear inverse problem. Moreover, since $\inf_{\tau \in \R_+} J_2(\tau)$ is an accurate estimate of $\delta(\fdc_2)$, it follows that $\inf_{\vtau \in \R^2_+} J_1(\vtau)$ is an accurate estimate of $\delta(\fdc_1)$. 
\end{remark}


Using the above results, we can obtain an error bound for Recipe \ref{rec: calc_statistical_dimension_1} when applied to problem \eqref{eq: linear inverse bounded norm}. 
\begin{proposition} \label{prop: error bound for Recipe 1}
  Consider problem \eqref{eq: linear inverse bounded norm} to recover $\fxs$. Assume that $f_0$ is a norm and denote $\fdc_1$ the prior restricted cone of problem \eqref{eq: linear inverse bounded norm}. Then under the conditions of Corollary \ref{coro: calc_bounded_norm}, we have
  	\begin{equation} \label{eq: error bound for Recipe 1}
	  	0 \le \inf_{\vtau \in \R^2_+} J_1(\vtau) - \delta(\fdc_1) \le \frac{2 \sup \{ \|\vs\|_2: \vs \in \partial f_0(\fxs) \}}{f_0(\fxs/\|\fxs\|_2)} + \frac{1}{2}.
	\end{equation}
\end{proposition}

\begin{proof}
  	This is a direct consequence of \cite[Theorem 4.3]{AMEL2014}, Proposition \ref{prop: minimizer of J in bounded case}, and Proposition \ref{prop: equivalence of statistical dimensions}. We omit the proof.
\end{proof}

Proposition \ref{prop: error bound for Recipe 1} implies that our Recipe \ref{rec: calc_statistical_dimension_1} can provide an accurate estimate of the statistical dimension of the prior restricted cone, when applied to problem \eqref{eq: linear inverse bounded norm}. Next, as a more concrete example, let us study the phase transition of the compressed sensing problem with $\ell_2$ norm constraints:
\begin{equation} \label{eq: compressed sensing with bounded norm}
		\min \|\fx\|_1, \quad \st \  \vy = \mA \fx, \ \|\fx\|_2 \le \|\fxs\|_2.
\end{equation}
We have the following results:

\begin{corollary} \label{coro: calc_CS_bounded_norm}
	Let $\fdc_1$ be the prior restricted cone of problem \eqref{eq: compressed sensing with bounded norm}, and $\fxs \in \R^n$ be the true signal. Assume that $\fxs$ has exactly $s$ non-zero entries. Then the statistical dimension of the prior restricted cone of problem \eqref{eq: compressed sensing with bounded norm} satisfies
	$$
	\psi_1(s/n) - \frac{2}{\sqrt{sn}} - \frac{1}{2n} \le \frac{\delta(\fdc_1)}{n} \le \psi_1(s/n),
	$$
	where the function $\psi_1:[0,1] \rightarrow [0,1]$ is
	\begin{equation} \label{eq: function psi_1}
		\psi_1(\rho) = \inf_{\tau \in \R_+} \Big\{ \rho (1+\tau^2) + (1-\rho)\int_{\tau}^{\infty}(u - \tau)^2 \cdot \varphi(u) \mathrm{d}u \Big\},
  	\end{equation}
	where the function $\varphi(u) = \sqrt{\frac{2}{\pi}} e^{-u^2/2}$. The infimum is attained at the unique $\tau$, which solves the stationary equation
	$$
	\frac{\rho}{1-\rho} = \int_{\tau}^{\infty} \Big( \frac{u}{\tau} - 1 \Big) \cdot \varphi(u) \mathrm{d}u.
	$$
\end{corollary}

\begin{proof}
  	This is a direct consequence of Proposition \ref{prop: minimizer of J in bounded case}, Proposition \ref{prop: equivalence of statistical dimensions}, Proposition \ref{prop: error bound for Recipe 1} and \cite[Proposition 4.1]{AMEL2014}. We omit the proof.
\end{proof}

\subsection{Phase Transition of Linear Inverse Problem with Non-negativity Constraints} \label{subsec: linear inverse with non-negativity}

In this subsection, we make use of Recipe \ref{rec: calc_statistical_dimension_2} to study the phase transition of linear inverse problems with non-negativity constraints, i.e., problem \eqref{eq: linear inverse nonnegative}. We have confirmed that the subdifferential of $I_{\R^{n}_+}$ is unbounded and contains the origin. Therefore, to apply Recipe \ref{rec: calc_statistical_dimension_2}, we have to find the normal cone $N(I_{\R^n_+}, \fxs)$ directly. Indeed, notice that the descent cone of $I_{\R^{n}_+}$ at $\fxs$ is 
$$
D(I_{\R^{n}_+}, \fxs) = \bigcup_{\upsilon \ge 0} \big\{ \upsilon (\fx - \fxs): \fx \in \R^{n}_+ \big\} = \big\{\fx \in \R^n: \fx_i \in \R \ \textnormal{if} \ \fxs_i > 0, \ \fx_i \in [0,\infty) \ \textnormal{if} \ \fxs_i = 0 \big\}.
$$
Thus, the normal cone $N(I_{\R^n_+}, \fxs)$ is
\begin{equation} \label{eq: normal cone of nonegative}
	N \coloneqq N(I_{\R^{n}_+}, \fxs) = \big\{ \fx \in \R^n: \fx_i = 0 \ \textnormal{if} \ \fxs_i > 0, \ \fx_i \in (-\infty, 0] \ \textnormal{if} \ \fxs_i = 0\big\}.
\end{equation}
Applying Theorem \ref{th: calc_statistical_dimension_2}, we obtain the following results.
\begin{corollary} \label{coro: calc linear inverse nonnegative}
  	Consider the convex problem with non-negativity constraint \eqref{eq: linear inverse nonnegative}, and denote $\fdc_3$ its prior restricted cone. Suppose that $\partial f_0(\fxs)$ is non-empty, compact, and does not contain the origin. Assume that the descent cones satisfy
	$$
	\ri\big(D(f_0, \fxs)\big) \cap \ri\big(D(I_{\R^n_+}, \fxs)\big) \neq \emptyset.
	$$
	Assume that $N + \partial f(\fxs)$ does not contain the origin. Define the function $J_3: \R_+ \rightarrow \R$ to be
	$$
	J_3(\tau) = \E \dist^2 \big(\vg, N + \tau \cdot \partial f(\fxs) \big),
	$$
	where $\vg \sim N(\vO, \mI_n)$. Then the statistical dimension of the prior restricted cone of problem \eqref{eq: linear inverse nonnegative} has the following bound:
	\begin{equation} \label{eq: calc linear inverse nonnegative}
	  	\delta (\fdc_3) \le \inf_{\tau \ge 0} J_3(\tau).
	\end{equation}
	The function $J_3(\tau)$ is convex, continuous, and continuously differentiable in $\R_+$. It attains its minimum in a compact subset of $\R_+$. Moreover, suppose that for any $\tau \neq \tilde{\tau} \in \R_+$, the two sets $N+ \tau \cdot \partial f(\fxs)$ and $N+ \tilde{\tau}\cdot \partial f(\fxs)$ are not identical. Then the function $J_3(\tau)$ is strictly convex, continuously differentiable for $\tau \ge 0$, and attains its minimum at a unique point. For the derivative of $J_3$ at the origin, we interpret it as the right derivative.
\end{corollary}

\begin{proof}
  	Corollary \ref{coro: calc linear inverse nonnegative} follows from Theorem \ref{th: calc_statistical_dimension_2} directly.
\end{proof}

In the case of $f_0$ is some norm, we can obtain a reverse bound of \eqref{eq: calc linear inverse nonnegative}:

\begin{proposition} \label{prop: error bound for calc}
  	Let $f_0$ be some norm on $\R^n$, and $\fxs \in \R^n$. Then under the conditions of Corollary \ref{coro: calc linear inverse nonnegative}, if $N + \cone \big(\partial f_0(\fxs) \big)$ is closed, we have the error bound
	$$
	0 \le \inf_{\tau \ge 0} J_3(\tau) - \delta (\fdc_3) \le \frac{2 \sup \{\|\vs\|_2: \vs \in \partial f_0(\fxs)\}}{f_0(\fxs / \|\fxs\|_2)}.
	$$
\end{proposition}

\begin{proof}
  	See Appendix \ref{subsec: error bound for non-negativity}.
\end{proof}


\begin{remark}
  In \cite{AMEL2014}, Amelunxen \textit{et al.} proposed a recipe to compute the statistical dimension of a descent cone, and presented an error bound for their recipe. The error bound in Proposition \ref{prop: error bound for calc} generalizes their ideas from problem \eqref{eq: problem no addi_prior} to problem \eqref{eq: linear inverse nonnegative}.
\end{remark}

As a more concrete example, let us apply Recipe \ref{rec: calc_statistical_dimension_2} to study the phase transition of the $\ell_1$ minimization problem with non-negativity constraints:
\begin{equation} \label{eq: l1_nonnegative}
		\min  \|\fx\|_1, \quad \st \ \vy = \mA \fx, \ \fx \ge \vO.
\end{equation}
We have the following results:
\begin{corollary} \label{coro: bound_l1_nonnegative}
  	Consider problem \eqref{eq: l1_nonnegative}. Assume that $\fxs \in \R^n_+$ has exactly $s$ non-zero entries. Then the statistical dimension of the prior restricted cone $\fdc_3$ of problem \eqref{eq: l1_nonnegative} has the following bounds:
	\begin{equation*} \label{eq: bound_l1_nonnegative}
	  	\psi_2(s/n) - \frac{2}{\sqrt{sn}} \le \frac{\delta(\fdc_3)}{n} \le \psi_2 (s/n).
	\end{equation*}
	The function $\psi_2: [0,1] \rightarrow [0,1]$ is defined to be
	\begin{equation} \label{eq: function psi}
	  \psi_2(\rho) = \inf_{\tau \ge 0} \Big\{ \rho(1+\tau^2) + \frac{1}{2}(1-\rho) \int_{-\infty}^{\tau} (u - \tau)^2 \varphi(u) \mathrm{d}u \Big\},
  	\end{equation}
	where the function $\varphi(u) = \sqrt{\frac{2}{\pi}} e^{-u^2/2}$. Moreover, the infimum in \eqref{eq: function psi} is attained at the unique $\tau$ which solves the stationary equation
	\begin{equation} \label{eq: stationary equation}
	  	\frac{2\rho}{1-\rho} = \int_{\tau}^{\infty} \Big(\frac{u}{\tau}-1\Big)\varphi(u) \mathrm{d}u.
	\end{equation}
\end{corollary}

\begin{proof}
  It is easy to check that the conditions in Corollary \ref{coro: calc linear inverse nonnegative} are satisfied in problem \eqref{eq: l1_nonnegative}. Thus, Corollary \ref{coro: bound_l1_nonnegative} results from a direct application of Corollary \ref{coro: calc linear inverse nonnegative} and Proposition \ref{prop: error bound for calc}. For a detailed proof, please refer to Appendix \ref{app: calc_l1_nonnegative}.
\end{proof}

\begin{remark}
  	In \cite{AMEL2014}, Amelunxen \textit{et al.} demonstrated that the phase transition of the $\ell_1$ minimization problem:
	\begin{equation} \label{eq: l1 minimization}
		\min  \|\fx\|_1, \quad \st \ \vy = \mA \fx
	\end{equation}
	occurs at the statistical dimension of $D(\|\cdot\|_1, \fxs)$, and the statistical dimension has the bound
	\begin{equation*} \label{eq: bound_l1_minimization}
	  \psi_1(s/n) - \frac{2}{\sqrt{sn}} \le \frac{\delta\big(D(\|\cdot\|_1, \fxs)\big)}{n} \le \psi_1 (s/n).
	\end{equation*}
	The function $\psi_1: [0,1] \rightarrow [0,1]$ is defined in \eqref{eq: function psi_1}. It is easy to see that $\psi_2(\rho) \le \psi_1(\rho)$ for any $0 \le \rho \le  1$. This is consistent with the intuition that adding a non-negativity constraint means more prior information, so less measurements are needed. See Fig. \ref{fig: two curves} for a comparison of the curves of $\psi_1(\rho)$ and $\psi_2(\rho)$.
\end{remark}

\begin{figure}
  	\centering
  	\includegraphics[width = .8\textwidth]{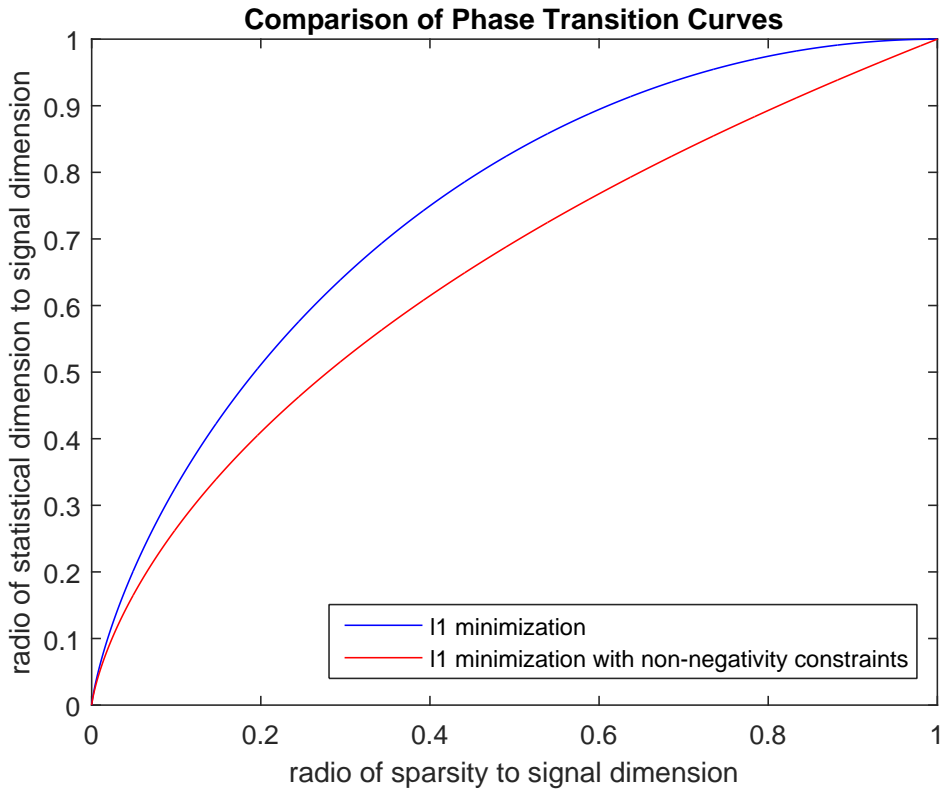}
	\caption{Illustration of the phase transition points for $\ell_1$ minimization problem with non-negativity constraints \eqref{eq: l1_nonnegative} and $\ell_1$ minimization problem \eqref{eq: l1 minimization}. The blue curve is the curve of $\psi_1(\theta)$, which is the phase transition point of problem \eqref{eq: l1 minimization}, and the red curve is the curve of $\psi_2(\theta)$, which is the phase transition point of problem \eqref{eq: l1_nonnegative}.}
	\label{fig: two curves}
\end{figure}

\begin{remark}
  	In \cite{DONOHO2005, DONOHO2009}, Donoho and Tanner studied the $\ell_1$ minimization problem with non-negativity constraints \eqref{eq: l1_nonnegative}. They proved the existence of \textit{weak threshold} and \textit{strong threshold} and showed that at the weak threshold, the probability that problem \eqref{eq: l1_nonnegative} succeeds jumps from $1$ to $1-\epsilon$, where $\epsilon > 0$ is some number. Compared with their results, our results are more precise, i.e., we demonstrate that sharp phase transition exists, and provide an accurate estimate for the phase transition point.
\end{remark}

\section{Simulation results} \label{sec: simulation results}

In this section, we employ several numerical experiments to verify our theoretical results and our computation recipes. In the experiments, we use CVX Matlab package \cite{CVX1} \cite{CVX2} to solve convex programs.

\subsection{Simulation Results for $\ell_1$ Minimization with $\ell_2$ Norm Constraints}

We first design an experiment to verify our results about Recipe \ref{rec: calc_statistical_dimension_1}. More precisely, we design the signal to be sparse and assume that its $\ell_2$ norm is know beforehand, and solve problem \eqref{eq: compressed sensing with bounded norm} to recover the signal. The experiment settings are as follows: We set the ambient dimension $n$ to be $128$. The measurement number $m$ increases from $1$ to $128$ with step $1$, and the sparsity level $s$ of the signal increases from $1$ to $128$ with step $1$ as well. For each pair of selections of $m$ and $s$, we generate the true signal $\fxs$ with $s$ independent standard normal entries and $n-s$ zeros, sample the sensing matrix $\mA$ from the standard normal distribution on $\R^{m \times n}$, and obtain the observation $\vy = \mA \fxs$. Then we run and solve problem \eqref{eq: compressed sensing with bounded norm} $20$ times. We declare success if the solution $\hat{\fx}$ satisfies $\|\hat{\fx} - \fxs\|_2 \le 10^{-4}$. After all these are done, we calculate the empirical probability of successful recovery. At last, we plot the theoretical curve predicted by Corollary \ref{coro: calc_CS_bounded_norm}.

Moreover, Proposition \ref{prop: minimizer of J in bounded case} and Proposition \ref{prop: equivalence of statistical dimensions} imply that the phase transition point of problem \eqref{eq: compressed sensing with bounded norm} and that of \eqref{eq: l1 minimization} are nearly the same. Therefore, as a comparison, we design an experiment to obtain the empirical probability of successful recovery of problem \eqref{eq: l1 minimization}. The experiment settings are absolutely the same as the experiment for problem \eqref{eq: compressed sensing with bounded norm}, except that we solve problem \eqref{eq: l1 minimization} for recovery this time.

The simulation results of problems \eqref{eq: compressed sensing with bounded norm} and \eqref{eq: l1 minimization} are presented in Fig. \ref{fig: simulation results}. Fig. \ref{fig: simulation results}(a) shows that the theoretical threshold, predicted by our Corollary \ref{coro: calc_CS_bounded_norm}, matches the empirical phase transition of problem \eqref{eq: compressed sensing with bounded norm} perfectly. Moreover, comparing Fig. \ref{fig: simulation results}(a) and Fig. \ref{fig: simulation results}(b), we can see that the phase transition points of problem \eqref{eq: compressed sensing with bounded norm} and \eqref{eq: l1 minimization} are almost the same, which verifies our Proposition \ref{prop: minimizer of J in bounded case} and Proposition \ref{prop: equivalence of statistical dimensions}. These results imply that our Recipe \ref{rec: calc_statistical_dimension_1} can provide an accurate estimation of the statistical dimension of the prior restricted cone, when applied to problem \eqref{eq: compressed sensing with bounded norm}.

\begin{figure}
  	\centering
	\subfigure[Phase Transition of Problem \eqref{eq: compressed sensing with bounded norm}]{
	  	\includegraphics[width= .48\textwidth]{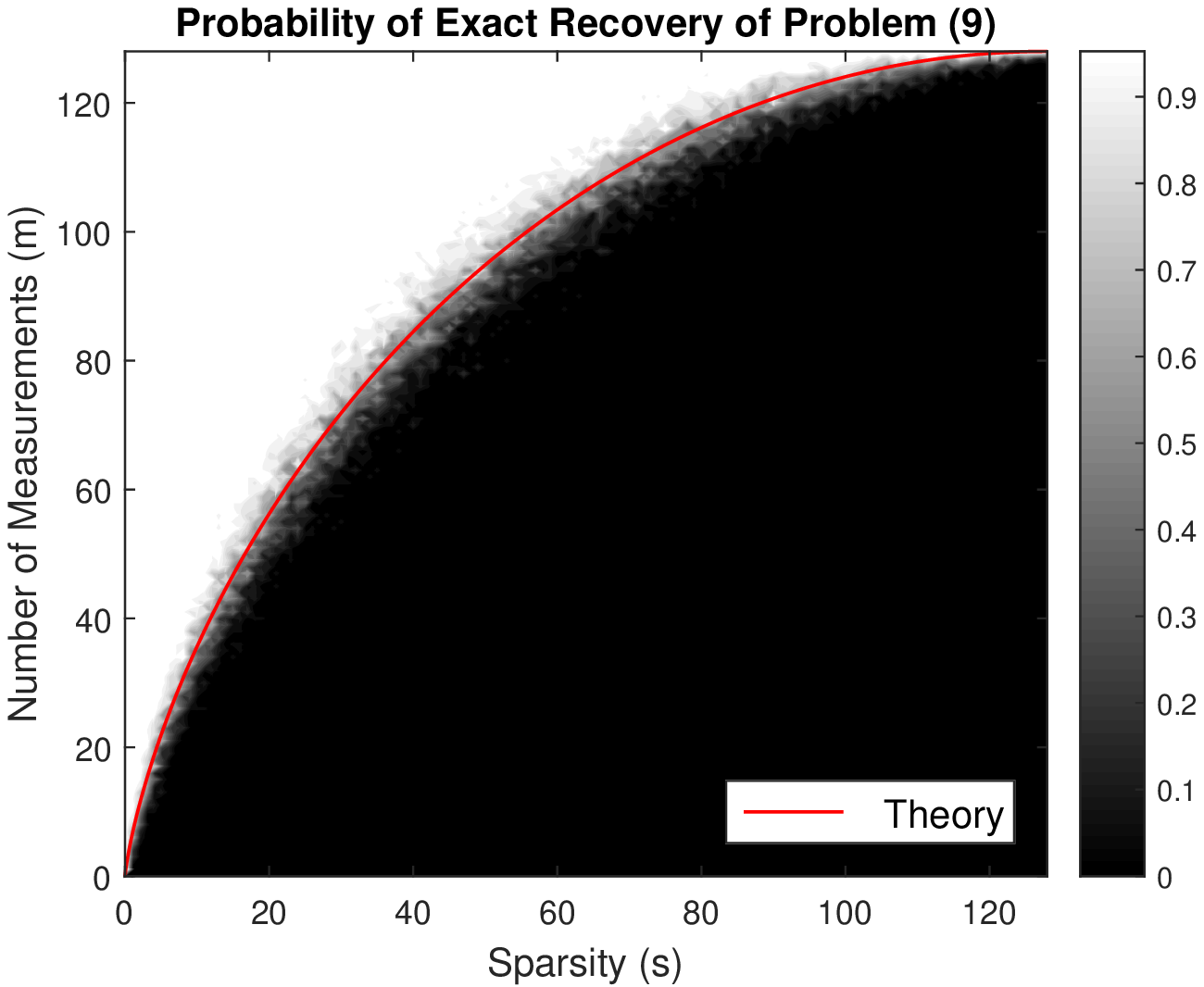}
	}
	\subfigure[Phase Transition of Problem \eqref{eq: l1 minimization}]{
	  	\includegraphics[width= .48\textwidth]{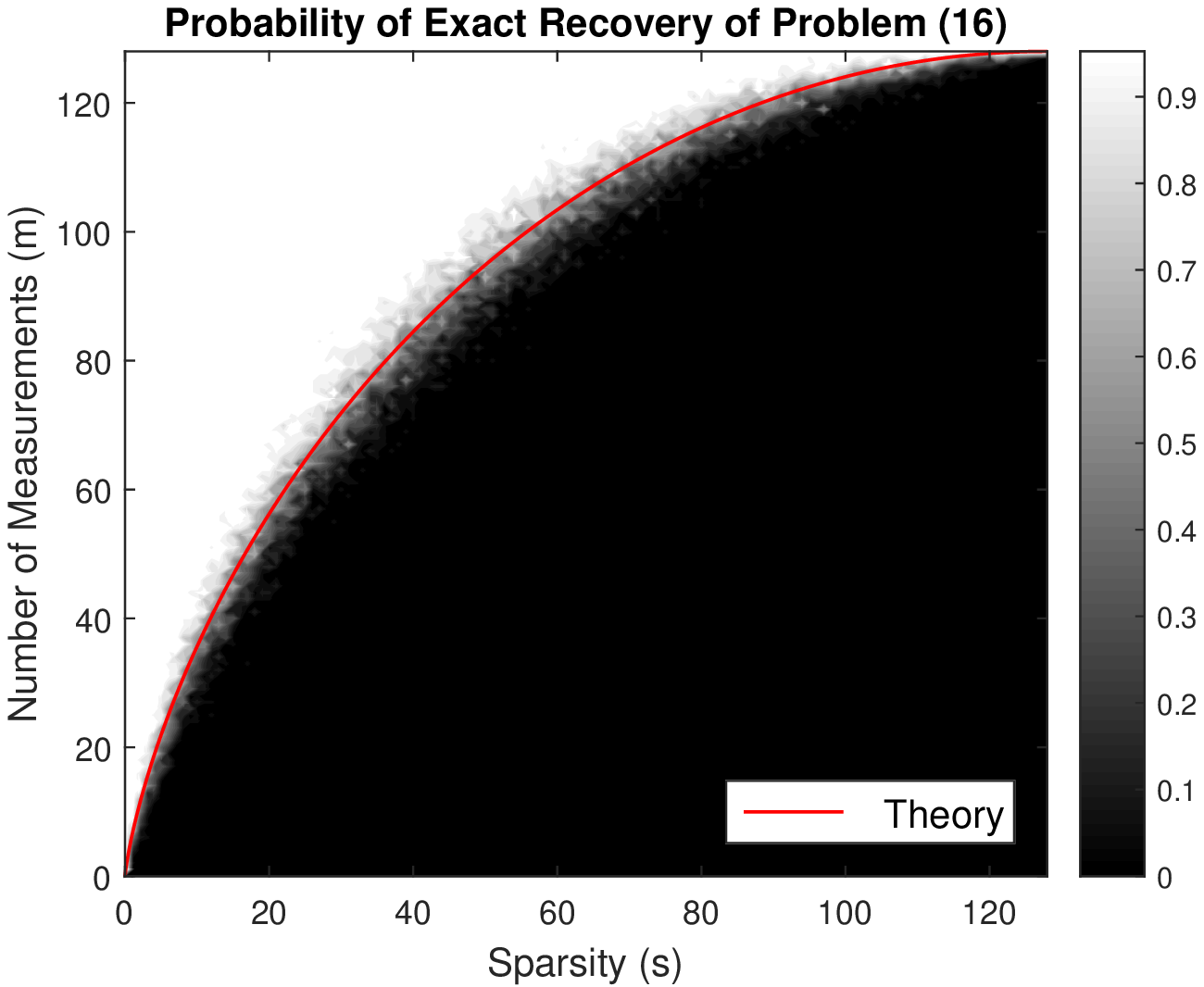}
	}
	\caption{Simulation results. In figures (a), (b), we present the simulation results for the phase transition of problems \eqref{eq: compressed sensing with bounded norm} and \eqref{eq: l1 minimization}, respectively. In both figures, the gray level represents the empirical probability of successful recovery: Black means failure, and white means success. The red curves plot the phase transition points predicted by our theoretical results in Corollary \ref{coro: calc_CS_bounded_norm}, i.e., either of the red curves denotes the curve of $n \cdot \psi_1(s/n)$.}.
	\label{fig: simulation results}
\end{figure}

\subsection{Simulation results for $\ell_1$ minimization with non-negativity constraints}

The second experiment is designed to verify our results about Recipe \ref{rec: calc_statistical_dimension_2}. More precisely, we design the signal to be non-negative and sparse, and solve problem \eqref{eq: l1_nonnegative} to recover the signal. The experiment settings are similar as the previous experiment: The ambient dimension $n$ is setted to be $128$, the measurement number $m$ increases from $1$ to $128$ with step $1$, and the sparsity level $s$ of the signal increases from $1$ to $128$ with step $1$. For each pair of selections of $m$ and $s$, we repeat the following process $20$ times. We generate a sparse vector $\tilde{\fx}$ with $s$ independent standard normal entries and $n-s$ zeros, make the true signal $\fxs_i = |\tilde{\fx}_i|$ for all $1 \le i \le n$, sample the sensing matrix $\mA$ from the standard normal distribution on $\R^{m \times n}$, and obtain the observation $\vy = \mA \fxs$. Then we run and solve problem \eqref{eq: l1_nonnegative}. We declare success if the solution $\hat{\fx}$ to problem \eqref{eq: l1_nonnegative} satisfies $\|\hat{\fx} - \fxs\|_2 \le 10^{-4}$. After all these are done, we calculate the empirical probability of successful recovery. At last, we plot the theoretical curve predicted by Corollary \ref{coro: bound_l1_nonnegative}.


Fig. \ref{fig: simulation results_2} reports our simulation results. It reflects that our theoretical phase transition curve, given by Corollary \ref{coro: bound_l1_nonnegative}, can predict the empirical phase transition of problem \eqref{eq: l1_nonnegative} accurately. This implies that our Recipe \ref{rec: calc_statistical_dimension_2} can provide a reliable estimate of the statistical dimension of the prior restricted cone, when applied to problem \eqref{eq: l1_nonnegative}.

\begin{figure}
  	\centering
	  	\includegraphics[width= .8\textwidth]{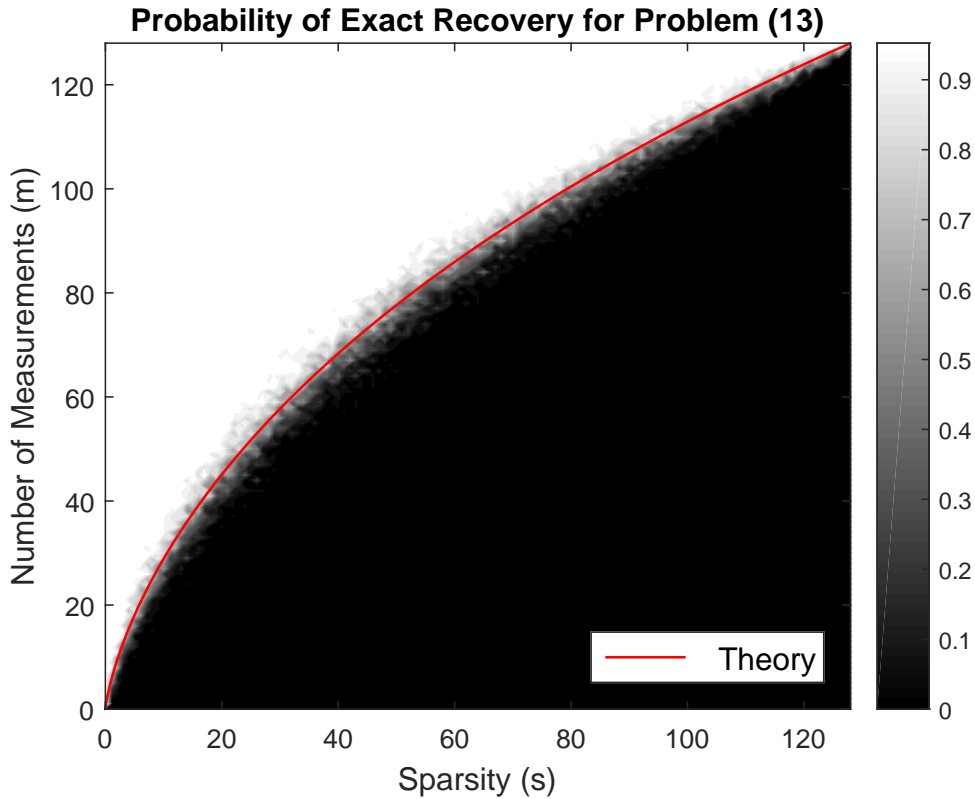}
	\caption{Simulation results for the phase transition of problem \eqref{eq: l1_nonnegative}. In this figure, the gray level represents the empirical probability of successful recovery: Black means failure, and white means success. The red curve plots the phase transition threshold predicted by Corollary \ref{coro: bound_l1_nonnegative}, i.e., $n\cdot \psi_2(s/n)$.}
	\label{fig: simulation results_2}
\end{figure}

\section{Conclusion} \label{sec: conclusion}
This paper studied the phase transition of convex programs with multiple prior constraints, to solve the linear inverse problem. Given such a convex program, we defined its prior restricted set and prior restricted cone, and proved that the phase transition occurs at the statistical dimension of the prior restricted cone. To apply our theoretical results, we presented two recipes, which works under different conditions, to compute the statistical dimension of the prior restricted cone, and a precise analysis of these two recipes were given. Moreover, to illustrate our results, we applied our theoretical results and the estimation recipes to several specific problems, and obtained computable formulas for the statistical dimension and related error bounds. Simulations were provided to demonstrate our results.

\appendices

\section{Proof of Lemma \ref{lem: optimality condition}} \label{app: proof of main results}

\textit{Sufficiency.} We argue by contradiction. Suppose that $\fdc \cap \nullspace(\mA) = \{\vO\}$, but problem \eqref{eq: problem with multiple prior} fails. Then, the solution to problem \eqref{eq: problem with multiple prior}, $\hat{\fx}$, satisfies $\hat{\fx} \neq \fxs$. Since $\hat{\fx}$ is the solution, $\hat{\fx}$ must have smaller or equal cost than $\fxs$, and satisfy all the constraints, i.e.,
\begin{equation} \label{eq: solution satisfy prior constraints}
	f_i(\hat{\fx}) \le f_i(\fxs), \quad \textnormal{for all } 0 \le i \le k,
\end{equation}
and
\begin{equation} \label{eq: solution satisfy measurements}
	\vy = \mA \hat{\fx}.
\end{equation}
The identity \eqref{eq: solution satisfy prior constraints} implies that $\hat{\fx} - \fxs \in \fdc$ and the identity \eqref{eq: solution satisfy measurements} implies that $\hat{\fx} - \fxs \in \nullspace(\mA)$. Thus, $\hat{\fx} - \fxs \in \fdc \cap \nullspace(\mA)$. As $\hat{\fx} \neq \fxs$, we know that $\fdc \cap \nullspace(\mA) \neq \{\vO\}$. A contradiction.

\textit{Necessaty.} Again we argue by contradiction. Suppose that problem \eqref{eq: problem with multiple prior} succeeds, but we have $\fdc \cap \nullspace(\mA) \neq \{\vO\}$. Take any $\vd \in \fdc \cap \nullspace(\mA)$ and $\vd \neq \vO$. Since $\fdc = \cone (\fds)$, where $\fds$ is the prior restricted set of problem \eqref{eq: problem with multiple prior}, there exists a $t > 0$ such that $t \vd \in \fds$. The definition of $\fds$ implies that 
$$
f_i(\fxs + t \vd) \le f_i(\fxs), \quad \textnormal{for all } 0\le i \le k.
$$
Moreover, $\vd \in \nullspace(\mA)$ implies that
$$
\vy = \mA (\fxs + t \vd).
$$
In other words, we have shown that $\fxs + t \vd$ has smaller or equal cost than $\fxs$, and satisfies all the constraints. Thus, $\fxs$ must not be the unique solution to problem \eqref{eq: problem with multiple prior}. A contradiction.

\section{Proof of Theorem \ref{th: calc_statistical_dimension_1}} \label{app: proof for recipe 1}

In this section, we prove our Theorem \ref{th: calc_statistical_dimension_1} and related results. In subsection \ref{subsec: statistical dimension in terms of normal cones}, we prove Lemma \ref{lem: calc_statistical_dimension}. In subsections \ref{subsec: distance to sum of subdifferentials} and \ref{subsec: expected distance to sum of sets}, we give a detailed a proof of the properties of the function $J$, defined in Theorem \ref{th: calc_statistical_dimension_1}. The proof idea is inspired by \cite[Appendix C]{AMEL2014}, but our proof relies on some different proof techniques. In subsection \ref{subsec: completion of proof of recipe 1}, we complete the proof for Theorem \ref{th: calc_statistical_dimension_1}.

\subsection{Proof of Lemma \ref{lem: calc_statistical_dimension}} \label{subsec: statistical dimension in terms of normal cones}

To begin, note that the prior restricted cone $\fdc$ is determined by several descent cones of convex functions. Actually, by definition, the prior restricted set $\PP$ can be expressed as:
$$
\PP = D_s(f_0,\fxs) \cap D_s(f_1, \fxs) \cap \dots \cap D_s(f_k, \fxs),
$$
where
$$
D_s(f_i,\fxs) = \big\{ \vd: f_i(\fxs + \vd) \le f_i(\fxs) \big\} \quad \textnormal{for } i = 0,1,\dots,k.
$$
Now we argue that
\begin{equation} \label{eq: P_c is interaction of tangent cones}
	\fdc = D(f_0,\fxs) \cap D(f_1, \fxs) \cap \dots \cap D(f_k, \fxs),
\end{equation}
where $D(f_i,\fxs)$ denotes the descent cones of $f_i$ at $\fxs$, $i = 0,1,\dots,k$, i.e.,
$$
D(f_i,\fxs) = \cone\big(D_s(f_i, \fxs)\big) = \big\{\vd: \exists \, t > 0, f_i(\fxs + t \vd) \le f_i(\fxs) \big\} \quad \textnormal{for } i = 0,1,\dots,k.
$$
To see this, first note that it is clear that $\fdc \subseteq D(f_0,\fxs) \cap D(f_1, \fxs) \cap \dots \cap D(f_k, \fxs)$, so it remains to show the reverse relation holds. Take any $\vd \in D(f_0,\fxs) \cap D(f_1, \fxs) \cap \dots \cap D(f_k, \fxs)$, then there exists some number $t_i > 0$ such that $t_i \vd \in D_s(f_i, \fxs)$ for any $0 \le i \le k$. Denote $t \coloneqq \min_{0 \le i \le k}t_i > 0$. The convexity of $f_i$ implies that
$$
f_i(\fxs + t \vd) = f_i\big( (1-\lambda_i) \fxs + \lambda_i (\fxs + t_i \vd)\big) \le (1-\lambda_i) f_i(\fxs) + \lambda_i f_i(\fxs + t_i \vd) \le f_i(\fxs),
$$
where $\lambda_i = t / t_i \in (0, 1]$. Since the above inequality holds for any $0 \le i \le k$, we obtain that $t \vd \in \fds$. Thus, $\vd \in \fdc$. The identity \eqref{eq: P_c is interaction of tangent cones} follows immediately. 

Next, taking polar on both sides of \eqref{eq: P_c is interaction of tangent cones} yields
\begin{equation*} 
  	\fdc^{\circ} = \big[ D(f_0,\fxs) \cap D(f_1, \fxs) \cap \dots \cap D(f_k, \fxs) \big]^{\circ}.
\end{equation*}
Since we have assumed that $\ri\big(D(f_0,\fxs)\big) \cap \ri\big(D(f_1,\fxs)\big) \cap \dots \cap \ri\big(D(f_k,\fxs)\big) \neq \emptyset$, by \cite[Corollary 23.8.1]{ROCK1970}, the normal cone to the intersection of sets is the Minkowski sum of the normal cones to the individual sets:
\begin{equation} \label{eq: polar of P_c}
	\fdc^{\circ} = \big[ D(f_0,\fxs) \cap D(f_1, \fxs) \cap \dots \cap D(f_k, \fxs) \big]^{\circ} = \sum_{i=0}^{k} N(f_i,\fxs).
\end{equation}
Recall that the statistical dimension of a convex cone can be expressed via its polar \cite[Proposition 3.1 (4)]{AMEL2014}, so we obtain from \eqref{eq: polar of P_c} that
$$
\delta(\fdc) = \E \big[ \dist^2(\vg, \fdc^{\circ}) \big] = \E \big[ \dist^2 (\vg, \sum_{i=0}^{k} N(f_i,\fxs))\big].
$$

\subsection{Distance to the Sum of Compact Sets} \label{subsec: distance to sum of subdifferentials}

In this subsection, we study some analytic properties of the function $J_{\vu}(\vtau)$, which is related to $J(\vtau)$, but more simpler. We begin by studying some properties of the Minkowski sum of compact sets.

\begin{lemma} [Sum of compact sets] \label{lem: sum of sets}
  	For any $0 \le i \le k$, let $S_i$ be a non-empty, compact, convex subset of $\R^n$ that does not contain the origin, and $\vtau \in \S^k \cap \R^{k+1}_{+}$. Suppose that $\|\vs_i\|_2 \le B_i$ for some $B_i > 0$ and for any $\vs_i \in S_i$, $0 \le i \le k$. Then there exists a number $B > 0$ such that
	\begin{equation} \label{eq: bound for sum of compact sets}
		\Big\|\sum_{i = 0}^{k} \vtau_i \vs_i \Big\|_2 \le B, \ \textnormal{for any } \vtau \in \S^{k} \cap \R^{k+1}_{+}\ \textnormal{and} \ \vs_i \in S_i, \ 0 \le i \le k.
  	\end{equation}
	Furthermore, suppose that
	\begin{equation*} \label{eq: origin not in sum}
	  	\vO \notin \sum_{i = 0}^{k} \vtau_i S_i, \ \textnormal{for any } \vtau \in \S^{k} \cap \R^{k+1}_{+}.
  	\end{equation*}
	Then there exists a number $b > 0$ such that
	\begin{equation} \label{eq: bound for sum of sets}
		\Big\|\sum_{i = 0}^{k} \vtau_i \vs_i \Big\|_2 \ge b, \ \textnormal{for any } \vtau \in \S^{k} \cap \R^{k+1}_{+}, \ \textnormal{and} \ \vs_i \in S_i, \ 0 \le i \le k.
  	\end{equation}
\end{lemma}

\begin{proof}
  	\textit{Upper bound.} The upper bound in \eqref{eq: bound for sum of compact sets} is easy to obtain. Actually, by the triangle inequality and the Cauchy-Schwarz inequality, for any $\vtau \in \S^{q} \cap \R^{q+1}_{+}$, we have
	$$
	\Big\|\sum_{i = 0}^{k} \vtau_i \vs_i \Big\|_2 \le \sum_{i = 0}^{k} \vtau_i \| \vs_i \|_2 \le \|\vtau\|_2 \cdot \sqrt{\sum_{i = 0}^{k} \|\vs_i\|_2^2} \le \|\vtau\|_2 \cdot \sqrt{\sum_{i = 0}^{k} B_i^2} = \sqrt{\sum_{i = 0}^{k} B_i^2} \coloneqq B.
	$$

	\textit{Lower bound.} We prove the lower bound by contradiction. Suppose that there does not exist $b > 0$ satisfying \eqref{eq: bound for sum of sets}, which implies that
	\begin{equation} \label{eq: infimum is 0}
	  \inf_{\vtau \in \S^{k} \cap \R^{k+1}_{+}} \inf_{\vs_i \in S_i, 0 \le i \le k} \Big\|\sum_{i=0}^{k} \vtau_i \vs_i\Big\|_2 = 0.
  	\end{equation}
	Let's consider the function $r(\vtau) \coloneqq \inf_{\vs_i \in S_i, 0 \le i \le k} \big\| \sum_{i=0}^{k} \vtau_i \vs_i\big\|_2$, where $\vtau \in \R^{k+1}_{+}$, and prove that it is continuous. To this end, let $\vtau, \tilde{\vtau} \in \R^{k+1}_{+}$. Note that the sum of compact sets is compact \cite[Excercise 3(d), page 38]{RUDIN1991}. As a result, both $\sum_{i = 0}^{k} \vtau_i S_i$ and $\sum_{i = 0}^{k} \tilde{\vtau}_i S_i$ are compact. It follows that there exists $\vs_i^{\star} \in S_i$, $0 \le i \le k$, such that 
	$$
	\Big\|\sum_{i=0}^{q} \tilde{\vtau}_i \vs_i^{\star}\Big\|_2 = r(\tilde{\vtau}).
	$$
	Therefore, by the triangle inequality, we have
	\begin{align} \label{eq: one side bound}
	  r(\vtau) - r(\tilde{\vtau}) 
	  &\le \Big\|\sum_{i=0}^{k} \vtau_i \vs_i^{\star}\Big\|_2 - \Big\|\sum_{i=0}^{k} \tilde{\vtau}_i \vs_i^{\star}\Big\|_2 \le \Big\|\Big( \sum_{i=0}^{k} \vtau_i \vs_i^{\star} \Big) - \Big( \sum_{i=0}^{k} \tilde{\vtau}_i \vs_i^{\star} \Big)\Big\|_2 \notag \\
	  &= \Big\|\sum_{i=0}^{k} (\vtau_i-\tilde{\vtau}_i) \vs_i^{\star}\Big\|_2 \le \|\vtau - \tilde{\vtau}\|_2 \cdot B. 
  	\end{align}
	In the last inequality, we have used the upper bound \eqref{eq: bound for sum of compact sets}. By interchanging the roles of $\vtau$ and $\tilde{\vtau}$ in \eqref{eq: one side bound}, we obtain that
	\begin{equation} \label{eq: two side bound}
	  \big|r(\vtau) - r(\tilde{\vtau})\big| \le \|\vtau - \tilde{\vtau}\|_2 \cdot B, 
  	\end{equation}
	which implies that $r(\vtau)$ is Lipschitz function. The continuity of $r(\vtau)$ follows immediately. Now recall that a continuous function in a compact set must attain its infimum \cite[Theorem 4.16]{RUDIN1976}, therefore, \eqref{eq: infimum is 0} indicates that there exists a $\vtau \in \S^{k} \cap \R^{k+1}_{+}$ such that
	\begin{equation} \label{eq: origin in sum}
		\inf_{\vs_i \in S_i, 0 \le i \le k} \Big\|\sum_{i=0}^{k} \vtau_i \vs_i\Big\|_2 = 0.
  	\end{equation}
	Since $\sum_{i = 0}^{k} \vtau_i S_i$ is closed, \eqref{eq: origin in sum} implies that $\vO \in \sum_{i=0}^{k} \vtau_i S_i$. A contradiction. Therefore, there must exist some $b > 0$ satisfying \eqref{eq: bound for sum of sets}.
\end{proof}

Lemma \ref{lem: sum of sets} gives upper and lower bounds for the length of elements of sum of compact sets. We remind that when we write $B$ and $b$ hereafter, we always mean the numbers in \eqref{eq: bound for sum of compact sets} and \eqref{eq: bound for sum of sets}, respectively. Using Lemma \ref{lem: sum of sets}, we can study the properties of function $J_{\vu}$, which is the distance of a point to sum of compact sets.

\begin{lemma} [Distance to the sum of compact sets] \label{lem: distance to sum of sets}
  	Let $S_i$, $0 \le i \le k$, be a non-empty, compact, convex subset of $\R^n$ that does not contain the origin. Suppose that $\|\vs_i\|_2 \le B_i$ for some $B_i > 0$ and for any $\vs_i \in S_i$, $0 \le i \le k$. Moreover, suppose that
	\begin{equation} \label{eq: origin not in sum 2}
	  	\vO \notin \sum_{i = 0}^{k} \vtau_i S_i, \ \ \textnormal{for any } \vtau \in \S^{k} \cap \R^{k+1}_{+}.
  	\end{equation}
	Fix a point $\vu \in \R^n$, and define the function $J_{\vu} : \R_+^{k + 1} \rightarrow \R$ by
	$$
	J_{\vu}(\vtau) \coloneqq \dist^2(\vu, \sum_{i = 0}^{k} \vtau_i S_i),
	$$
	where $\vtau = (\vtau_0, \vtau_1,\dots, \vtau_{k}) \in \R_{+}^{k + 1}$.
	Then $J_{\vu}(\vtau)$ has the following properties:
	\begin{enumerate}
	  \item
		The function $J_{\vu}$ is convex and continuous.
	  \item
		The function $J_{\vu}$ has the lower bound
		\begin{equation} \label{eq: lower bound for J_u}
		  J_{\vu}(\vtau) \ge (\|\vtau\|_2 b - \|\vu\|_2)^2,\ \textnormal{when } \|\vtau\|_2 \ge \frac{\|\vu\|_2}{b}.
	  	\end{equation}
		In particular, $J_{\vu}$ attains its minimum in the compact subset $\B(\vO, 2\|\vu\|_2/b) \cap \R^{k+1}_{+}$.
	  \item
		The function $J_{\vu}$ is continuously differentiable, and its partial derivative is
		\begin{equation} \label{eq: partial derivative of J_u}
		  	\frac{\partial J_{\vu}}{\partial \vtau_i}(\vtau) = -2\left< \vu- \sum_{i=0}^{k} \vtau_i \bar{\vs}_i, \bar{\vs}_i\right> \quad \textnormal{for any} \ \vtau \in \R^{k+1}_+,
		\end{equation}
		where $\bar{\vs}_i \in S_i, \, 0 \le i \le k$, satisfies $\dist^2(\vu, \sum_{i = 0}^{k} \vtau_i \bar{\vs}_i) = J_{\vu}(\vtau)$. For $\vtau$ in the boundary of $\R^{k+1}_+$, we interpret the partial derivative $\frac{\partial J_{\vu}}{\partial \vtau_i}$ similarly as the right derivative if $\vtau_i = 0$, i.e., 
		$$
		\frac{\partial J_{\vu}}{\partial \vtau_i} (\vtau) = \lim_{\epsilon \downarrow 0} \frac{J_{\vu}(\vtau_0, \dots, \vtau_{i-1}, \epsilon, \dots, \vtau_k) - J_{\vu}(\vtau_0, \dots, \vtau_{i-1}, 0, \dots, \vtau_k)}{\epsilon}.
		$$
	  \item
		The partial derivative of $J_{\vu}$ satisfies the following bound:
		\begin{equation} \label{eq: bound for partial derivative}
			\Big| \frac{\partial J_{\vu}}{\partial \vtau_i} (\vtau) \Big| \le 2B_i \big( \|\vu\|_2 + \|\vtau\|_2 B \big).
	  	\end{equation}
	  \item
		For any fixed $\vtau \in \R^{k+1}_+$ and any $0 \le i \le k$, the map $\vu \mapsto \frac{\partial J_{\vu}}{\partial \vtau_i}(\vtau)$ is Lipschitz:
		\begin{equation} \label{eq: partial derivative is Lipschitz}
		  	\Big| \frac{\partial J_{\vu}}{\partial \vtau_i}(\vtau) - \frac{\partial J_{\vu'}}{\partial \vtau_i}(\vtau) \Big| \le 2B_i \cdot \|\vu - \vu'\|_2.
		\end{equation}
	\end{enumerate}
\end{lemma}

\begin{proof}
  	Lemma \ref{lem: distance to sum of sets} is a generalization of \cite[Lemma C.1]{AMEL2014} from a dilated set to the sum of several sets.

	\textit{Convexity.} Note that to prove the convexity of $J_{\vu}$, it is sufficient to prove that the function
	$$
	J^{\frac{1}{2}}_{\vu}(\vtau) \coloneqq \sqrt{J_{\vu}(\vtau)} = \dist (\vu, \sum_{i = 0}^{k} \vtau_i S_i)
	$$
	is convex. To this end, fix any $\vtau, \tilde{\vtau} \in \R^{k + 1}_{+}$ and $\lambda_1, \lambda_2 \in \R_{+}$ satisfying $\lambda_1 + \lambda_2 = 1$. Since $S_i$ is a convex set for $0 \le i \le k$, it follows from \cite[Theorem 3.2]{ROCK1970} that
	\begin{equation} \label{eq: convex decomposition of set}
	  (\lambda_1\vtau_i + \lambda_2 \tilde{\vtau}_i) S_i = \lambda_1\vtau_i S_i + \lambda_2 \tilde{\vtau}_i S_i \ \ \textnormal{for any} \ \ 0 \le i \le k.
  	\end{equation}
	Then by the definition of $J_{\vu}^{\frac{1}{2}}$ and the triangle inequality, we have
	\begin{align*} \label{eq: convexity of J^1/2}
	  	J^{\frac{1}{2}}_{\vu}(\lambda_1\vtau + \lambda_2 \tilde{\vtau}) 
		&= \dist \big(\vu, \sum_{i = 0}^{k} (\lambda_1\vtau_i + \lambda_2 \tilde{\vtau}_i) S_i\big)
		= \dist \big(\vu, \sum_{i = 0}^{k} (\lambda_1\vtau_iS_i + \lambda_2 \tilde{\vtau}_i S_i)\big) \\
		&= \dist \big(\lambda_1\vu + \lambda_2\vu, \lambda_1 (\sum_{i = 0}^{k} \vtau_iS_i) + \lambda_2 (\sum_{i = 0}^{k} \tilde{\vtau}_i S_i)\big) \\
		&= \inf_{\vs_i \in S_i, \tilde{\vs}_i \in S_i, 0 \le i \le k} \Big\|\lambda_1\vu + \lambda_2\vu - \big[\lambda_1 \cdot \big( \sum_{i = 0}^{k} \vtau_i\vs_i \big) + \lambda_2 \cdot \big( \sum_{i = 0}^{k} \tilde{\vtau}_i \tilde{\vs}_i \big)\big] \Big\|_2 \\
		&\le \inf_{\vs_i \in S_i, \tilde{\vs}_i \in S_i, 0 \le i \le k} \lambda_1 \Big\|\vu - \sum_{i = 0}^{k} \vtau_i\vs_i \Big\|_2 + \lambda_2 \Big\| \vu - \sum_{i = 0}^{k} \tilde{\vtau}_i \tilde{\vs}_i \Big\|_2 \\
		&= \lambda_1 \cdot \dist\big(\vu, \sum_{i = 0}^{k} \vtau_i S_i\big) + \lambda_2 \cdot \dist\big(\vu, \sum_{i = 0}^{k} \tilde{\vtau}_i S_i\big) \\
		&=\lambda_1 J^{\frac{1}{2}}_{\vu}(\vtau) + \lambda_2 J^{\frac{1}{2}}_{\vu}(\tilde{\vtau}),
	\end{align*}
	which implies that $J^{\frac{1}{2}}_{\vu}$ is convex. The convexity of $J_{\vu}$ follows immediately.

	\textit{Continiuty.} We first consider the case when $\vtau \in \R^{k+1}_{++}$ and take any $\vepsilon \in \R^{k+1}$. To check the continuity, note that
	\begin{equation} \label{eq: dist to e+t}
	  J^{\frac{1}{2}}_{\vu}(\vtau + \vepsilon) = \dist\Big(\vu, \sum_{i=0}^{k} (\vtau_i + \vepsilon_i) S_i \Big) = \inf_{\vs_i \in S_i, 0 \le i \le k} \Big\| \vu - \big(\sum_{i=0}^{k} \vepsilon_i \vs_i + \sum_{i=0}^{k} \vtau_i \vs_i\big) \Big\|_2.
	\end{equation}
	The triangle inequality gives us that
	\begin{equation} \label{eq: triangle to dist to e+t}
		\Big\| \vu - \sum_{i=0}^{k} \vtau_i \vs_i \Big\|_2 - \Big\| \sum_{i=0}^{k} \vepsilon_i \vs_i \Big\|_2 \le \Big\| \vu - \big(\sum_{i=0}^{k} \vepsilon_i \vs_i + \sum_{i=0}^{k} \vtau_i \vs_i\big) \Big\|_2 \le \Big\| \vu - \sum_{i=0}^{k} \vtau_i \vs_i \Big\|_2 + \Big\| \sum_{i=0}^{k} \vepsilon_i \vs_i \Big\|_2.
  	\end{equation}
	Putting \eqref{eq: dist to e+t} and \eqref{eq: triangle to dist to e+t} together, we obtain that
	\begin{equation} \label{eq: triangle to dist}
		\dist\Big(\vu, \sum_{i=0}^{k} \vtau_i S_i\Big) - \sup_{\vs_i \in S_i, \atop 0 \le i \le k} \Big\| \sum_{i=0}^{k} \vepsilon_i \vs_i \Big\|_2 \le \dist\Big(\vu, \sum_{i=0}^{k} (\vtau_i + \vepsilon_i) S_i\Big)
		\le \dist\Big(\vu, \sum_{i=0}^{k} \vtau_i S_i\Big) + \sup_{\vs_i \in S_i, \atop 0 \le i \le k} \Big\| \sum_{i=0}^{k} \vepsilon_i \vs_i \Big\|_2.
	  \end{equation}
	Now recalling the upper bound in \eqref{eq: bound for sum of compact sets}, we obtain from \eqref{eq: triangle to dist} that
	$$
	\dist\Big(\vu, \sum_{i=0}^{k} \vtau_i S_i\Big) - \|\vepsilon\|_2 B \le \dist\Big(\vu, \sum_{i=0}^{k} (\vtau_i + \vepsilon_i) S_i \Big) \le \dist\Big(\vu, \sum_{i=0}^{k} \vtau_i S_i\Big) + \|\vepsilon\|_2 B.
	$$
	In other words,
	\begin{equation} \label{eq: bound for dist}
		J^{\frac{1}{2}}_{\vu}(\vtau) - \|\vepsilon\|_2 B \le J^{\frac{1}{2}}_{\vu}(\vtau + \vepsilon) \le J^{\frac{1}{2}}_{\vu}(\vtau) + \|\vepsilon\|_2 B.
  	\end{equation}
	Squaring both sides, we obtain that
	$$
	\|\vepsilon\|^2_2 B^2 - 2 \|\vepsilon\|_2 B \cdot J^{\frac{1}{2}}_{\vu} (\vtau) \le J_{\vu}(\vtau + \vepsilon) - J_{\vu}(\vepsilon) \le \|\vepsilon\|_2^2 B^2 + 2 \|\vepsilon\|_2 B \cdot J_{\vu}^{\frac{1}{2}} (\vtau).
	$$
	Moreover, select any $\vs_i \in S_i$, $0 \le i \le k$, and we have
	$$
  J_{\vu}^{\frac{1}{2}}(\vtau) = \dist(\vu, \sum_{i=0}^k \vtau_i S_i) \le \|\vu - \sum_{i=0}^k \vtau_i \vs_i\|_2 \le \|\vu\|_2 + \big\|\sum_{i=0}^k \vtau_i \vs_i\big\|_2 \le \|\vu\|_2 + \|\vtau\|_2 B,
	$$
	where we have used the triangle inequality and the upper bound \eqref{eq: bound for sum of compact sets}. It follows that
	\begin{align} \label{eq: difference of J}
	  	\big| J_{\vu}(\vtau + \vepsilon) - J_{\vu}(\vepsilon) \big|
	  	&\le \|\vepsilon\|_2^2 B^2 + 2 \|\vepsilon\|_2 B \cdot J_{\vu}^{\frac{1}{2}} (\vtau) \notag 
		= \|\vepsilon\|_2^2 B^2 + 2 \|\vepsilon\|_2 B \cdot \dist(\vu, \sum_{i=0}^k \vtau_i S_i) \notag \\
		&\le \|\vepsilon\|_2^2 B^2 + 2 \|\vepsilon\|_2 B \cdot \big( \|\vu\|_2 + \|\vtau\|_2 B \big).
	\end{align}
	Now it is easy to see that if $\epsilon \rightarrow \vO$, we have $\big| J_{\vu}(\vtau + \vepsilon) - J_{\vu}(\vepsilon) \big| \rightarrow 0$. Similar argument holds as well when $\vtau$ is on the boundary of $\R_{+}^{k+1}$. Therefore, we conclude that the function $J_{\vu}$ is continuous in $\R_+^{k+1}$.

	\textit{Attainment of minimum.} Note that by Lemma \ref{lem: sum of sets}, we know that there exists a number $b > 0$ such that
	$$\Big\|\sum_{i = 0}^{k} \vtau_i \vs_i \Big\|_2 \ge b, \ \textnormal{for any } \vtau \in \S^{k} \cap \R^{k+1}_{+}, \ \vs_i \in S_i, \ 0 \le i \le k.
	$$
	Therefore, for any $\vtau \neq \vO$,
	\begin{equation} \label{eq: expression of J}
	  	J^{\frac{1}{2}}_{\vu}(\vtau) 
		= \dist(\vu, \sum_{i=0}^k \vtau_i S_i) = \inf_{\vs_i \in S_i, 0 \le i \le k} \big\|\vu - \sum_{i=0}^{k} \vtau_i \vs_i \big\|_2 \ge \inf_{\vs_i \in S_i, 0 \le i \le k} \big\|\sum_{i=0}^{k} \vtau_i \vs_i \big\|_2 - \|\vu\|_2 \ge \|\vtau\|_2 \cdot b - \|\vu\|_2.
  	\end{equation}
	Thus, when $\|\vtau\|_2 \ge \|\vu\|_2/b$, by squaring both sides of \eqref{eq: expression of J}, we obtain the lower bound
	$$
	J_{\vu}(\vtau) \ge \big( \|\vtau\|_2 \cdot b - \|\vu\|_2 \big)^2.
	$$
	Moreover, if $\|\vtau\|_2 \ge 2\|\vu\|_2/b$, we have $J_{\vu}(\vtau) \ge \|\vu\|_2^2 = J_{\vu}(\vO)$. Then, it follows from the convexity and continuity of $J_{\vu}$ that the function $J_{\vu}$ must attain its minimum in the compact set $\B(\vO, 2\|\vu\|_2/b) \cap \R^{k+1}_{+}$.

	\textit{Continuous differentiability in $\R^{k+1}_{++}$.} To prove that $J_{\vu}$ is continuously differential in $\R^{k+1}_{++}$, we need to show that the partial derivative $\partial J_{\vu} / \partial \vtau_i$ exists and is continuous, for any $0 \le i \le k$. For this purpose, fix any $0 \le i \le k$, and define the function $\tilde{J}_{\vu}(\vtau_i)$ to be
	\begin{equation} \label{eq: rewritten of J}
	  	\tilde{J}_{\vu}(\vtau_i) \coloneqq J_{\vu}(\vtau) =  \dist^2(\vu, \sum_{i=0}^{k} \vtau_i S_i) = \dist^2(\vu, T + \vtau_i S_i) = \inf_{\vt \in T} \dist^2(\vu - \vt, \vtau_i S_i),
  \end{equation}
	where $T = \sum_{0 \le j \le k, j \neq i}\vtau_i S_i$. Now define another function $g(\vtau_i, \vt) = \dist^2(\vu - \vt, \vtau_i S_i)$, $(\vtau_i, \vt) \in \R_{++} \times \R^n$. The function $g(\vtau_i, \vt)$ is continuously differentiable. To see this, first note that the function $\partial g / \partial \vtau_i$ exists, and takes the form
	$$
	\frac{\partial g}{\partial \vtau_i} = -\frac{2}{\vtau_i}\left< \vu- \vt - \Pi_{\vtau_i S_i}(\vu - \vt), \Pi_{\vtau_i S_i}(\vu- \vt)\right>.
	$$
	Moreover, $\partial g / \partial \vtau_i$ is continuous \cite[Lemma C.1, (3)]{AMEL2014}. Next, the function $\tilde{g}(\vt) = \dist^2(\vu- \vt, \vtau_i S_i)$ is differentiable, and the differential is
	$$
	\nabla \tilde{g}(\vt) = -2\big(\vu - \vt - \Pi_{\vtau_i S_i}(\vu - \vt)\big).
	$$
	This point results from \cite[Theorem 2.26]{ROCK1998}. Furthermore, the projection onto a convex set is continuous \cite[Theorem 2.26]{ROCK1998}, hence, $\nabla \tilde{g}$ is a continuous function. It follows that $\partial g / \partial \vt_j$ is continuous for any $1 \le j \le n$. Therefore, we obtain that the function $g(\vtau_i, \vt)$ is continuously differentiable in $\R_{++} \times \R^n$. As a result of \cite[Theorem 2.8]{SPIVAK1965}, $g(\vtau_i, \vt)$ is differentiable in $\R_{++} \times \R^n$, and the differential is
	$$
	\nabla g(\vtau_i, \vt) = \Big[-\frac{2}{\vtau_i}\left< \vu- \vt - \Pi_{\vtau_i S_i}(\vu - \vt), \Pi_{\vtau_i S_i}(\vu- \vt)\right>, -2\big(\vu - \vt - \Pi_{\vtau_i S_i}(\vu - \vt)\big)^T \Big]^T.
	$$
	The subdifferential of a differentiable function contains only the differential of the function \cite[Theorem 25.1]{ROCK1970}. Thus, the subdifferential of $g$ at $(\vtau_i, \vt)$ is
	\begin{equation} \label{eq: subdifferential of g}
		\partial g(\vtau_i, \vt) = \Big\{ \Big[-\frac{2}{\vtau_i}\left< \vu- \vt - \Pi_{\vtau_i S_i}(\vu - \vt), \Pi_{\vtau_i S_i}(\vu- \vt)\right>, -2\big(\vu - \vt - \Pi_{\vtau_i S_i}(\vu - \vt)\big)^T \Big]^T \Big\}.
	\end{equation}
	Since $T$ is compact, we can take a $\bar{\vt} \in T$ such that $g(\vtau_i, \bar{\vt}) = \tilde{J}_{\vu}(\vtau_i)$. Then let us confirm that $-\nabla \tilde{g}(\bar{\vt}) = 2\big(\vu - \bar{\vt} - \Pi_{\vtau_i S_i}(\vu - \bar{\vt})\big) \in N(\bar{\vt}; T)$, where $N(\bar{\vt}; T) \coloneqq \{ \vw \in \R^n: \left< \vw, \vt - \bar{\vt}\right> \le 0,\ \forall \, \vt \in T\}$, denotes the normal cone to $T$ at $\bar{\vt}$. To this end, let $\vs_i \in S_i$ such that $\vtau_i \vs_i = \Pi_{\vtau_i S_i}(\vu - \bar{\vt})$. From another point of view, it is not difficult to see that $\bar{\vt} = \Pi_{T}(\vu - \vtau_i \vs_i)$. Thus, 
	$$
	-\nabla \tilde{g}(\bar{\vt}) = 2\big(\vu - \bar{\vt} - \Pi_{\vtau_i S_i}(\vu - \bar{\vt})\big) = 2\big( \vu - \vtau_i \vs_i - \Pi_{T}(\vu - \vtau_i \vs_i) \big).
	$$
	By \cite[Theorem III.3.1.1]{HIRIART1993}, we know that
	$$
	\left< \vu - \vtau_i \vs_i - \Pi_{T}(\vu - \vtau_i \vs_i), \vt - \Pi_{T}(\vu - \vtau_i \vs_i)\right> \le 0, \ \textnormal{for any} \ \vt \in T.
	$$
	Therefore, we obtain that 
	\begin{equation} \label{eq: subdifferential in normal cone}
	  -\nabla \tilde{g}(\bar{\vt}) = 2\big(\vu - \bar{\vt} - \Pi_{\vtau_i S_i}(\vu - \bar{\vt})\big) \in N(\bar{\vt}; T).
  	\end{equation}
	Now we can give a conclusion about the subdifferential of $\tilde{J}_{\vu}$:
	$$
	\partial \tilde{J}_{\vu}(\vtau_i) = \Big\{ -\frac{2}{\vtau_i}\left< \vu- \bar{\vt} - \Pi_{\vtau_i S_i}(\vu - \bar{\vt}), \Pi_{\vtau_i S_i}(\vu- \bar{\vt})\right> \Big\}.
	$$
	This is a direct consequence of \cite[Example 2.59 and Theorem 2.61]{MORD2014}, \eqref{eq: subdifferential of g}, \eqref{eq: subdifferential in normal cone}, and the fact that $g(\vtau_i, \vt)$ is continuous. That the subdifferential of $\tilde{J}_{\vu}$ is a singleton implies $\tilde{J}_{\vu}$ is differentiable \cite[Theorem 25.1]{ROCK1970}, and the differential is
	$$
	\tilde{J}'_{\vu}(\vtau_i) = -\frac{2}{\vtau_i}\left< \vu- \bar{\vt} - \Pi_{\vtau_i S_i}(\vu - \bar{\vt}), \Pi_{\vtau_i S_i}(\vu- \bar{\vt'})\right>.
	$$
	The above formula is equivalent to that the partial derivative $\partial J_{\vu} / \partial \vtau_i$ exists, and takes the form
	$$
	\frac{\partial J_{\vu}}{\partial \vtau_i}(\vtau) = -\frac{2}{\vtau_i}\left< \vu- \bar{\vt} - \vtau_i \bar{\vs}_i, \vtau_i \bar{\vs}_i\right>,
	$$
	for any $\bar{\vt} \in T, \bar{\vs}_i \in S_i$ such that $\dist^2(\vu, \bar{\vt} + \vtau_i \bar{\vs}_i) = \tilde{J}_{\vu}(\vtau_i) = J_{\vu}(\vtau)$. Since $T = \sum_{0 \le j \le k, j \neq i} \vtau_i S_i$ is compact, hence, 
	$$
  	\bar{\vt} = \sum_{0 \le j \le k, j \neq i}\vtau_j \bar{\vs}_j, \ \textnormal{for some} \ \bar{\vs}_j \in S_j, \ 0 \le j \le k, \ j \neq i.
	$$
	Therefore, the partial derivative $\partial J_{\vu} / \partial \vtau_i$ can be rewritten as
	$$
		\frac{\partial J_{\vu}}{\partial \vtau_i}(\vtau) = -\frac{2}{\vtau_i}\left< \vu- \sum_{i=0}^{k} \vtau_i \bar{\vs}_i, \vtau_i \bar{\vs}_i\right> = -2\left< \vu- \sum_{i=0}^{k} \vtau_i \bar{\vs}_i, \bar{\vs}_i\right>
		$$
	for any $\bar{\vs}_i \in S_i, 0 \le i \le k$, such that $\|\vu - \sum_{i = 0}^{k} \vtau_i \bar{\vs}_i\|_2^2 = J_{\vu}(\vtau)$. It remains to prove that $\partial J_{\vu} / \partial \vtau_i$ is continuous in $\vtau_i$. Indeed, $\tilde{J}_{\vu}$ is a proper convex function, and is differential in $\R_{++}$. It follows from \cite[Theorem 25.5]{ROCK1970} that the gradient mapping $\tilde{J}'_{\vu}$ is continuous in $\R_{++}$, which means that $\partial J_{\vu} / \partial \vtau_i$ is continuous in $\R_{++}$. Since for any $0 \le i \le k$, $\partial J_{\vu} / \partial \vtau_i$ exists and is continuous in $\R_{++}$, we obtain that $J_{\vu}$ is continuously differentiable in $\R^{k+1}_{++}$.

	\textit{Differential at the boundary of $\R^{k+1}_{+}$ and its continuity.} The function $\tilde{J}_{\vu}$ is a closed proper convex function. It is continuous in $[0, +\infty]$ and continuously differentiable in $(0, +\infty)$. Hence, as a consequence of \cite[Theorem 24.1]{ROCK1970}, the right derivative at the origin exists and the limit formula holds. In other words, for any $\vtau \in \R^{k+1}_+$ with $\vtau_i = 0$, we have
	$$
	\frac{\partial J_{\vu}}{\partial \vtau_i} (\vtau) \coloneqq \lim_{\epsilon \downarrow 0} \frac{J_{\vu}(\vtau_0, \dots, \vtau_{i-1},\epsilon, \dots, \vtau_k) - J_{\vu}(\vtau_0, \dots, \vtau_{i-1}, 0, \dots, \vtau_k)}{\epsilon} = \lim_{\vtau_i \downarrow 0}\frac{\partial J_{\vu}}{\partial \vtau_i}(\vtau).
	$$
	To study the continuity of the differential of $J_{\vu}$ at the boundary of $\R^{k+1}_{+}$, without loss of generality, we assume that $\vtau = (\vtau_0, \vtau_1, \dots, \vtau_l, \vtau_{l+1}, \dots, \vtau_k)$, where $\vtau_i > 0$ for $0 \le i \le l$ and $\vtau_i = 0$ for $l < i \le k$. Let $\vh = (\vh_0, \vh_1, \dots, \vh_l, \vh_{l+1}, \dots, \vh_{k})$, where $\vh_i \ge 0$ for $l < i \le k$. Similar as the proof for \cite[Theorem 2.8]{SPIVAK1965}, we have
	\begin{align*}
	  J_{\vu}(\vtau + \vh) - J_{\vu}(\vtau) = & J_{\vu}(\vtau_0 + \vh_0, \vtau_1, \dots, \vtau_q) - J_{\vu}(\vtau_0, \vtau_1, \dots, \vtau_k) \\
	                                          & +J_{\vu}(\vtau_0 + \vh_0, \vtau_1 + \vh_1, \vtau_2, \dots, \vtau_q) - J_{\vu}(\vtau_0 + \vh_0, \vtau_1, \vtau_2, \dots, \vtau_k) \\
											  & + \dots \\
											  & +J_{\vu}(\vtau_0 + \vh_0, \vtau_1 + \vh_1, \dots, \vtau_{k-1} + \vh_{k-1}, \vtau_k + \vh_k) - J_{\vu}(\vtau_0 + \vh_0, \vtau_1 + \vh_1, \dots, \vtau_{k-1} + \vh_{k-1}, \vtau_k).
	\end{align*}
	Let us look at the first term $J_{\vu}(\vtau_0 + \vh_0, \vtau_1, \dots, \vtau_k) - J_{\vu}(\vtau_0, \vtau_1, \dots, \vtau_k)$. By the mean-value theorem, we know that there exist some $\vb_0$ between $\vtau_0$ and $\vtau_0 + \vh_0$ such that
	$$
	J_{\vu}(\vtau_0 + \vh_0, \vtau_1, \dots, \vtau_k) - J_{\vu}(\vtau_0, \vtau_1, \dots, \vtau_k) = \frac{\partial J_{\vu}}{\partial \vtau_0}(\vb_0, \vtau_1, \dots, \vtau_k) \cdot \vh_0.
	$$
	Similarly, for the $i$-th term, there exists some $\vb_{i-1}$ between $\vtau_{i-1}$ and $\vtau_{i-1} + \vh_{i-1}$ such that
	\begin{multline*}
	  	J_{\vu}(\vtau_0 + \vh_0, \dots, \vtau_{i-2} + \vh_{i-2}, \vtau_{i-1} + \vh_{i-1}, \vtau_i, \dots, \vtau_k) - J_{\vu}(\vtau_0 + \vh_0, \dots, \vtau_{i-2} + \vh_{i-2}, \vtau_{i-1}, \vtau_i, \dots, \vtau_k) \\
		= \frac{\partial J_{\vu}}{\partial \vtau_{i-1}}(\vtau_0 + \vh_0, \dots, \vtau_{i-2} + \vh_{i-2}, \vb_{i-1}, \vtau_i, \dots, \vtau_k) \cdot \vh_{i-1}.
  	\end{multline*}
	Then,
	\begin{align*}
	  &\lim_{\vh_i \rightarrow 0, 0 \le i \le l, \atop \vh_i \downarrow 0, l < i \le k} \frac{|J_{\vu}(\vtau + \vh) - J_{\vu}(\vtau) - \sum_{i=0}^k \frac{\partial J_{\vu}}{\partial \vtau_i} \cdot \vh_i|}{\|\vh\|_2} \\
	  &\hspace*{120pt}= \lim_{\vh_i \rightarrow 0, 0 \le i \le l, \atop \vh_i \downarrow 0, l < i \le k} \frac{ \Big|\sum_{i=0}^k \big[\frac{\partial J_{\vu}}{ \partial \vtau_i} (\vtau_0, \dots, \vb_i, \dots, \vtau_k) - \frac{\partial J_{\vu}}{\partial \vtau_i}(\vtau_0, \dots, \vtau_i, \dots\vtau_k) \big] \cdot \vh_i \Big|}{\|\vh\|_2} \\
	  &\hspace*{120pt}\le \lim_{\vh_i \rightarrow 0, 0 \le i \le l, \atop \vh_i \downarrow 0, l < i \le k}  \sum_{i=0}^k  \Big|\frac{\partial J_{\vu}}{ \partial \vtau_i} (\vtau_0, \dots, \vb_i, \dots, \vtau_k) - \frac{\partial J_{\vu}}{\partial \vtau_i}(\vtau_0, \dots, \vtau_i, \dots\vtau_k) \Big| \cdot \frac{|\vh_i|}{\|\vh\|_2} \\
	  &\hspace*{120pt}\le \lim_{\vh_i \rightarrow 0, 0 \le i \le l, \atop \vh_i \downarrow 0, l < i \le k}  \sum_{i=0}^k  \Big|\frac{\partial J_{\vu}}{ \partial \vtau_i} (\vtau_0, \dots, \vb_i, \dots, \vtau_k) - \frac{\partial J_{\vu}}{\partial \vtau_i}(\vtau_0, \dots, \vtau_i, \dots\vtau_k) \Big| \\
	  &\hspace*{120pt}= \lim_{\vb_i \rightarrow \vtau_i, 0 \le i \le l, \atop \vb_i \downarrow \vtau_i, l < i \le k}  \sum_{i=0}^k  \Big|\frac{\partial J_{\vu}}{ \partial \vtau_i} (\vtau_0, \dots, \vb_i, \dots, \vtau_k) - \frac{\partial J_{\vu}}{\partial \vtau_i}(\vtau_0, \dots, \vtau_i, \dots\vtau_k) \Big| \\
	  &\hspace*{120pt} = 0.
	\end{align*}
	The last identity holds because the partial derivative $\frac{\partial J_{\vu}}{ \partial \vtau_i}$ is continuous in $[0, +\infty)$.

	\textit{Bound for the partial derivative.} Using the Cauchy-Schwarz inequality to \eqref{eq: partial derivative of J_u}, we obtain that
	\begin{equation} \label{eq: bound for derivative}
		\Big| \frac{\partial J_{\vu}}{\partial \vtau_i}(\vtau) \Big| \le 2 \big\| \vu- \sum_{i=0}^{k} \vtau_i \bar{\vs}_i\big\|_2 \cdot \|\bar{\vs}_i\|_2.
  	\end{equation}
	The triangle inequality gives
	$$
	\big\| \vu - \sum_{i = 0}^{k} \vtau_i \bar{\vs}_i \big\|_2 \le \|\vu\|_2 + \| \sum_{i = 0}^{k} \vtau_i \bar{\vs}_i \|_2 \le \|\vu\|_2 + \|\vtau\|_2 B.
	$$
	The last inequality comes from \eqref{eq: bound for sum of compact sets}. Substituting it into \eqref{eq: bound for derivative} yields the desired result
	$$
	\Big| \frac{\partial J_{\vu}}{\partial \vtau_i}(\vtau) \Big| \le 2 B_i\big( \|\vu\|_2 + \|\vtau\|_2 B).
	$$

	\textit{Lipschitz property.} Fix any $i$, $0 \le i \le k$, and $\vtau \in \R^{k+1}_+$ satisfying $\vtau_i > 0$. We first make use of \cite[Theorem III.3.1.1]{HIRIART1993} to obtain that
	$$
	\left<\vu - \sum_{j=0}^k \vtau_j \bar{\vs}_j, \sum_{j=0}^k \vtau_j \bar{\vs}_j\right> \ge \left<\vu - \sum_{j=0}^k \vtau_j \bar{\vs}_j, \vtau_i \vs_i + \sum_{0 \le j \le k, \atop j \neq i} \vtau_j \bar{\vs}_j\right> \quad \textnormal{for any} \ \vs_i \in S_i,
	$$
	where $\bar{\vs}_i \in S_i, 0 \le i \le k$, satisfying $\|\vu - \sum_{i = 0}^{k} \vtau_i \bar{\vs}_i\|_2 = \dist(\vu, \sum_{i = 0}^{k} \vtau_i S_i)$. Simplifying the above inequality yields
	$$
	\left<\vu - \sum_{j=0}^k \vtau_j \bar{\vs}_j, \bar{\vs}_i\right> \ge \left<\vu - \sum_{j=0}^k \vtau_j \bar{\vs}_j, \vs_i\right> \quad \textnormal{for any} \ \vs_i \in S_i.
	$$
	Therefore, for any $\vu, \vu' \in \R^n$,
	\begin{align} \label{eq: lipschitz and non-expansive}
	  \left<\vu - \sum_{j=0}^k \vtau_j \bar{\vs}_j, \bar{\vs}_i\right> - \left<\vu' - \sum_{j=0}^k \vtau_j \bar{\vs}_j', \bar{\vs}_i'\right> 
	  &\le \left< \Big[\vu - \sum_{j=0}^k \vtau_j \bar{\vs}_j\Big] - \Big[\vu' - \sum_{j=0}^k \vtau_j \bar{\vs}_j'\Big], \bar{\vs}_i\right> \notag \\
	  &\le \big\| (\mI - \Pi_E )(\vu) - (\mI - \Pi_E )(\vu') \big\|_2 \cdot \|\bar{\vs}_i\|_2 \notag \\
	  &\le \| \vu - \vu' \|_2 \cdot B_i,
	\end{align}
	where $\bar{\vs}'_j \in S_j$ satisfying $\|\vu' - \sum_{j=0}^k \vtau_j \bar{\vs}'_j\|_2 = \dist(\vu', \sum_{j=0}^k \vtau_j S_j)$, and $\Pi_E(\vu)$ denotes the projection of $\vu$ onto the set $E \coloneqq \sum_{i=0}^k \vtau_i S_i$. In the second inequality, we have used the Cauchy-Schwarz inequality, and the last inequality comes from the fact that the map $\mI - \Pi_E$ is non-expansive with respect to the Euclidean norm \cite[pp. 275]{AMEL2014}. Interchanging the roles of $\vu$ and $\vu'$ in \eqref{eq: lipschitz and non-expansive}, we obtain that
	$$
	\Big| \left<\vu - \sum_{j=0}^k \vtau_j \bar{\vs}_j, \bar{\vs}_i\right> - \left<\vu' - \sum_{j=0}^k \vtau_j \bar{\vs}_j', \bar{\vs}_i'\right> \Big| \le \| \vu - \vu' \|_2 \cdot B_i.
	$$
	Now recall the expression \eqref{eq: partial derivative of J_u} for the partial derivative of $J$. The above inequality implies that
	$$
	\Big| \frac{\partial J_{\vu}}{\partial \vtau_i}(\vtau) - \frac{\partial J_{\vu'}}{\partial \vtau_i}(\vtau) \Big| \le 2B_i \cdot \|\vu - \vu'\|_2.
	$$
	For the case when $\vtau_i = 0$, the above formula holds because the limit formula holds. Therefore, the map $\vu \mapsto J_{\vu}$ is Lipschitz.
\end{proof}

\subsection{The Expected Distance to the Sum of Compact Sets} \label{subsec: expected distance to sum of sets}

Using the results in Lemma \ref{lem: distance to sum of sets}, we can study the expected distance to the sum of multiple sets.

\begin{lemma} \label{lem: expected distance to the sum of multiple sets}
  	Let $S_i$, $0 \le i \le k$, be some non-empty, compact, convex subsets of $\R^n$ that do not contain the origin. Suppose that $\|\vs_i\|_2 \le B_i$ for some $B_i > 0$ and for any $\vs_i \in S_i$, $0 \le i \le k$. Suppose that
	\begin{equation*}
	  	\vO \notin \sum_{i = 0}^{k} \vtau_i S_i, \ \ \textnormal{for any } \vtau \in \S^{k} \cap \R^{k+1}_{+}.
  	\end{equation*}
	Define the function $J : \R_+^{k + 1} \rightarrow \R$ by
	$$
	J(\vtau) \coloneqq \E \dist^2(\vg, \sum_{i = 0}^{k} \vtau_i S_i) = \E [J_{\vg}(\vtau)], \ \textnormal{for } \vtau = (\vtau_0, \vtau_1,\dots, \vtau_{k}) \in \R_+^{k + 1},
	$$
	where $\vg \sim N(\vO, \mI_n)$. The function $J$ is convex, continuous, and continuously differentiable in $\R_{+}^{k+1}$. It attains its minimum in a compact subset of $\R_+^{k+1}$. The differential of $J$ is
	\begin{equation} \label{eq: differential of J}
		\nabla J(\vtau) = \E [ \nabla J_{\vg}(\vtau)] \ \textnormal{for all } \vtau \in \R_{+}^{k+1}.
  	\end{equation}
	For $\vtau$ on the boundary of $\R_{+}^{k+1}$, we interpret the partial derivative $\frac{\partial J}{\partial \vtau_i}(\vtau)$ as the right partial derivative if $\vtau_i = 0$, i.e., 
	$$
	\frac{\partial J}{\partial \vtau_i} (\vtau) = \lim_{\epsilon \downarrow 0} \frac{{J}(\vtau_0, \dots, \vtau_{i-1},\epsilon, \dots, \vtau_k) - J(\vtau_0, \dots, \vtau_{i-1}, 0, \dots, \vtau_k)}{\epsilon}.
	$$
	Moreover, suppose that 
	\begin{equation} \label{eq: sets not equal}
	  \sum_{i = 0}^{k} \vtau_i S_i \neq \sum_{i = 0}^{k} \tilde{\vtau}_i S_i \quad \textnormal{for any }\vtau \neq \tilde{\vtau} \in \R^{k+1}_+.
  	\end{equation}
	Then the function $J(\vtau)$ is strictly convex, and attains its minimum at a unique point. 
\end{lemma}

\begin{proof}
  	There properties follow from the results in Lemma \ref{lem: distance to sum of sets}.

	\textit{Continuity.} We first consider the case when $\vtau \in \R_{++}^{k+1}$ and let $\vepsilon \in \R^{k+1}$. Note that by Jensen's inequality, we have
	$$
	\big| J(\vtau + \vepsilon) - J(\vepsilon) \big| = \Big| \E [ J_{\vg}(\vtau + \vepsilon) - J_{\vg}(\vepsilon)] \Big| \le \E \Big| [ J_{\vg}(\vtau + \vepsilon) - J_{\vg}(\vepsilon)] \Big|.
	$$
	Combining the bound for $\big| [ J_{\vg}(\vtau + \vepsilon) - J_{\vg}(\vepsilon)] \big|$ in \eqref{eq: difference of J}, we obtain
	\begin{align*}
	  	\big| J(\vtau + \vepsilon) - J(\vepsilon) \big|
		&\le \|\vepsilon\|_2^2 B^2 + 2 \|\vepsilon\|_2 B \cdot \big( \E \|\vg\|_2 + \|\vtau\|_2 B \big) \rightarrow 0 \ \textnormal{when }\vepsilon \rightarrow \vO.
	\end{align*}
	Similar argument holds as well when $\vtau$ is on the boundary of $\R_{+}^{k+1}$. Therefore, the function $J$ is continuous in $\R_+^{k+1}$.

	\textit{Convexity.} The convexity of the function $J$ comes from the convexity of the function $J_{\vg}$. In fact, take $\vtau, \tilde{\vtau} \in \R^{k+1}_+$ and let $\lambda_1, \lambda_2 \in \R_+$ and $\lambda_1 + \lambda_2 = 1$. The convexity of $J_{\vg}$ implies that
	$$
	J(\lambda_1 \vtau + \lambda_2 \tilde{\vtau}) = \E J_{\vg} (\lambda_1 \vtau + \lambda_2 \tilde{\vtau}) \le \E \big[ \lambda_1 J_{\vg}(\vtau) + \lambda_2 J_{\vg}(\tilde{\vtau}) \big] = \lambda_1 J(\vtau) + \lambda_2 J(\tilde{\vtau}).
	$$
	Thus, the function $J$ is convex in $\R^{k+1}_+$.

	\textit{Continuous differentiability.} The differentiability of $J$ is a direct consequence of the Dominated Convergence Theorem \cite[Corollary 5.9]{BARTLE1995}. To apply this theorem, note that for any $\vtau \in \R^{k+1}_+$, the function $J_{\vg}(\vtau)$ is integrable with respect to the Gaussian measure, since
	$$
	\E \big| J_{\vg}(\vtau) \big| = \E \inf_{\vs_i \in S_i, 0 \le i \le k} \big\|\vg - \sum_{i=0}^k \vtau_i \vs_i \big\|_2^2 \le \E \big( \|\vg\|_2 + \big\| \sum_{i=0}^k \vtau_i \vs_i \big\|_2 \big)^2 \le \big(\sqrt{n} + \|\vtau\|_2B\big)^2 < \infty,
	$$
	where in the first inequality, we have used the triangle inequality, and in the second inequality, we have used the bound in \eqref{eq: bound for sum of compact sets}. Moreover, the function $J_{\vg}$ is continuously differentiable, and the partial derivative $\frac{\partial J_{\vg}}{\partial \vtau_i} (\vtau)$ has the upper bound in \eqref{eq: bound for partial derivative}. Therefore, we can use the Dominated Convergence Theorem \cite[Corollary 5.9]{BARTLE1995}, which implies that the function $J$ is continuously differentiable, and the partial derivative is
	$$
	\frac{\partial J}{\partial \vtau_i}(\vtau) = \E \Big[ \frac{\partial J_{\vg}}{\partial \vtau_i}(\vtau) \Big] \ \textnormal{for all } \vtau \in \R_{+}^{k+1}.
	$$
	The differential formula \eqref{eq: differential of J} follows immediately.

	\textit{Attainment of minimum in a compact subset.} When $\|\vtau\|_2 b \ge \sqrt{n}$, we have
	\begin{align*}
	  J(\vtau) = \E [ J_{\vg}(\vtau) ] \ge \E \big[ J_{\vg}(\vtau) | \|\vg\|_2 \le \sqrt{n} \big] \cdot \P \big\{ \|\vg\|_2 \le \sqrt{n} \big\} \ge \frac{1}{2} \E \big[ (\|\vtau\|_2 b - \|\vg\|_2)^2 | \|\vg\|_2 \le \sqrt{n} \big] \ge \frac{1}{2} (\|\vtau\|_2 b - \sqrt{n})^2,
	\end{align*}
	where in the first inequality we have used the law of total expectation, and the second comes from \eqref{eq: lower bound for J_u} and the fact that the median of random variable $\|\vg\|_2$ does not exceed $\sqrt{n}$. Therefore, when $\|\vtau\|_2 \ge (1+\sqrt{2})\sqrt{n} / b$, we have
	$$
	J(\vtau) \ge \frac{1}{2}\big( (1+\sqrt{2})\sqrt{n} - \sqrt{n} \big)^2 = n = J(\vO).
	$$
	Since $J$ is convex and continuous, the minimum of $J$ must be attained in the compact set $\B\big(\vO, (1+\sqrt{2})\sqrt{n} / b\big) \cap \R^{k+1}_+$.

	\textit{Strict convexity.} We prove this point by contradiction. Suppose the condition \eqref{eq: sets not equal} holds, but $J$ is not strictly convex. Then by the definition of strict convexity, there exist $\vtau, \tilde{\vtau} \in \R^{k+1}_+$, $\vtau \neq \tilde{\vtau}$, and $\eta \in (0,1)$ such that
	\begin{equation} \label{eq: not strictly convex}
	  \E \big[ J_{\vg} \big(\eta \vtau + (1-\eta)\tilde{\vtau} \big) \big] = \eta \E J_{\vg}(\vtau) + (1-\eta) \E J_{\vg}(\tilde{\vtau}).
	\end{equation}
	In Lemma \ref{lem: distance to sum of sets}, we have shown that $J_{\vg}$ is convex, which means
	\begin{equation} \label{eq: relation resulted from convexity of J_g}
	  J_{\vg} \big(\eta \vtau + (1-\eta)\tilde{\vtau} \big) \le \eta J_{\vg}(\vtau) + (1-\eta) J_{\vg}(\tilde{\vtau}).
  	\end{equation}
	Therefore, the identity \eqref{eq: not strictly convex} holds if and only if the two sides of \eqref{eq: relation resulted from convexity of J_g} is equal almost surely with respect to the Gaussian measure. However, since $\vtau \neq \tilde{\vtau}$, by \eqref{eq: sets not equal}, the two sets $E_1 \coloneqq \sum_{i = 0}^{k} \vtau_i S_i$ and $ E_2 \coloneqq \sum_{i = 0}^{k} \tilde{\vtau}_i S_i$ are not identical. Thus, without loss of generality, we can find a point $\va \in E_1$ but $\va \notin E_2$. It follows that $\Pi_{E_1}(\va) = \va$. But since $E_2$ is compact, we have $\Pi_{E_2}(\va) \neq \va$, so we obtain $\Pi_{E_1}(\va) \neq \Pi_{E_2}(\va)$. Now, let $\vg = \va$, we have
	\begin{align} \label{eq: strict inequality}
	  \eta J_{\vg}(\vtau) + (1-\eta) J_{\vg}(\tilde{\vtau})
		&= \eta \|\vg - \Pi_{E_1}(\vg)\|_2^2 + (1-\eta) \|\vg - \Pi_{E_2}(\vg)\|_2^2
		> \big\| \eta \big(\vg - \Pi_{E_1}(\vg)\big) + (1-\eta)\big(\vg - \Pi_{E_2}(\vg)\big)\big\|_2^2 \notag \\
		&= \big\| \vg - \big(\eta \Pi_{E_1}(\vg) + (1-\eta)\Pi_{E_2}(\vg)\big)\big\|_2^2.
	\end{align}
	The strict inequality comes from the strict convexity of square function, the fact that $0 < \eta < 1$ and the fact that $\Pi_{E_1}(\vg) \neq \Pi_{E_2}(\vg)$. In addition, note that
	\begin{equation} \label{eq: projection in sum}
		\eta \Pi_{E_1}(\vg) + (1-\eta)\Pi_{E_2}(\vg) \in \eta E_1 + (1-\eta) E_2,
  	\end{equation}
	and that
	\begin{equation} \label{eq: put sets together}
	  \eta E_1 + (1-\eta) E_2 = \eta \sum_{i = 0}^{k} \vtau_i S_i + (1-\eta)\sum_{i = 0}^{k} \tilde{\vtau}_i S_i = \sum_{i=0}^k \big(\eta \vtau_i + (1-\eta) \tilde{\vtau}_i \big) S_i,
  	\end{equation}
	where the last identity results from \cite[Theorem 3.2]{ROCK1970}. Putting \eqref{eq: projection in sum} and \eqref{eq: put sets together} together, we obtain
	\begin{equation} \label{eq: projection in set}
	  \eta \Pi_{E_1}(\vg) + (1-\eta)\Pi_{E_2}(\vg) \in \sum_{i=0}^k \big(\eta \vtau_i + (1-\eta) \tilde{\vtau}_i \big) S_i.
	\end{equation}
	Substituting \eqref{eq: projection in set} into \eqref{eq: strict inequality}, we obtain 
	\begin{align*}
	  \eta J_{\vg}(\vtau) + (1-\eta) J_{\vg}(\tilde{\vtau}) 
		&> \big\| \vg - \big(\eta \Pi_{E_1}(\vg) + (1-\eta)\Pi_{E_2}(\vg)\big)\big\|_2^2 \ge \inf_{\vs_i \in S_i, 0 \le i \le k} \big\| \vg - \sum_{i = 0}^k \big(\eta \vtau_i + (1-\eta) \tilde{\vtau}_2 \big)\vs_i \big\|_2^2 \\
		&= J_{\vg}\big(\eta \vtau + (1-\eta) \tilde{\vtau}\big).
  	\end{align*}
	Moreover, it is easy to see that the map $\vg \mapsto J_{\vg}$ is continuous. Therefore, there exists some $\epsilon > 0$ such that when $\vg \in \B(\va, \epsilon)$, we have
	$$
	\eta J_{\vg}(\vtau) + (1-\eta) J_{\vg}(\tilde{\vtau}) > J_{\vg}\big(\eta \vtau + (1-\eta) \tilde{\vtau}\big).
	$$
	This contravenes \eqref{eq: not strictly convex}. 

	\textit{Attainment of minimum at a unique point.} We have shown that $J$ attains its minimum in the compact set $\B\big(\vO, (1+\sqrt{2})\sqrt{n} / b\big) \cap \R^{k+1}_+$. Now, since $J$ is strictly convex and continuous, it must attain its minimum at a unique point in $\B\big(\vO, (1+\sqrt{2})\sqrt{n} / b\big) \cap \R^{k+1}_+$.
\end{proof}

\subsection{Proof of Theorem \ref{th: calc_statistical_dimension_1}} \label{subsec: completion of proof of recipe 1}

Actually, Lemma \ref{lem: sum of sets}, Lemma \ref{lem: distance to sum of sets} and Lemma \ref{lem: expected distance to the sum of multiple sets} together almost prove our Theorem \ref{th: calc_statistical_dimension_1}, except that we do not show the conditions in Lemma \ref{lem: sum of sets} are satisfied. Thus, to prove Theorem \ref{th: calc_statistical_dimension_1}, it remains to show that $\vO \notin \sum_{i = 1}^k \vtau_i \cdot \partial f(\fxs)$ for any $\vtau \in \S^k \cap \R^{k+1}_+$. The following lemma confirms this point.

\begin{lemma} \label{lem: sum of subdifferentials not contain origin}
  	Suppose that for any $0 \le i \le k$, the function $f_i: \R^n \rightarrow \overline{\R}$ is a proper convex function, that
	\begin{equation} \label{eq: intersection of interiors not empty}
		\ri\big(D(f_0, \fxs) \big) \cap \ri\big(D(f_1, \fxs)\big) \cap \dots \cap \ri\big(D(f_k, \fxs)\big) \neq \emptyset,
  	\end{equation}
	and that the subdifferential $\partial f(\fxs)$ is non-empty and does not contain the origin. Then for any $\vtau \in \S^{k} \cap \R^{k+1}_+$, we have
	$$
	\vO \notin \sum_{i=0}^k \vtau_i \cdot \partial f(\fxs).
	$$
\end{lemma}

\begin{proof}
  	We prove by contradiction. Suppose condition \eqref{eq: intersection of interiors not empty} holds, but there exist a $\vtau \in \S^k \cap \R^{k+1}_+$ such that
  	$$
	\vO \in \sum_{i=0}^k \vtau_i \cdot \partial f(\fxs).
	$$
	So we can find $\vomega_i \in \partial f_i(\fxs)$, $0 \le i \le k$, satisfying $\sum_{i = 0}^k \vtau_i \vomega_i = \vO$. Now fix any $\fx \in \R^n$. By the definition of subdifferential, we know that
	$$
	f_i(\fx) - f_i(\fxs) \ge \left< \vomega_i, \fx - \fxs \right> \ \textnormal{for} \ 0 \le i \le k.
	$$
	Multiplying both sides by $\tau_i$ and taking the sum over $0 \le i \le k$, we obtain that
	\begin{equation} \label{eq: sum of functions be positive}
		\sum_{i = 0}^k \vtau_i \big( f_i(\fx) - f_i(\fxs) \big) \ge \left< \sum_{i=0}^k \vtau_i \vomega_i, \fx - \fxs \right> = 0.
  	\end{equation}
  	On the other hand, take any non-zero point $\vd \in \ri\big(D(f_0, \fxs) \big) \cap \ri\big(D(f_1, \fxs)\big) \cap \dots \cap \ri\big(D(f_k, \fxs)\big)$. Since $D(f_i,\fxs) = \cone\big(D_s(f_i,\fxs)\big)$, where $D_s(f_i, \fxs)$ is defined as
	$$
	D_s(f_i, \fxs) = \{ \vd: f_i(\fxs + \vd) \le f_i(\fxs) \},
	$$
	it follows from \cite[Corollary 6.8.1]{ROCK1970} that there exist numbers $t_0, \dots, t_k > 0$ such that
	$$
	t_i \vd \in \ri\big(D_s(f_i, \fxs)\big) \quad \textnormal{for any}\ 0 \le i \le k.
	$$
	Let $t = \min_{0\le i \le k} t_i$. The convexity of $f_i$ implies that the set $D_s(f_i,\fxs)$ is a convex set. Thus, we have
	$$
	t \vd = (1-\lambda_i) \cdot \vO + \lambda_i \cdot t_i \vd \in \ri \big(D_s(f_i,\fxs)\big) \quad \textnormal{for any} \ 0 \le i \le k,
	$$
  	where $\lambda_i = t / t_i \in (0,1]$. This point follows from \cite[Theorem 6.1]{ROCK1970}. In other words, 
	$$
	\fxs + t \vd \in \ri \{ \fx: f_i(\fx) \le f_i(\fxs)\} \quad \textnormal{for any} \ 0 \le i \le k.
	$$
	Now \cite[Theorem 7.6]{ROCK1970} tells us that the two set $\{ \fx: f_i(\fx) \le f_i(\fxs)\}$ and $\{ \fx: f_i(\fx) < f_i(\fxs)\}$ have the same closure and the same relative interior, so
	$$
	\fxs + t \vd \in \ri \{ \fx: f_i(\fx) < f_i(\fxs)\} \quad \textnormal{for any} \ 0 \le i \le k.
	$$
	As a result, we must have
	\begin{equation} \label{eq: strict inequality for every f}
		f_i(\fxs + t \vd) < f_i(\fxs) \quad \textnormal{for any} \ 0 \le i \le k.
  	\end{equation}
	Since $\vtau \in \S^{k} \cap \R^{k+1}_+$, some of the coordinates of $\vtau$ are positive, so \eqref{eq: strict inequality for every f} indicates that
	$$
	\sum_{i = 0}^k \vtau_i \big( f_i(\fxs + t \vd) - f_i(\fxs) \big) < 0.
	$$
	This contravenes \eqref{eq: sum of functions be positive}.
\end{proof}

\section{Proof of Theorem \ref{th: calc_statistical_dimension_2}} \label{app: proof for recipe 2}

In this section, we prove our Theorem \ref{th: calc_statistical_dimension_2}. The proof idea is essentially the same with that of Theorem \ref{th: calc_statistical_dimension_1}, either of which is inspired by \cite{AMEL2014}. Nevertheless, some of the details are different. For the sake of completeness, we include the detailed proof for Theorem \ref{th: calc_statistical_dimension_2}.

\subsection{Distance to the Sum of Sets} \label{subsec: dist to sum of unbouded sets}

Similar as in the proof for Theorem \ref{th: calc_statistical_dimension_1}, in this subsection, we study a simpler function, which describes the distance of a point to the sum of sets.

\begin{lemma} [Sum of sets] \label{lem: sum of unbounded sets}
  Let $S_i$, $0 \le i \le q$, be some non-empty, compact, convex subsets of $\R^n$ that do not contain the origin, and $K$ be a non-empty and convex subset of $\R^n$ that contains the origin. Suppose that $\|\vs_i\|_2 \le \tilde{B}_i$ for some $\tilde{B}_i > 0$ and for any $\vs_i \in S_i$, $0 \le i \le q$. Then there exists a number $\tilde{B} > 0$ such that
	\begin{equation} \label{eq: upper bound for sum of unbounded sets}
	  \Big\|\sum_{i = 0}^{q} \vtau_i \vs_i \Big\|_2 \le \tilde{B}, \ \textnormal{for any } \vtau \in \S^{q} \cap \R^{q+1}_{+}\ \textnormal{and} \ \vs_i \in S_i, \ 0 \le i \le q.
  	\end{equation}
	Moreover, suppose that
 	\begin{equation*} \label{eq: origin not in sum_2}
 	  	\vO \notin \overline{K} + \sum_{i = 0}^{q} \vtau_i S_i, \ \textnormal{for any } \vtau \in \S^{q} \cap \R^{q+1}_{+}.
   	\end{equation*}
	Then there exists a number $\tilde{b} > 0$ such that
 	\begin{equation} \label{eq: bound for sum of unbounded sets}
	  \Big\|\vkappa + \sum_{i = 0}^{q} \vtau_i \vs_i \Big\|_2 \ge \tilde{b}, \ \textnormal{for any } \vtau \in \S^{q} \cap \R^{q+1}_{+}, \vkappa \in K, \ \textnormal{and} \ \vs_i \in S_i, \ 0 \le i \le q.
   	\end{equation}
\end{lemma}

\begin{proof}
  	The proof of the upper bound \eqref{eq: upper bound for sum of unbounded sets} is the same with that of Lemma \ref{lem: sum of sets}, hence, we omit it. For the lower bound, we prove it by contradiction. Suppose that there does not exist $\tilde{b} > 0$ satisfying \eqref{eq: bound for sum of unbounded sets}, which implies that
 	\begin{equation} \label{eq: infimum is 0_2}
 		\inf_{\vtau \in \S^{q} \cap \R^{q+1}_{+}} \inf_{\vkappa \in K, \atop \vs_i \in S_i, 0 \le i \le q} \Big\|\vkappa + \sum_{i=0}^{q} \vtau_i \vs_i\Big\|_2 = 0.
   	\end{equation}
	Let's consider the function $r(\vtau) \coloneqq \inf_{\vkappa \in K, \vs_i \in S_i, 0 \le i \le q} \big\| \vkappa + \sum_{i=0}^{q} \vtau_i \vs_i\big\|_2$, where $\vtau \in \R^{q+1}_{+}$, and prove that it is continuous. To this end, let $\vtau, \tilde{\vtau} \in \R^{q+1}_{+}$. Since the sum of compact sets is compact \cite[Excercise 3(d), page 38]{RUDIN1991} and the sum of a compact set and a closed set is closed \cite[Exercise 3(e), page 38]{RUDIN1991}, as a result, both $\overline{K} + \sum_{i = 0}^{q} \vtau_i S_i$ and $\overline{K}+\sum_{i = 0}^{q} \tilde{\vtau}_i S_i$ are closed. It follows that there exist $\vkappa^{\star} \in \overline{K}$ and $\vs_i^{\star} \in S_i, 0 \le i \le q$ such that 
 	$$
	\Big\|\vkappa^{\star} + \sum_{i=0}^{q} \tilde{\vtau}_i \vs_i^{\star}\Big\|_2 = r(\tilde{\vtau}).
 	$$
 	Therefore, by the triangle inequality, we have
 	\begin{align} \label{eq: one side bound_2}
	  r(\vtau) - r(\tilde{\vtau}) 
	  &\le \Big\|\vkappa^{\star} + \sum_{i=0}^{q} \vtau_i \vs_i^{\star}\Big\|_2 - \Big\|\vkappa^{\star} + \sum_{i=0}^{q} \tilde{\vtau}_i \vs_i^{\star}\Big\|_2 \le \Big\|\Big( \vkappa^{\star} + \sum_{i=0}^{q} \vtau_i \vs_i^{\star} \Big) - \Big( \vkappa^{\star} + \sum_{i=0}^{q} \tilde{\vtau}_i \vs_i^{\star} \Big)\Big\|_2 \notag \\
	  &= \Big\|\sum_{i=0}^{q} (\vtau_i-\tilde{\vtau}_i) \vs_i^{\star}\Big\|_2 \le \|\vtau - \tilde{\vtau}\|_2 \cdot \tilde{B}. 
   	\end{align}
	The last inequality comes from the inequality \eqref{eq: upper bound for sum of unbounded sets}. By interchanging the roles of $\vtau$ and $\tilde{\vtau}$ in \eqref{eq: one side bound_2}, we obtain that
 	\begin{equation} \label{eq: two side bound_2}
	  \big|r(\vtau) - r(\tilde{\vtau})\big| \le \|\vtau - \tilde{\vtau}\|_2 \cdot \tilde{B},
   	\end{equation}
 	which implies that $r(\vtau)$ is Lipschitz function. Therefore, $r(\vtau)$ is continuous. Recall that a continuous function in a compact set must attain its infimum \cite[Theorem 4.16]{RUDIN1976}, therefore, \eqref{eq: infimum is 0_2} indicates that there exists a $\vtau \in \S^{q} \cap \R^{q+1}_{+}$ such that
 	\begin{equation} \label{eq: origin in sum_2}
 		\inf_{\vkappa \in K, \atop \vs_i \in S_i, 0 \le i \le q} \Big\|\vkappa + \sum_{i=0}^{q} \vtau_i \vs_i\Big\|_2 = 0.
   	\end{equation}
	Since $\overline{K} + \sum_{i = 0}^{q} \vtau_i S_i$ is closed, \eqref{eq: origin in sum_2} implies that $\vO \in \overline{K} + \sum_{i=0}^{q} \vtau_i S_i$. A contradiction. Therefore, there must exist some $\tilde{b} > 0$ satisfying \eqref{eq: bound for sum of unbounded sets}.
\end{proof}

Similar as before, we remind that when we write $\tilde{B}$ and $\tilde{b}$ hereafter, we always mean the numbers in \eqref{eq: upper bound for sum of unbounded sets} and \eqref{eq: bound for sum of unbounded sets}, respectively. Using Lemma \ref{lem: sum of unbounded sets}, we can study the properties of the function $J_{\vu}$, which is closely related with $J$.
 
\begin{lemma} [Distance to the sum of sets] \label{lem: distance to sum of unbounded sets}
   	Let $S_i$, $0 \le i \le q$, be some non-empty, compact, convex subsets of $\R^n$ that do not contain the origin, and $K$ be a non-empty and convex cone of $\R^n$ that contains the origin. Suppose that $\|\vs_i\|_2 \le \tilde{B}_i$ for some $\tilde{B}_i > 0$ and for any $\vs_i \in S_i$, $0 \le i \le q$. Moreover, suppose that
 	\begin{equation} \label{eq: origin not in sum 2_2}
 	  	\vO \notin \overline{K} + \sum_{i = 0}^{q} \vtau_i S_i, \ \ \textnormal{for any } \vtau \in \S^{q} \cap \R^{q+1}_{+}.
   	\end{equation}
 	Fix a point $\vu \in \R^n$, and define the function $J_{\vu} : \R_+^{q + 1} \rightarrow \R$ by
 	$$
 	J_{\vu}(\vtau) \coloneqq \dist^2(\vu, K + \sum_{i = 0}^{q} \vtau_i S_i),
 	$$
 	where $\vtau = (\vtau_0, \vtau_1,\dots, \vtau_{q}) \in \R_{+}^{q + 1}$.
 	Then $J_{\vu}(\vtau)$ has the following properties:
 	\begin{enumerate}
 	  \item
 		The function $J_{\vu}$ is convex and continuous.
 	  \item
 		The function $J_{\vu}$ satisfies the lower bound
 		\begin{equation} \label{eq: lower bound for J_u_2}
		  J_{\vu}(\vtau) \ge (\|\vtau\|_2 \tilde{b} - \|\vu\|_2)^2,\ \textnormal{when } \|\vtau\|_2 \ge \|\vu\|_2 / \tilde{b}.
 	  	\end{equation}
		In particular, $J_{\vu}$ attains its minimum in the compact subset $\B(0, 2\|\vu\|_2/\tilde{b}) \cap \R^{q+1}_{+}$.
 	  \item
 		The function $J_{\vu}$ is continuously differential, and the partial derivative is
		\begin{equation} \label{eq: partial derivative of J_2}
 			\frac{\partial J_{\vu}}{\partial \vtau_i}(\vtau) = -\frac{2}{\vtau_i}\left< \vu- \big(\bar{\vkappa} + \sum_{i=0}^{q} \vtau_i \bar{\vs}_i \big), \vtau_i \bar{\vs}_i\right> = -2\left< \vu- \big(\bar{\vkappa} + \sum_{i=0}^{q} \vtau_i \bar{\vs}_i \big), \bar{\vs}_i\right>
 		\end{equation}
 	for any $\bar{\vkappa} \in \overline{K}, \bar{\vs}_i \in S_i, 0 \le i \le q$ such that $\|\vu - ( \bar{\vkappa} + \sum_{i = 0}^{q} \vtau_i \bar{\vs}_i)\|_2^2 = J_{\vu}(\vtau)$. For $\vtau$ on the boundary of $\R^{q+1}_+$, we interpret the partial derivative $\frac{\partial J_{\vu}}{\partial \vtau_i}$ similarly as the right derivative if $\vtau_i = 0$, i.e., 
		$$
		\frac{\partial J_{\vu}}{\partial \vtau_i} (\vtau) = \lim_{\epsilon \downarrow 0} \frac{J_{\vu}(\vtau_0, \dots, \vtau_{i-1}, \epsilon, \dots, \vtau_q) - J_{\vu}(\vtau_0, \dots, \vtau_{i-1}, 0, \dots, \vtau_q)}{\epsilon}.
		$$
 	  \item
 		The partial derivative of $J_{\vu}$ has the following bound:
 		\begin{equation} \label{eq: bound for partial derivative_2}
		  \Big| \frac{\partial J_{\vu}}{\partial \vtau_i} (\vtau) \Big| \le 2\tilde{B}_i \big( \|\vu\|_2 + \|\vtau\|_2 \tilde{B} \big).
 	  	\end{equation}
 	  \item
		For any $\vtau \in \R^{q+1}_+$ and any $0 \le i \le q$, the map $\vu \mapsto \frac{\partial J_{\vu}}{\partial \vtau_i}(\vtau)$ is Lipschitz:
 		\begin{equation} \label{eq: partial derivative is Lipschitz_2}
		  \Big| \frac{\partial J_{\vu}}{\partial \vtau_i}(\vtau) - \frac{\partial J_{\vu'}}{\partial \vtau_i}(\vtau) \Big| \le 2\tilde{B}_i \cdot \|\vu - \vu'\|_2.
 		\end{equation}
 	\end{enumerate}
\end{lemma}
 
\begin{proof}
  	Lemma \ref{lem: distance to sum of unbounded sets} generalizes Lemma \ref{lem: distance to sum of sets} by allowing some of the sets to be unbounded or contain the origin.
 
 	\textit{Convexity.} Note that to prove the convexity of $J_{\vu}$, it is sufficient to prove that the function
 	$$
 	J^{\frac{1}{2}}_{\vu}(\vtau) \coloneqq \sqrt{J_{\vu}(\vtau)} = \dist (\vu, K + \sum_{i = 0}^{q} \vtau_i S_i)
 	$$
	is convex. To this end, fix any $\vtau, \tilde{\vtau} \in \R^{q + 1}_{+}$ and $\lambda_1, \lambda_2 \in \R_{+}$ satisfying $\lambda_1 + \lambda_2 = 1$. Since $K$ and $S_i$, $0 \le i \le q$, are convex sets, it from \cite[Theorem 3.2]{ROCK1970} that:
 	\begin{equation} \label{eq: convex decomposition of set_2}
	  K = \lambda_1 K + \lambda_2 K, \ (\lambda_1\vtau_i + \lambda_2 \tilde{\vtau}_i) S_i = \lambda_1\vtau_i S_i + \lambda_2 \tilde{\vtau}_i S_i \ \ \textnormal{for} \ \ 0 \le i \le q.
   	\end{equation}
	Then by the definition of $J_{\vu}^{\frac{1}{2}}$ and the triangle inequality, we have
 	\begin{align*} \label{eq: convexity of J^1/2_2}
	  	J^{\frac{1}{2}}_{\vu}(\lambda_1\vtau + \lambda_2\tilde{\vtau}) 
		&= \dist \big(\vu, K + \sum_{i = 0}^{q} (\lambda_1\vtau_i + \lambda_2 \tilde{\vtau}_i) S_i\big)
		= \dist \big(\vu, (\lambda_1 K + \lambda_2 K) + \sum_{i = 0}^{q} (\lambda_1\vtau_iS_i + \lambda_2 \tilde{\vtau}_i S_i)\big) \\
		&= \dist \big(\lambda_1\vu + \lambda_2\vu, \lambda_1 (K + \sum_{i = 0}^{q} \vtau_iS_i) + \lambda_2 (K + \sum_{i = 0}^{q} \tilde{\vtau}_i S_i)\big) \\
		&= \inf_{\vkappa \in K, \tilde{\vkappa} \in K \atop \vs_i \in S_i, \tilde{\vs}_i \in S_i, 0 \le i \le q} \Big\|\lambda_1\vu + \lambda_2\vu - \big[\lambda_1 (\vkappa + \sum_{i = 0}^{q} \vtau_i\vs_i) + \lambda_2 (\tilde{\vkappa} + \sum_{i = 0}^{q} \tilde{\vtau}_i \tilde{\vs}_i)\big] \Big\|_2 \\
		&\le \inf_{\vkappa \in K, \tilde{\vkappa} \in K \atop \vs_i \in S_i, \tilde{\vs}_i \in S_i, 0 \le i \le q} \lambda_1 \Big\|\vu - \big(\vkappa + \sum_{i = 0}^{q} \vtau_i\vs_i\big) \Big\|_2 + \lambda_2 \Big\| \vu - \big(\tilde{\vkappa} + \sum_{i = 0}^{q} \tilde{\vtau}_i \tilde{\vs}_i\big) \Big\|_2 \\
		&= \lambda_1 \cdot \dist\big(\vu, K + \sum_{i = 0}^{q} \vtau_i S_i\big) + \lambda_2 \cdot \dist\big(\vu, K + \sum_{i = 0}^{q} \tilde{\vtau}_i S_i\big) \\
		&=\lambda_1 J^{\frac{1}{2}}_{\vu}(\vtau) + \lambda_2 J^{\frac{1}{2}}_{\vu}(\tilde{\vtau}),
	\end{align*}
 	which implies that $J^{\frac{1}{2}}_{\vu}$ is convex. The convexity of $J_{\vu}$ follows immediately.
 
 	\textit{Continiuty.} We first consider the case when $\vtau \in \R^{q+1}_{+}$ and take any $\vepsilon \in \R^{q+1}$. To verify the continuity, note that
 	\begin{equation} \label{eq: dist to e+t_2}
 	  J^{\frac{1}{2}}_{\vu}(\vtau + \vepsilon) = \dist\Big(\vu, K + \sum_{i=0}^{q} (\vtau_i + \vepsilon_i) S_i \Big) = \inf_{\vkappa \in K \atop \vs_i \in S_i, 0 \le i \le q} \Big\| \vu - \big(\vkappa + \sum_{i=0}^{q} \vepsilon_i \vs_i + \sum_{i=0}^{q} \vtau_i \vs_i\big) \Big\|_2.
 	\end{equation}
 	The triangle inequality gives us that
 	\begin{equation} \label{eq: triangle to dist to e+t_2}
 		\Big\| \vu - \big(\vkappa  + \sum_{i=0}^{q} \vtau_i \vs_i\big) \Big\|_2 - \Big\| \sum_{i=0}^{q} \vepsilon_i \vs_i \Big\|_2 \le \Big\| \vu - \big(\vkappa + \sum_{i=0}^{q} \vepsilon_i \vs_i + \sum_{i=0}^{q} \vtau_i \vs_i\big) \Big\|_2 \le \Big\| \vu - \big(\vkappa  + \sum_{i=0}^{q} \vtau_i \vs_i\big) \Big\|_2 + \Big\| \sum_{i=0}^{q} \vepsilon_i \vs_i \Big\|_2.
   	\end{equation}
 	Putting \eqref{eq: dist to e+t_2} and \eqref{eq: triangle to dist to e+t_2} together, we obtain that
 	\begin{multline} \label{eq: triangle to dist_2}
 		\dist\Big(\vu, K + \sum_{i=0}^{q} \vtau_i S_i\Big) - \sup_{\vs_i \in S_i \atop 0 \le i \le q} \Big\| \sum_{i=0}^{q} \vepsilon_i \vs_i \Big\|_2 \le \dist\Big(\vu, K + \sum_{i=0}^{q} \vepsilon_i S_i + \sum_{i = 0}^{q} \vtau_i S_i \Big) \\
 		\le \dist\Big(\vu, K + \sum_{i=0}^{q} \vtau_i S_i\Big) + \sup_{\vs_i \in S_i \atop 0 \le i \le q} \Big\| \sum_{i=0}^{q} \vepsilon_i \vs_i \Big\|_2.
   	\end{multline}
 	Recalling the bound in \eqref{eq: upper bound for sum of unbounded sets}, we obtain from \eqref{eq: triangle to dist_2} that
 	$$
	\dist\Big(\vu, K + \sum_{i=0}^{q} \vtau_i S_i\Big) - \|\vepsilon\|_2 \tilde{B} \le \dist\Big(\vu, K + \sum_{i=0}^{q} \vepsilon_i S_i + \sum_{i = 0}^{q} \vtau_i S_i \Big) \le \dist\Big(\vu, K + \sum_{i=0}^{q} \vtau_i S_i\Big) + \|\vepsilon\|_2 \tilde{B}.
 	$$
 	In other words,
 	\begin{equation*} \label{eq: bound for dist_2}
	  J^{\frac{1}{2}}_{\vu}(\vtau) - \|\vepsilon\|_2 \tilde{B} \le J^{\frac{1}{2}}_{\vu}(\vtau + \vepsilon) \le J^{\frac{1}{2}}_{\vu}(\vtau) + \|\vepsilon\|_2 \tilde{B},
   	\end{equation*}
 	which implies that
	\begin{equation} \label{eq: differences of J}
	\|\vepsilon\|^2_2 \tilde{B}^2 - 2 \|\vepsilon\|_2 \tilde{B} \cdot J^{\frac{1}{2}}_{\vu} (\vtau) \le J_{\vu}(\vtau + \vepsilon) - J_{\vu}(\vepsilon) \le \|\vepsilon\|_2^2 \tilde{B}^2 + 2 \|\vepsilon\|_2 \tilde{B} \cdot J_{\vu}^{\frac{1}{2}} (\vtau).
  \end{equation}
	Moreover, select $\vkappa = \vO \in K$ and any $\vs_i \in S_i$, and we have
	$$
	J_{\vu}^{\frac{1}{2}} (\vtau) \le \big\|\vu - \vkappa - \sum_{i=0}^q \vtau_i \vs_i \big\|_2 \le \|\vu\|_2 + \|\vkappa\|_2 + \|\sum_{i=0}^q \vtau_i \vs_i\|_2 \le \|\vu\|_2 + \|\vtau\|_2 \cdot \tilde{B}.
	$$
	Substituting the above inequality into \eqref{eq: differences of J} yields
 	\begin{align} \label{eq: difference of J_2}
 	  	\big| J_{\vu}(\vtau + \vepsilon) - J_{\vu}(\vepsilon) \big|
		\le \|\vepsilon\|_2^2 \tilde{B}^2 + 2 \|\vepsilon\|_2 \tilde{B} \cdot J_{\vu}^{\frac{1}{2}} (\vtau) 
		 \le \|\vepsilon\|_2^2 \tilde{B}^2 + 2 \|\vepsilon\|_2 \tilde{B} \cdot \big( \|\vu\|_2 + \|\vtau\|_2 \tilde{B} \big).
 	\end{align}
 	Now it is easy to see that if $\epsilon \rightarrow \vO$, we have $\big| J_{\vu}(\vtau + \vepsilon) - J_{\vu}(\vepsilon) \big| \rightarrow 0$. Similar argument holds as well when $\vtau$ is on the boundary of $\R_{+}^{q+1}$. Therefore, the function $J_{\vu}$ is continuous in $\R_+^{q+1}$.
 
	\textit{Attainment of minimum.} Note that by Lemma \ref{lem: sum of unbounded sets}, we know that there exists a number $\tilde{b} > 0$ such that
 	$$
	\Big\| \vkappa + \sum_{i = 0}^{q} \vtau_i \vs_i \Big\|_2 \ge \tilde{b}, \ \textnormal{for any } \vtau \in \S^{q} \cap \R^{q+1}_{+}, \vkappa \in K,\ \textnormal{and} \ \vs_i \in S_i, \ 0 \le i \le q.
 	$$
 	Therefore, for any $\vtau \neq \vO$, by the triangle inequality,
 	\begin{align} \label{eq: expression of J_2}
 	  	J^{\frac{1}{2}}_{\vu}(\vtau) 
 		&= \dist(\vu, K + \sum_{i=0}^q \vtau_i S_i) = \inf_{\vkappa \in K, \atop \vs_i \in S_i, 0 \le i \le q} \big\|\vu - (\vkappa + \sum_{i=0}^{q} \vtau_i \vs_i) \big\|_2 \ge \inf_{\vkappa \in K, \atop \vs_i \in S_i, 0 \le i \le q} \big\|\vkappa + \sum_{i=0}^{q} \vtau_i \vs_i \big\|_2 - \|\vu\|_2 \notag \\
		&= \|\vtau\|_2 \cdot \inf_{\vkappa \in K, \atop \vs_i \in S_i, 0 \le i \le q} \Big\|\vkappa + \sum_{i=0}^{q} \frac{\vtau_i}{\|\vtau\|_2} \vs_i \Big\|_2 - \|\vu\|_2 \ge \|\vtau\|_2 \cdot \tilde{b} - \|\vu\|_2.
   	\end{align}
	The identity in the second line holds because $K$ is a convex cone, and the last inequality comes from \eqref{eq: bound for sum of unbounded sets}. Therefore, when $\|\vtau\|_2 \ge \|\vu\|_2 / \tilde{b}$, by squaring both sides of \eqref{eq: expression of J_2}, we obtain the bound
	$$
	J_{\vu}(\vtau) \ge \big( \|\vtau\|_2 \cdot \tilde{b} - \|\vu\|_2 \big)^2.
	$$
	Moreover, if $\|\vtau\|_2 \ge 2\|\vu\|_2 / \tilde{b}$, we have $J_{\vu}(\vtau) \ge \|\vu\|_2^2 \ge J_{\vu}(\vO) = \dist^2(\vu, K)$, since $K$ contains the origin. Then, it follows from the convexity and continuity of $J_{\vu}$ that the function $J_{\vu}$ must attain its minimum in the compact set $\B(\vO, 2\|\vu\|_2/\tilde{b}) \cap \R^{q+1}_{+}$.
 
 	\textit{Continuous differentiability in $\R^{q+1}_{++}$.} To prove that $J_{\vu}$ is continuously differentiable in $\R^{q+1}_{++}$, we need to show that the partial derivative $\partial J_{\vu} / \partial \vtau_i$ exists and is continuous, for each $\vtau_i$, $0 \le i \le q$. For this purpose, fix any $i$, $0 \le i \le q$, and define the function $\tilde{J}_{\vu}(\vtau_i)$ to be
 	\begin{equation} \label{eq: rewritten of J_2}
 	  	\tilde{J}_{\vu}(\vtau_i) \coloneqq J_{\vu}(\vtau) =  \dist^2(\vu, K + \sum_{i=0}^{q} \vtau_i S_i) = \dist^2(\vu, T + \vtau_i S_i) = \inf_{\vt \in T} \dist^2(\vu - \vt, \vtau_i S_i),
   \end{equation}
   where $T = \overline{K} + \sum_{0 \le j \le q, j \neq i}\vtau_i S_i$. Note that $T$ is closed since the sum of compact sets are compact \cite[Exercise 3(d), page 38]{RUDIN1991} and the sum of a compact set and a closed set is closed \cite[Exercise 3(e), page 38]{RUDIN1991}. Now define the function $g(\vtau_i, \vt) = \dist^2(\vu - \vt, \vtau_i S_i)$, $(\vtau_i, \vt) \in \R_{++} \times \R^n$. The function $g(\vtau_i, \vt)$ is continuously differentiable. To see this, first note that the function $\partial g / \partial \vtau_i$ exists, and takes the form
 	$$
 	\frac{\partial g}{\partial \vtau_i} = -\frac{2}{\vtau_i}\left< \vu- \vt - \Pi_{\vtau_i S_i}(\vu - \vt), \Pi_{\vtau_i S_i}(\vu- \vt)\right>.
 	$$
	Moreover, $\partial g / \partial \vtau_i$ is continuous \cite[Lemma C.1 (3)]{AMEL2014}. Next, the function $\tilde{g}(\vt) = \dist^2(\vu- \vt, \vtau_i S_i)$ is differentiable, and the differential is
 	$$
	\nabla \tilde{g}(\vt) = -2\big(\vu - \vt - \Pi_{\vtau_i S_i}(\vu - \vt)\big).
 	$$
	This point results from \cite[Theorem 2.26]{ROCK1998}. Furthermore, the projection onto a convex set is continuous \cite[Theorem 2.26]{ROCK1998}, hence, $\nabla \tilde{g}$ is a continuous function. It follows that $\partial g / \partial \vt_j$ is continuous for any $1 \le j \le n$. Therefore, we obtain that the function $g(\vtau_i, \vt)$ is continuously differentiable in $\R_{++} \times \R^n$. As a result, by \cite[Theorem 2.8]{SPIVAK1965}, $g(\vtau_i, \vt)$ is differentiable in $\R_{++} \times \R^n$, and the differential is
 	$$
 	\nabla g(\vtau_i, \vt) = \Big[-\frac{2}{\vtau_i}\left< \vu- \vt - \Pi_{\vtau_i S_i}(\vu - \vt), \Pi_{\vtau_i S_i}(\vu- \vt)\right>, -2\big(\vu - \vt - \Pi_{\vtau_i S_i}(\vu - \vt)\big)^T \Big]^T.
 	$$
 	The subdifferential of a differentiable function contains only the differential of the function \cite[Theorem 25.1]{ROCK1970}. Thus, the subdifferential of $g$ at $(\vtau_i, \vt)$ is
 	\begin{equation} \label{eq: subdifferential of g_2}
 		\partial g(\vtau_i, \vt) = \Big\{ \Big[-\frac{2}{\vtau_i}\left< \vu- \vt - \Pi_{\vtau_i S_i}(\vu - \vt), \Pi_{\vtau_i S_i}(\vu- \vt)\right>, -2\big(\vu - \vt - \Pi_{\vtau_i S_i}(\vu - \vt)\big)^T \Big]^T \Big\}.
 	\end{equation}
	Now select\footnote{Before we do such selection, we must argue that the infimum of $g(\vtau_i, \vt)$ can be attained over $\vt \in T$ for any fixed $\vtau_i > 0$. Actually, let $c > 0$ be a sufficiently large constant. Then when $\|\vt\|_2 \ge c$, we can make $\dist(\vu-\vt, \vtau_i S_i)$ be sufficiently large such that $\dist(\vu - \vt, \vtau_i S_i) \ge \dist(\vu, \vtau_i S_i) = g(\vtau_i, \vO)$, because $S_i$ is compact. Furthermore, the function $g(\vtau_i, \vt)$ is convex in $\vt$, because the distance function to a convex set is convex, and the composition of a convex function and an affine mapping is convex. Thus, the infimum of $g(\vtau_i, \vt)$ over $\vt \in T$ must be attained when $\|\vt\|_2 \le c$, i.e., when $\vt \in \B(\vO, c) \cap T$. Clearly, the set $\B(\vO, c) \cap T$ is compact. Thus, the continuity of $g(\vtau_i, \vt)$ in $\vt$ implies that the infimum must be attained at some point.} any $\bar{\vt} \in T$ such that $g(\vtau_i, \bar{\vt}) = \tilde{J}_{\vu}(\vtau_i)$. Let us confirm that $-\nabla \tilde{g}(\bar{\vt}) = 2\big(\vu - \bar{\vt} - \Pi_{\vtau_i S_i}(\vu - \bar{\vt})\big) \in N(\bar{\vt}; T)$, where $N(\bar{\vt}; T) \coloneqq \{ \vw \in \R^n: \left< \vw, \vt - \bar{\vt}\right> \le 0,\ \forall \, \vt \in T\}$, denotes the normal cone to $T$ at $\bar{\vt}$. To this end, let $\bar{\vs}_i \in S_i$ such that $\vtau_i \bar{\vs}_i = \Pi_{\vtau_i S_i}(\vu - \bar{\vt})$. Then by the definition of projection, it is not difficult to see that $\bar{\vt} = \Pi_{T}(\vu - \vtau_i \bar{\vs}_i)$. Thus, 
 	$$
	-\nabla \tilde{g}(\bar{\vt}) = 2\big(\vu - \bar{\vt} - \Pi_{\vtau_i S_i}(\vu - \bar{\vt})\big) = 2\big( \vu - \vtau_i \bar{\vs}_i - \Pi_{T}(\vu - \vtau_i \bar{\vs}_i) \big).
 	$$
 	By \cite[Theorem III.3.1.1]{HIRIART1993}, we know that
 	$$
	\left< \vu - \vtau_i \bar{\vs}_i - \Pi_{T}(\vu - \vtau_i \bar{\vs}_i), \vt - \Pi_{T}(\vu - \vtau_i \bar{\vs}_i)\right> \le 0, \ \textnormal{for any} \ \vt \in T.
 	$$
 	Therefore, we obtain that 
 	\begin{equation} \label{eq: subdifferential in normal cone_2}
	  -\nabla \tilde{g}(\bar{\vt}) = 2\big(\vu - \bar{\vt} - \Pi_{\vtau_i S_i}(\vu - \bar{\vt})\big) \in N(\bar{\vt}; T).
   	\end{equation}
 	Now we can give a conclusion about the subdifferential of $\tilde{J}_{\vu}$:
 	$$
 	\partial \tilde{J}_{\vu}(\vtau_i) = \Big\{ -\frac{2}{\vtau_i}\left< \vu- \bar{\vt} - \Pi_{\vtau_i S_i}(\vu - \bar{\vt}), \Pi_{\vtau_i S_i}(\vu- \bar{\vt})\right> \Big\}.
 	$$
 	This is a direct consequence of \cite[Example 2.59 and Theorem 2.61]{MORD2014}, \eqref{eq: subdifferential of g_2}, \eqref{eq: subdifferential in normal cone_2}, and the fact that $g(\vtau_i, \vt)$ is continuous. That the subdifferential of $\tilde{J}_{\vu}$ is a singleton implies $\tilde{J}_{\vu}$ is differentiable \cite[Theorem 25.1]{ROCK1970}, and the differential is
 	$$
 	\nabla \tilde{J}_{\vu}(\vtau_i) = -\frac{2}{\vtau_i}\left< \vu- \bar{\vt} - \Pi_{\vtau_i S_i}(\vu - \bar{\vt}), \Pi_{\vtau_i S_i}(\vu- \bar{\vt})\right>.
 	$$
 	The above formula is equivalent to that the partial derivative $\partial J_{\vu} / \partial \vtau_i$ exists, and takes the form
 	$$
 	\frac{\partial J_{\vu}}{\partial \vtau_i}(\vtau) = -\frac{2}{\vtau_i}\left< \vu- \bar{\vt} - \vtau_i \bar{\vs}_i, \vtau_i \bar{\vs}_i\right>,
 	$$
	for any $\bar{\vt} \in T, \bar{\vs}_i \in S_i$ such that $\|\vu - ( \bar{\vt} + \vtau_i \bar{\vs}_i)\|_2 = \dist(\vu, T+\vtau_i S_i)$. Since $T = \overline{K} + \sum_{0 \le j \le q, j \neq i} \vtau_i S_i$ is closed, we have
 	$$
   	\bar{\vt} = \bar{\vkappa} + \sum_{0 \le j \le q, j \neq i}\vtau_j \bar{\vs}_j, \ \textnormal{for some} \ \bar{\vkappa} \in \overline{K}, \ \bar{\vs}_j \in S_j, \ 0 \le j \le q, \ j \neq i.
 	$$
 	Therefore, the partial derivative $\partial J_{\vu} / \partial \vtau_i$ can be rewritten as
	$$
 		\frac{\partial J_{\vu}}{\partial \vtau_i}(\vtau) = -\frac{2}{\vtau_i}\left< \vu- \big(\bar{\vkappa} + \sum_{i=0}^{q} \vtau_i \bar{\vs}_i \big), \vtau_i \bar{\vs}_i\right> = -2\left< \vu- \big(\bar{\vkappa} + \sum_{i=0}^{q} \vtau_i \bar{\vs}_i \big), \bar{\vs}_i\right>
	$$
 	for any $\bar{\vkappa} \in \overline{K}, \bar{\vs}_i \in S_i, 0 \le i \le q$ such that $\|\vu - ( \bar{\vkappa} + \sum_{i = 0}^{q} \vtau_i \bar{\vs}_i)\|_2^2 = J_{\vu}(\vtau)$. It remains to prove that $\partial J_{\vu} / \partial \vtau_i$ is continuous in $\vtau_i$. Indeed, $\tilde{J}_{\vu}$ is a proper convex function, and is differential in $\R_{++}$. It follows from \cite[Theorem 25.5]{ROCK1970} that the gradient mapping $\nabla \tilde{J}_{\vu}$ is continuous in $\R_{++}$, which means that $\partial J_{\vu} / \partial \vtau_i$ is continuous in $\R_{++}$. Since for any $0 \le i \le q$, $\partial J_{\vu} / \partial \vtau_i$ exists and is continuous in $\R_{++}$, we obtain that $J_{\vu}$ is continuously differentiable in $\R^{q+1}_{++}$.
 
	\textit{Differentiablity at the boundary of $\R^{q+1}_{+}$ and its continuity.} The function $\tilde{J}_{\vu}$ is a closed proper convex function. It is continuous in $[0, +\infty]$ and continuously differentiable in $(0, +\infty)$. Hence, as a consequence of \cite[Theorem 24.1]{ROCK1970}, the right derivative at the origin exists and the limit formula holds. To study the continuity of the differential of $J_{\vu}$ at the boundary of $\R^{q+1}_{+}$, without loss of generality, we assume that $\vtau = (\vtau_0, \vtau_1, \dots, \vtau_l, \vtau_{l+1}, \dots, \vtau_q)$, where $\vtau_i > 0$ for $0 \le i \le l$ and $\vtau_i = 0$ for $l < i \le q$. Let $\vh = (\vh_0, \vh_1, \dots, \vh_l, \vh_{l+1}, \dots, \vh_{q})$, where $\vh_i \ge 0$ for $l < i \le q$. Similar as the proof for \cite[Theorem 2.8]{SPIVAK1965}, we have
 	\begin{align*}
 	  J_{\vu}(\vtau + \vh) - J_{\vu}(\vtau) = & J_{\vu}(\vtau_0 + \vh_0, \vtau_1, \dots, \vtau_q) - J_{\vu}(\vtau_0, \vtau_1, \dots, \vtau_q) \\
 	                                          & +J_{\vu}(\vtau_0 + \vh_0, \vtau_1 + \vh_1, \vtau_2, \dots, \vtau_q) - J_{\vu}(\vtau_0 + \vh_0, \vtau_1, \vtau_2, \dots, \vtau_q) \\
 											  & + \dots \\
 											  & +J_{\vu}(\vtau_0 + \vh_0, \vtau_1 + \vh_1, \dots, \vtau_{q-1} + \vh_{q-1}, \vtau_q + \vh_q) - J_{\vu}(\vtau_0 + \vh_0, \vtau_1 + \vh_1, \dots, \vtau_{q-1} + \vh_{q-1}, \vtau_q).
 	\end{align*}
 	Let us look at the first term $J_{\vu}(\vtau_0 + \vh_0, \vtau_1, \dots, \vtau_q) - J_{\vu}(\vtau_0, \vtau_1, \dots, \vtau_q)$ first. By the mean-value theorem, we know that there exists some $\vb_0$ between $\vtau_0$ and $\vtau_0 + \vh_0$ such that
 	$$
 	J_{\vu}(\vtau_0 + \vh_0, \vtau_1, \dots, \vtau_q) - J_{\vu}(\vtau_0, \vtau_1, \dots, \vtau_q) = \frac{\partial J_{\vu}}{\partial \vtau_0}(\vb_0, \vtau_1, \dots, \vtau_q) \cdot \vh_0.
 	$$
 	Similarly, for the $i$-th term, there exists some $\vb_{i-1}$ between $\vtau_{i-1}$ and $\vtau_{i-1} + \vh_{i-1}$ such that
 	\begin{multline*}
 	  	J_{\vu}(\vtau_0 + \vh_0, \dots, \vtau_{i-2} + \vh_{i-2}, \vtau_{i-1} + \vh_{i-1}, \vtau_i, \dots, \vtau_q) - J_{\vu}(\vtau_0 + \vh_0, \dots, \vtau_{i-2} + \vh_{i-2}, \vtau_{i-1}, \vtau_i, \dots, \vtau_q) \\
 		= \frac{\partial J_{\vu}}{\partial \vtau_{i-1}}(\vtau_0 + \vh_0, \dots, \vtau_{i-2} + \vh_{i-2}, \vb_{i-1}, \vtau_i, \dots, \vtau_q) \cdot \vh_{i-1}.
   	\end{multline*}
 	Then,
 	\begin{align*}
 	  &\lim_{\vh_i \rightarrow 0, 0 \le i \le l, \atop \vh_i \rightarrow 0^+, l < i \le q} \frac{|J_{\vu}(\vtau + \vh) - J_{\vu}(\vtau) - \sum_{i=0}^q \frac{\partial J_{\vu}}{\partial \vtau_i} \cdot \vh_i|}{\|\vh\|_2} \\
 	  &\hspace*{120pt}= \lim_{\vh_i \rightarrow 0, 0 \le i \le l, \atop \vh_i \rightarrow 0^+, l < i \le q} \frac{ \Big|\sum_{i=0}^q \big[\frac{\partial J_{\vu}}{ \partial \vtau_i} (\vtau_0, \dots, \vb_i, \dots, \vtau_q) - \frac{\partial J_{\vu}}{\partial \vtau_i}(\vtau_0, \dots, \vtau_i, \dots\vtau_q) \big] \cdot \vh_i \Big|}{\|\vh\|_2} \\
 	  &\hspace*{120pt}\le \lim_{\vh_i \rightarrow 0, 0 \le i \le l, \atop \vh_i \rightarrow 0^+, l < i \le q}  \sum_{i=0}^q  \Big|\frac{\partial J_{\vu}}{ \partial \vtau_i} (\vtau_0, \dots, \vb_i, \dots, \vtau_q) - \frac{\partial J_{\vu}}{\partial \vtau_i}(\vtau_0, \dots, \vtau_i, \dots\vtau_q) \Big| \cdot \frac{|\vh_i|}{\|\vh\|_2} \\
 	  &\hspace*{120pt}\le \lim_{\vh_i \rightarrow 0, 0 \le i \le l, \atop \vh_i \rightarrow 0^+, l < i \le q}  \sum_{i=0}^q  \Big|\frac{\partial J_{\vu}}{ \partial \vtau_i} (\vtau_0, \dots, \vb_i, \dots, \vtau_q) - \frac{\partial J_{\vu}}{\partial \vtau_i}(\vtau_0, \dots, \vtau_i, \dots\vtau_q) \Big| \\
 	  &\hspace*{120pt}= \lim_{\vb_i \rightarrow \vtau_i, 0 \le i \le l, \atop \vb_i \rightarrow \vtau_i^+, l < i \le q}  \sum_{i=0}^q  \Big|\frac{\partial J_{\vu}}{ \partial \vtau_i} (\vtau_0, \dots, \vb_i, \dots, \vtau_q) - \frac{\partial J_{\vu}}{\partial \vtau_i}(\vtau_0, \dots, \vtau_i, \dots\vtau_q) \Big| \\
 	  &\hspace*{120pt} = 0.
 	\end{align*}
 	The last identity holds because the partial derivative $\frac{\partial J_{\vu}}{ \partial \vtau_i}$ is continuous in $[0, +\infty)$.
 
 	\textit{Bound for the partial derivative.} Using the Cauchy-Schwarz inequality to \eqref{eq: partial derivative of J_2}, we obtain that
 	\begin{equation} \label{eq: bound for derivative_2}
 		\Big| \frac{\partial J_{\vu}}{\partial \vtau_i}(\vtau) \Big| \le 2 \big\| \vu- \big(\bar{\vkappa} + \sum_{i=0}^{q} \vtau_i \bar{\vs}_i \big)\big\|_2 \cdot \|\bar{\vs}_i\|_2.
   	\end{equation}
 	Since $\bar{\vkappa}$ and $\bar{\vs}_i$ satisfy that $\bar{\vkappa} \in \overline{K}$, $\bar{\vs}_i \in S_i, 0 \le i \le q$, and $\|\vu - ( \bar{\vkappa} + \sum_{i = 0}^{q} \vtau_i \bar{\vs}_i)\|_2^2 = J_{\vu}(\vtau)$, it holds that for any $\tilde{\vkappa} \in \overline{K}$ and $\tilde{\vs}_i \in S_i$,
 	$$
 	\big\| \vu - \big( \bar{\vkappa} + \sum_{i = 0}^{q} \vtau_i \bar{\vs}_i\big) \big\|_2 \le \big\| \vu - \big( \tilde{\vkappa} + \sum_{i = 0}^{q} \vtau_i \tilde{\vs}_i \big) \big\|_2.
 	$$
 	Since $K$ is a convex cone containing the origin, we can set $\tilde{\vkappa} = \vO$ and obtain
 	$$
	\big\| \vu - \big( \bar{\vkappa} + \sum_{i = 0}^{q} \vtau_i \bar{\vs}_i\big) \big\|_2 \le \|\vu\|_2 + \| \sum_{i = 0}^{q} \vtau_i \tilde{\vs}_i \|_2 \le \|\vu\|_2 + \|\vtau\|_2 \cdot \tilde{B},
 	$$
 	where we have used the triangle inequality and \eqref{eq: upper bound for sum of unbounded sets}. Substituting it into \eqref{eq: bound for derivative_2} yields the desired result
 	$$
	\Big| \frac{\partial J_{\vu}}{\partial \vtau_i}(\vtau) \Big| \le 2 \tilde{B}_i\big( \|\vu\|_2 + \|\vtau\|_2 \tilde{B}).
 	$$
 
 	\textit{Lipschitz property.} Fix any $\vtau \in \R^{q+1}_+$ satisfying $\vtau_i > 0$. We make use of \cite[Theorem III.3.1.1]{HIRIART1993} to obtain that
 	$$
 	\left<\vu - (\bar{\vkappa} + \sum_{j=0}^q \vtau_j \bar{\vs}_j), \bar{\vkappa} + \sum_{j=0}^q \vtau_j \bar{\vs}_j\right> \ge \left<\vu - (\bar{\vkappa} + \sum_{j=0}^q \vtau_j \bar{\vs}_j), \bar{\vkappa} + \vtau_i \vs_i + \sum_{0 \le j \le q, \atop j \neq i} \vtau_j \bar{\vs}_j\right> \quad \textnormal{for any} \ \vs_i \in S_i,
 	$$
	where $\bar{\vkappa} \in \overline{K}, \bar{\vs}_i \in S_i, 0 \le i \le q$, satisfy $\|\vu - (\bar{\vkappa} + \sum_{i = 0}^{q} \vtau_i \bar{\vs}_i)\|_2 = \dist(\vu, K + \sum_{i = 0}^{q} \vtau_i S_i)$. Simplifying the above inequality yields
 	$$
 	\left<\vu - (\bar{\vkappa} + \sum_{j=0}^q \vtau_j \bar{\vs}_j), \bar{\vs}_i\right> \ge \left<\vu - (\bar{\vkappa} + \sum_{j=0}^q \vtau_j \bar{\vs}_j), \vs_i\right> \quad \textnormal{for any} \ \vs_i \in S_i.
 	$$
 	Therefore, for any $\vu, \vu' \in \R^n$,
 	\begin{align} \label{eq: lipschitz and non-expansive_2}
 	  \left<\vu - (\bar{\vkappa} + \sum_{j=0}^q \vtau_j \bar{\vs}_j), \bar{\vs}_i\right> - \left<\vu' - (\bar{\vkappa}' + \sum_{j=0}^q \vtau_j \bar{\vs}_j'), \bar{\vs}_i'\right> 
 	  &\le \left< \Big[\vu - (\bar{\vkappa} + \sum_{j=0}^q \vtau_j \bar{\vs}_j)\Big] - \Big[\vu' - (\bar{\vkappa}' + \sum_{j=0}^q \vtau_j \bar{\vs}_j')\Big], \bar{\vs}_i\right> \notag \\
 	  &\le \big\| (\mI - \Pi_E )(\vu) - (\mI - \Pi_E )(\vu') \big\|_2 \cdot \|\bar{\vs}_i\|_2 \notag \\
	  &\le \| \vu - \vu' \|_2 \cdot \tilde{B}_i,
 	\end{align}
 	where $\Pi_E(\vu)$ denotes the projection of $\vu$ onto the set $E \coloneqq \overline{K} + \sum_{i=0}^q \vtau_i S_i$. In the second inequality, we have used the Cauchy-Schwarz inequality, and the last inequality comes from the fact that the map $\mI - \Pi_E$ is non-expansive with respect to the Euclidean norm \cite[pp. 275]{AMEL2014}. Interchanging the roles of $\vu$ and $\vu'$ in \eqref{eq: lipschitz and non-expansive_2}, we obtain that
 	$$
	\Big| \left<\vu - (\bar{\vkappa} + \sum_{j=0}^q \vtau_j \bar{\vs}_j), \bar{\vs}_i\right> - \left<\vu' - (\bar{\vkappa}' + \sum_{j=0}^q \vtau_j \bar{\vs}_j'), \bar{\vs}_i'\right> \Big| \le \| \vu - \vu' \|_2 \cdot \tilde{B}_i.
 	$$
 	Now recall the expression \eqref{eq: partial derivative of J_2} for the partial derivative of $J$. The above inequality implies that
 	$$
	\Big| \frac{\partial J_{\vu}}{\partial \vtau_i}(\vtau) - \frac{\partial J_{\vu'}}{\partial \vtau_i}(\vtau) \Big| \le 2 \tilde{B}_i \cdot \|\vu - \vu'\|_2.
 	$$
	For the case when $\vtau_i = 0$, the above formula holds because the limit formula holds. Therefore, the map $\vu \mapsto J_{\vu}$ is Lipschitz.
\end{proof}

\subsection{The Expected Distance to the Sum of Sets} \label{subsec: expected distance to sum of unbounded sets}
 
 Using the results in Lemma \ref{lem: distance to sum of unbounded sets}, we can study the expected distance to the sum of sets.
 
\begin{lemma} \label{lem: expected distance to the sum of unbounded sets}
   	Let $S_i$, $0 \le i \le q$, be some non-empty, compact, convex subsets of $\R^n$ that do not contain the origin, and $K$ be a non-empty and convex cone of $\R^n$ that contains the origin. Suppose that $\|\vs_i\|_2 \le \tilde{B}_i$ for some $\tilde{B}_i > 0$ and for any $\vs_i \in S_i$, $0 \le i \le q$. Moreover, suppose that
 	\begin{equation*}
 	  	\vO \notin \overline{K} + \sum_{i = 0}^{q} \vtau_i S_i, \ \ \textnormal{for any } \vtau \in \S^{q} \cap \R^{q+1}_{+}.
   	\end{equation*}
 	Define the function $J : \R_+^{q + 1} \rightarrow \R$ by
 	$$
 	J(\vtau) \coloneqq \E \dist^2(\vg, K + \sum_{i = 0}^{q} \vtau_i S_i) = \E [J_{\vg}(\vtau)], \ \textnormal{for } \vtau = (\vtau_0, \vtau_1,\dots, \vtau_{q}) \in \R_+^{q + 1},
 	$$
	where $\vg \sim N(\vO, \mI_n)$. The function $J$ is convex, continuous and continuously differentiable in $\R_{+}^{q+1}$. It attains its minimum in a compact subset of $\R^{q+1}_+$. Furthermore, 
 	\begin{equation} \label{eq: differential of J_2}
 		\nabla J(\vtau) = \E [ \nabla J_{\vg}(\vtau)] \ \textnormal{for all } \vtau \in \R_{+}^{q+1}.
   	\end{equation}
 	For $\vtau$ on the boundary of $\R_{+}^{q+1}$, we interpret the partial derivative $\frac{\partial J}{\partial \vtau_i}(\vtau)$ as the right partial derivative if $\vtau_i = 0$, i.e., 
	$$
	\frac{\partial J}{\partial \vtau_i} (\vtau) = \lim_{\epsilon \downarrow 0} \frac{{J}(\vtau_0, \dots, \vtau_{i-1}, \epsilon, \dots, \vtau_k) - J(\vtau_0, \dots, \vtau_{i-1}, 0, \dots, \vtau_k)}{\epsilon}.
	$$
	Moreover, suppose that 
	\begin{equation} \label{eq: sets not equal_2}
	  \overline{K} + \sum_{i = 0}^{q} \vtau_iS_i \neq \overline{K} + \sum_{i = 0}^{q} \tilde{\vtau}_i S_i \quad \textnormal{for any }\vtau \neq \tilde{\vtau} \in \R^{q+1}_+,
  	\end{equation}
	then the function $J(\vtau)$ is strictly convex, and attains its minimum at a unique point. 
\end{lemma}
 
\begin{proof}
  	There properties follow from the results in Lemma \ref{lem: distance to sum of unbounded sets}. The proof is similar as that for Lemma \ref{lem: expected distance to the sum of multiple sets}. But for the sake of completeness, we present the whole proof.
 
 	\textit{Continuity.} We first consider the case when $\vtau \in \R_{+}^{q+1}$ and take any $\vepsilon \in \R^{q+1}$. Note that by Jensen's inequality, we have
 	$$
 	\big| J(\vtau + \vepsilon) - J(\vepsilon) \big| = \Big| \E [ J_{\vg}(\vtau + \vepsilon) - J_{\vg}(\vepsilon)] \Big| \le \E \Big| [ J_{\vg}(\vtau + \vepsilon) - J_{\vg}(\vepsilon)] \Big|.
 	$$
 	Now combining the bound for $\big| [ J_{\vg}(\vtau + \vepsilon) - J_{\vg}(\vepsilon)] \big|$ in \eqref{eq: difference of J_2}, we obtain
 	\begin{align*}
 	  	\big| J(\vtau + \vepsilon) - J(\vepsilon) \big|
 		&\le \|\vepsilon\|_2^2 \tilde{B}^2 + 2 \|\vepsilon\|_2 \tilde{B} \cdot \big( \E \|\vg\|_2 + \|\vtau\|_2 \tilde{B} \big) \rightarrow 0 \ \textnormal{when }\vepsilon \rightarrow \vO.
 	\end{align*}
 	Similar argument holds as well when $\vtau$ is on the boundary of $\R_{+}^{q+1}$. Therefore, the function $J$ is continuous in $\R_+^{q+1}$.
 
 	\textit{Convexity.} The convexity of the function $J$ comes from the convexity of the function $J_{\vg}$. In fact, take $\vtau, \tilde{\vtau} \in \R^{q+1}_+$ and let $\lambda_1, \lambda_2 \in \R_+$ and $\lambda_1 + \lambda_2 = 1$. The convexity of $J_{\vg}$ implies that
	$$
	J(\lambda_1 \vtau + \lambda_2 \tilde{\vtau}) = \E J_{\vg} (\lambda_1 \vtau + \lambda_2 \tilde{\vtau}) \le \E \big[ \lambda_1 J_{\vg}(\vtau) + \lambda_2 J_{\vg}(\tilde{\vtau}) \big] = \lambda_1 J(\vtau) + \lambda_2 J(\tilde{\vtau}).
	$$
	Thus, the function $J$ is convex in $\R^{q+1}_+$.
 
 	\textit{Continuous differentiability.} The differentiability of $J$ is a direct consequence of the Dominated Convergence Theorem \cite[Corollary 5.9]{BARTLE1995}. To apply this theorem, note that for any $\vtau \in \R^{q+1}_+$, the function $J_{\vg}(\vtau)$ is integrable with respect to the Gaussian measure, since
 	$$
	\E \big| J_{\vg}(\vtau) \big| = \E \inf_{\vkappa \in K \atop \vs_i \in S_i, 0 \le i \le q} \big\|\vg - \big( \vkappa + \sum_{i=0}^q \vtau_i \vs_i \big) \big\|_2^2 \le \E \big( \|\vg\|_2 + \big\| \sum_{i=0}^q \vtau_i \vs_i \big\|_2 \big)^2 \le \big(\sqrt{n} + \|\vtau\|_2 \tilde{B}\big)^2 < \infty,
 	$$
 	where in the first inequality, we have used the triangle inequality and the fact that $K$ contains the origin. Moreover, the function $J_{\vg}$ is continuously differentiable, and the partial derivative $\frac{\partial J_{\vg}}{\partial \vtau_i} (\vtau)$ has the upper bound in \eqref{eq: bound for partial derivative_2}. Therefore, we can use the Dominated Convergence Theorem \cite[Corollary 5.9]{BARTLE1995}, which implies that the function $J$ is continuously differentiable, and the partial derivative is
 	$$
 	\frac{\partial J}{\partial \vtau_i}(\vtau) = \E \Big[ \frac{\partial J_{\vg}}{\partial \vtau_i}(\vtau) \Big] \ \textnormal{for all } \vtau \in \R_{+}^{q+1}.
 	$$
 	The differential formula \eqref{eq: differential of J_2} follows immediately.
 
	\textit{Attainment of minimum in a compact set.} When $\|\vtau\|_2 \tilde{b} \ge \sqrt{n}$, we have
 	\begin{align*}
	  J(\vtau) = \E [ J_{\vg}(\vtau) ] \ge \E \big[ J_{\vg}(\vtau) | \|\vg\|_2 \le \sqrt{n} \big] \cdot \P \big\{ \|\vg\|_2 \le \sqrt{n} \big\} \ge \frac{1}{2} \E \big[ (\|\vtau\|_2 \tilde{b} - \|\vg\|_2)^2 | \|\vg\|_2 \le \sqrt{n} \big] \ge \frac{1}{2} (\|\vtau\|_2 \tilde{b} - \sqrt{n})^2,
 	\end{align*}
	where in the first inequality we have used the law of total expectation, and the second comes from \eqref{eq: lower bound for J_u_2} and the fact that the median of random variable $\|\vg\|_2$ does not exceed $\sqrt{n}$. Therefore, when $\|\vtau\|_2 \ge (1+\sqrt{2})\sqrt{n}/\tilde{b}$, we have
	$$
	J(\vtau) \ge \frac{1}{2}\big( (1+\sqrt{2})\sqrt{n} - \sqrt{n} \big)^2 = n = J(\vO).
	$$
	Since $J$ is convex and continuous, the minimum of $J$ must be attained in the compact set $\B\big(\vO, (1+\sqrt{2})\sqrt{n}/\tilde{b}\big) \cap \R^{q+1}_+$.

	\textit{Strict convexity.} We prove this point by contradiction. Suppose that the condition \eqref{eq: sets not equal_2} holds, but $J$ is not strictly convex. Then by the definition of strict convexity, there exist $\vtau, \tilde{\vtau} \in \R^{q+1}_+$, $\vtau \neq \tilde{\vtau}$, and $\eta \in (0,1)$ such that
	\begin{equation} \label{eq: not strictly convex_2}
	  \E \big[ J_{\vg} \big(\eta \vtau + (1-\eta)\tilde{\vtau} \big) \big] = \eta \E J_{\vg}(\vtau) + (1-\eta) \E J_{\vg}(\tilde{\vtau}).
	\end{equation}
	Recall that in Lemma \ref{lem: distance to sum of unbounded sets}, we have shown that $J_{\vg}$ is convex, which means
	\begin{equation} \label{eq: relation resulted from convexity of J_g_2}
	  J_{\vg} \big(\eta \vtau + (1-\eta)\tilde{\vtau} \big) \le \eta J_{\vg}(\vtau) + (1-\eta) J_{\vg}(\tilde{\vtau}).
  	\end{equation}
	Therefore, the identity \eqref{eq: not strictly convex_2} holds if and only if the two sides of \eqref{eq: relation resulted from convexity of J_g_2} is equal almost surely with respect to the Gaussian measure. However, since $\vtau \neq \tilde{\vtau}$, by \eqref{eq: sets not equal_2}, the two sets $E_1 \coloneqq \overline{K} + \sum_{i = 0}^{k} \vtau_i S_i$ and $ E_2 \coloneqq \overline{K} + \sum_{i = 0}^{k} \tilde{\vtau}_i S_i$ are not identical. Thus, without loss of generality, we can find a point $\va \in E_1$ but $\va \notin E_2$. Then $\Pi_{E_1}(\va) = \va$. But since $E_2$ is closed, so $\Pi_{E_2}(\va) \neq \va$. Thus, $\Pi_{E_1}(\va) \neq \Pi_{E_2}(\va)$. Let $\vg = \va$, we have
	\begin{align} \label{eq: strict inequality_2}
	  \eta J_{\vg}(\vtau) + (1-\eta) J_{\vg}(\tilde{\vtau})
		&= \eta \|\vg - \Pi_{E_1}(\vg)\|_2^2 + (1-\eta) \|\vg - \Pi_{E_2}(\vg)\|_2^2
		> \big\| \eta \big(\vg - \Pi_{E_1}(\vg)\big) + (1-\eta)\big(\vg - \Pi_{E_2}(\vg)\big)\big\|_2^2 \notag \\
		&= \big\| \vg - \big(\eta \Pi_{E_1}(\vg) + (1-\eta)\Pi_{E_2}(\vg)\big)\big\|_2^2.
	\end{align}
	The strict inequality comes from the strict convexity of square function, the fact that $0 < \eta < 1$ and the fact that $\Pi_{E_1}(\vg) \neq \Pi_{E_2}(\vg)$. In addition, note that
	\begin{equation} \label{eq: projection in sum_2}
		\eta \Pi_{E_1}(\vg) + (1-\eta)\Pi_{E_2}(\vg) \in \eta E_1 + (1-\eta) E_2,
  	\end{equation}
	and that
	\begin{equation} \label{eq: put sets together_2}
	  \eta E_1 + (1-\eta) E_2 = \eta \Big[ \overline{K}+ \sum_{i = 0}^{q} \vtau_i S_i \Big] + (1-\eta) \Big[ \overline{K} + \sum_{i = 0}^{q} \tilde{\vtau}_i S_i \Big] = \overline{K} + \sum_{i=0}^q \big(\eta \vtau_i + (1-\eta) \tilde{\vtau}_i \big) S_i,
  	\end{equation}
	where we have used \cite[Theorem 3.2]{ROCK1970}. Putting \eqref{eq: projection in sum_2} and \eqref{eq: put sets together_2} together, we get
	\begin{equation} \label{eq: projection in set_2}
	  \eta \Pi_{E_1}(\vg) + (1-\eta)\Pi_{E_2}(\vg) \in \overline{K} +  \sum_{i=0}^q \big(\eta \vtau_i + (1-\eta) \tilde{\vtau}_i \big) S_i.
	\end{equation}
	Substituting \eqref{eq: projection in set_2} into \eqref{eq: strict inequality_2}, we obtain 
	\begin{align*}
	  \eta J_{\vg}(\vtau) + (1-\eta) J_{\vg}(\tilde{\vtau}) 
	  &> \big\| \vg - \big(\eta \Pi_{E_1}(\vg) + (1-\eta)\Pi_{E_2}(\vg)\big)\big\|_2^2 \ge \inf_{\vkappa \in \overline{K}, \atop\vs_i \in S_i, 0 \le i \le q} \big\| \vg - \big[ \vkappa + \sum_{i = 0}^q \big(\eta \vtau_i + (1-\eta) \tilde{\vtau}_2 \big)\vs_i \big] \big\|_2^2 \\
	  &= J_{\vg}\big(\eta \vtau + (1-\eta) \tilde{\vtau}\big).
  	\end{align*}
	Moreover, it is easy to see that the map $\vg \mapsto J_{\vg}$ is continuous. Therefore, there exists some $\epsilon > 0$ such that when $\vg \in \B(\va, \epsilon)$, we have
	$$
	\eta J_{\vg}(\vtau) + (1-\eta) J_{\vg}(\tilde{\vtau}) > J_{\vg}\big(\eta \vtau + (1-\eta) \tilde{\vtau}\big).
	$$
	This contravenes \eqref{eq: not strictly convex_2}. 

	\textit{Attainment of minimum at a unique point.} We have shown that $J$ attains its minimum in the compact set $\B\big(\vO, (1+\sqrt{2})\sqrt{n} / b\big) \cap \R^{q+1}_+$. Now, since $J$ is strictly convex and continuous, it must attain its minimum at a unique point in $\B\big(\vO, (1+\sqrt{2})\sqrt{n} / b\big) \cap \R^{q+1}_+$.
\end{proof}

\section{Phase Transition of Linear Inverse Problems with $\ell_2$ Norm Constraints}

\subsection{Proof of Proposition \ref{prop: minimizer of J in bounded case}} \label{subsec: minimizer of J in bounded case}

Assume that $f_0$ is some norm. For any non-zero point $\fxs \in \R^n$, we know from \cite[Example VI.3.1]{HIRIART1993} that the subdifferential of $f_0$ at $\fxs$ is
\begin{equation} \label{eq: subdifferential of norm}
  	\partial f_0(\fxs) = \big\{ \vs \in \R^n: \left< \vs, \fxs \right> = f_0(\fxs) \ \textnormal{and} \ f_0^{\circ}(\vs) = 1 \big\},
\end{equation}
where $f_0^{\circ}$ is the dual norm to $f_0$. To find the minimum of $J_1$, let us compute the differential of $J_1$ first. Recall our previous results \eqref{eq: partial derivative of J_u} and \eqref{eq: differential of J}. The partial derivative of $J_1$ with respect to $\vtau_1$ satisfies
\begin{align*}
	\frac{\partial J_1}{\partial \vtau_1}(\vtau)
	&= \E \frac{\partial J_{\vg}}{\partial \vtau_1}(\vtau) = \E \Big[ -2 \left< \vg - \big(\vtau_0 \bar{\vs} + \vtau_1 \frac{\fxs}{\|\fxs\|_2} \big), \frac{\fxs}{\|\fxs\|_2} \right> \Big] = 2 \E \left< \vtau_0 \bar{\vs} + \vtau_1 \frac{\fxs}{\|\fxs\|_2}, \frac{\fxs}{\|\fxs\|_2} \right> \\
	&= 2 \vtau_0 \cdot \frac{f_0(\fxs)}{\|\fxs\|_2} + 2 \vtau_1 = 2 \vtau_0 \cdot f_0(\fxs/\|\fxs\|_2) + 2 \vtau_1 \ge 0.
\end{align*}
To reach the first identity in the second line, we use the fact that $\left< \bar{\vs} , \fxs \right> = f_0(\fxs)$ for $\bar{\vs} \in \partial f_0(\fxs)$. The second identity in the second line results from the homogeneity property of norm. Since $\fxs \neq \vO$, we have $\frac{\partial J_1}{\partial \vtau_1}(\vtau) = 0$ if and only if $\vtau = (0, 0)$. Now we argue that the minimizer $\vtau^{\star}$ satisfies $\vtau^{\star}_1 = 0$. If not, we have $\frac{\partial J_1}{\partial \vtau_1}(\vtau^{\star}) > 0$. Since $\frac{\partial J_1}{\partial \vtau_1}$ is continuous, we know that there exists some $\epsilon > 0$ such that $\frac{\partial J_1}{\partial \vtau_1}(\vtau^{\star}_0, c) > 0$ when $0 \le \vtau^{\star}_1 - \epsilon < c < \vtau^{\star}_1$. By the first-order condition for strictly convex function, we obtain
$$
J_1(\vtau^{\star}_0, \vtau^{\star}_1) > J_1(\vtau^{\star}_0, c) + \Big[ \frac{\partial J_1}{\partial \vtau_1}(\vtau^{\star}_0, c) \Big] \cdot (\vtau^{\star}_1 - c) > J_1(\vtau^{\star}_0, c).
$$
This contradicts with the assumption that $\vtau^{\star}$ is the unique minimizer of $J_1$. Therefore, we conclude that $\vtau^{\star}_1$ must be zero. It follows that $\vtau^{\star}_0$ is the unique minimizer of the function
$$
J_2(\tau) \coloneqq J_1(\tau, 0) = \E \dist^2 \big(\vg, \tau \cdot \partial f_0(\fxs) \big),
$$
and the infimum of $J_1$ and $J_2$ are equal. For the function $J_2$, Amelunxen \textit{et al.} have studied its properties: It is strictly convex, continuously differentiable in $\R_+$, and attains its minimum at a unique point. See \cite[Proposition 4.1]{AMEL2014} for details. This completes the proof.

\subsection{Proof of Proposition \ref{prop: equivalence of statistical dimensions}} \label{subsec: proof of equivalence of statistical dimensions}

Assume that $f_0$ is a norm. For any non-zero point $\fxs \in \R^n$ and any $\vs \in \partial f_0(\fxs)$, we have
\begin{equation} \label{eq: dot product of s and x}
	\left< \vs, \fxs \right> = f_0(\fxs) > 0.
\end{equation}
Since both $\partial f_0(\fxs)$ and $\partial \|\fxs\|_2$ are non-empty, compact, and do not contain the origin, we have $N(f_0, \fxs) = \cone \big( \partial f_0(\fxs) \big)$ and $N(\|\cdot\|_2, \fxs) = \cone(\fxs)$. Therefore, take any $\va \in N(f_0, \fxs)$ and $\vb \in N(\|\cdot\|_2, \fxs)$. The relation in \eqref{eq: dot product of s and x} implies that
$$
\left< \va, \vb \right> \ge 0.
$$
As a result of Fact \ref{fact: statistical dimension of sum}, we obtain that
\begin{equation} \label{eq: comparison_two_cone_one_side}
	\delta\big( N(f_0, \fxs) + N(\|\cdot\|_2, \fxs) \big) \le \delta \big( N(f_0, \fxs) \big) + \delta \big( N(\|\cdot\|_2, \fxs) \big) = \delta \big( N(f_0, \fxs) \big) + \frac{1}{2}.
\end{equation}
The identity holds because $\delta \big( N(\|\cdot\|_2, \fxs) \big) = 1/2$. This point results from the fact that $\delta(\R_+) = 1/2$ \cite[pp. 241]{AMEL2014}, the rotational invariance of the statistical dimension \cite[Proposition 3.8 (6)]{AMEL2014} and the embedding property of the statistical dimension \cite[Proposition 3.8 (9)]{AMEL2014}. On the other hand, since $N(f_0, \fxs) \subseteq N(f_0, \fxs) + N(\|\cdot\|_2, \fxs)$, we trivially have
\begin{equation} \label{eq: comparison_two_cone_two_side}
	\delta\big( N(f_0, \fxs) \big) \le \delta\big( N(f_0, \fxs) + N(\|\cdot\|_2, \fxs) \big).
\end{equation}
This is a consequence of the monotonicity property of the statistical dimension \cite[Proposition 3.8 (10)]{AMEL2014}. Putting \eqref{eq: comparison_two_cone_one_side} and \eqref{eq: comparison_two_cone_two_side} together, we obtain that
\begin{equation} \label{eq: relation of two normal cones}
\delta\big( N(f_0, \fxs) \big) \le \delta\big( N(f_0, \fxs) + N(\|\cdot\|_2, \fxs) \big) \le \delta \big( N(f_0, \fxs) \big) + \frac{1}{2}.
\end{equation}
Furthermore, note that $\fdc_1^{\circ} = N(f_0, \fxs) + N(\|\cdot\|_2, \fxs)$ and $\fdc_2^{\circ} = N(f_0, \fxs)$. By the complementarity property of the statistical dimension \cite[Proposition 3.8 (8)]{AMEL2014}, 
\begin{equation} \label{eq: relation of normal and descent}
\delta(\fdc_1) = n - \delta\big( N(f_0, \fxs) + N(\|\cdot\|_2, \fxs) \big) \quad \textnormal{and} \quad \delta(\fdc_2) = n - \delta \big( N(f_0, \fxs) \big)
\end{equation}
Simply combining \eqref{eq: relation of two normal cones} and \eqref{eq: relation of normal and descent} completes the proof.

\section{Phase Transition of Linear Inverse Problem with Non-negativity Constraints} \label{sec: proof linear inverse non-negative}

\subsection{Proof of Proposition \ref{prop: error bound for calc}} \label{subsec: error bound for non-negativity}

The proof is similar with that in \cite[Appendix C.2]{AMEL2014}. For the sake of completeness, we include the whole proof here. Before we begin to prove Proposition \ref{prop: error bound for calc}, we first show that for any $\vg \in \R^n$, there is a unique $\tau_{\vg}$ satisfies $J_{\vg}(\tau_{\vg}) = \inf_{\tau \ge 0} J_{\vg}(\tau)$, where 
$$
J_{\vg}(\tau) = \dist^2 \big(\vg, N + \tau \cdot \partial f_0(\fxs)\big).
$$
We prove this point by contradiction. Suppose there are $\tau_1,\tau_2 \ge 0$, $\tau_1 \neq \tau_2$, satisfying $J_{\vg}(\tau_1) = J_{\vg}(\tau_2) = \inf_{\tau \ge 0} J_{\vg}(\tau)$, then 
	$$
	\dist\big(\vg, N + \tau_1 \cdot \partial f_0(\fxs)\big) = \dist\big(\vg, N+\tau_2 \cdot \partial f_0(\fxs)\big) = \inf_{\tau \ge 0} \dist\big(\vg, N + \tau \cdot \partial f_0(\fxs) \big) = \dist(\vg, N + K), 
	$$
	where $K = \cone\big(\partial f_0(\fxs)\big) = \bigcup_{\tau \ge 0} \tau \cdot \partial f_0(\fxs)$. Since $N+K$ is convex and closed, the projection of $\vg$ onto it is unique \cite[pp. 116]{HIRIART1993}. Therefore, there exist $\vt_1,\vt_2 \in N$ and $\vs_1,\vs_2 \in \partial f_0(\fxs)$ satisfying
	\begin{equation} \label{eq: unique projection}
		\vt_1 + \tau_1 \vs_1 = \vt_2 + \tau_2 \vs_2 = \Pi_{N+K}(\vg).
  	\end{equation}
	However, note that $N$ is the normal cone of $I_{\R^n_{+}}$ at $\fxs$, and its definition \eqref{eq: normal cone of nonegative} implies $\left< \vt_1, \fxs \right> = \left< \vt_2, \fxs \right> = 0$. It follows that
	\begin{equation} \label{eq: product with fxs}
		\left< \vt_1 + \tau_1 \vs_1, \fxs \right> = \tau_1 \left< \vs_1, \fxs \right> = \tau_1 f_0(\fxs) \ \textnormal{and} \ \left< \vt_2 + \tau_2 \vs_2, \fxs \right> = \tau_2 \left< \vs_2, \fxs \right> = \tau_2 f_0(\fxs).
  	\end{equation}
	Combining \eqref{eq: unique projection} and \eqref{eq: product with fxs}, we see that
	$$
	\tau_1 f_0(\fxs) = \tau_2 f_0(\fxs).
	$$
	Since $\fxs \neq \vO$, it holds that $f_0(\fxs) \neq 0$. This contravenes the assumption that $\tau_1 \neq \tau_2$. So the optimal $\tau_g$, which satisfies $J_{\vg}(\tau_{\vg}) = \inf_{\tau \ge 0} J_{\vg}(\tau)$, is unique.

	Now let us derive the bound in Proposition \ref{prop: error bound for calc}. Since $J_3(\tau)$ is strictly convex, it attains its infimum at a unique point, so we may define $\tau_{\star}$ as
	$$
	\tau_{\star} \coloneqq \argmin_{\tau \ge 0} J_3(\tau).
	$$
	Moreover, we have proved that for any $\vg \in \R^n$, the function $J_{\vg}(\tau)$ attains its infimum at a unique point $\tau_{\vg}$. Using the first-order condition for convex function, we can bound the error between $J_{\vg}(\tau_{\vg})$ and $J_{\vg}(\tau_{\star})$ as follows:
	$$
	J_{\vg}(\tau_{\vg}) \ge J_{\vg}(\tau_{\star}) + (\tau_{\vg} - \tau_{\star}) \cdot J'_{\vg}(\tau_{\star}).
	$$
	Taking expectation both sides with respect to $\vg$ yields
	\begin{align} \label{eq: error by variance}
	  	\E \Big[ \inf_{\tau \ge 0} J_{\vg}(\tau) \Big] 
		&\ge \E \big[ J_{\vg}(\tau_{\star}) \big] + \E \big[ (\tau_{\vg} - \tau_{\star}) \cdot J'_{\vg}(\tau_{\star}) \big] \notag \\
		&= J_3(\tau_{\star}) + \E \Big[ (\tau_{\vg} - \tau_{\star}) \cdot \big( J'_{\vg}(\tau_{\star}) - \E [ J'_{\vg}(\tau_{\star}) ] \big) \Big] + \E (\tau_{\vg} - \tau_{\star}) \cdot \E \big[ J'_{\vg}(\tau_{\star}) \big] \notag \\
		&= J_3(\tau_{\star}) + \E \Big[ (\tau_{\vg} - \E \tau_{\vg}) \cdot \big( J'_{\vg}(\tau_{\star}) - \E [ J'_{\vg}(\tau_{\star}) ] \big) \Big] + \E (\tau_{\vg} - \tau_{\star}) \cdot \E \big[ J'_{\vg}(\tau_{\star}) \big] \notag \\
		&\ge \inf_{\tau \ge 0} J_3(\tau) - \Big[ \var(\tau_{\vg}) \cdot \var\big( J'_{\vg}(\tau_{\star})\big) \Big]^{1/2} + \E (\tau_{\vg} - \tau_{\star}) \cdot J'_{\vg}(\tau_{\star}).
	\end{align}
	The second identity holds because the term $J'_{\vg}(\tau_{\star}) - \E [ J'_{\vg}(\tau_{\star}) ]$ have zero mean. The last inequality is a consequence of the Cauthy-Schwarz inequality. Therefore, to bound the error, it is sufficient to bound the variances and the last term. 

	First, the last term is nonnegative, i.e.,
	\begin{equation} \label{eq: last term nonnegative}
	  	\E (\tau_{\vg} - \tau_{\star}) \cdot J'_{\vg}(\tau_{\star}) \ge 0.
	\end{equation}
	To see this, we consider to cases. Define $e_1 \coloneqq \E (\tau_{\vg} - \tau_{\star}) \cdot J'_{\vg}(\tau_{\star})$. On one hand, when $\tau_{\star} > 0$, the derivative $J'_{\vg}(\tau_{\star}) = 0$ because $\tau_{\star}$ is the minimizer of $J$. Hence, $e_1 = 0$. On the other hand, when $\tau_{\star} = 0$, the right derivate $J'(0)$ must be nonnegative, otherwise, since $J'(\tau)$ is continuous, $J'(0) < 0$ will imply that $J(0)$ is not the minimum of $J$. Combining this observation with the fact that $\tau_{\vg} \ge 0$, we see $e_1 \ge 0$.

	Next, let us verify that the map $\vg \mapsto \tau_{\vg}$ is Lipschitz, and compute the variance of $\tau_{\vg}$. Indeed, \eqref{eq: product with fxs} indicates that $\tau_{\vg}$ has the following expression:
	$$
	\tau_{\vg} = \frac{\left< \Pi_{N+K}(\vg), \fxs \right> }{f_0(\fxs)}.
	$$
	Therefore, for any $\vg, \vg' \in \R^n$, we have
	\begin{align*}
		|\tau_{\vg} - \tau_{\vg'}| 
		&= \Big| \frac{\left< \Pi_{N+K}(\vg), \fxs \right> }{f_0(\fxs)} - \frac{\left< \Pi_{N+K}(\vg'), \fxs \right> }{f_0(\fxs)} \Big| = \frac{1}{f_0(\fxs)} \big| \left< \Pi_{N+K}(\vg) - \Pi_{N+K}(\vg'), \fxs \right> \big| \\
		&\le \frac{\|\fxs\|_2}{f_0(\fxs)} \cdot \big\| \Pi_{N+K}(\vg) - \Pi_{N+K}(\vg') \big\|_2 \le \frac{\|\fxs\|_2}{f_0(\fxs)} \cdot \| \vg - \vg' \|_2.
  	\end{align*}
	In the last inequality, we have used the fact that the projection onto a convex set is non-expansive. Thus, the variance of $\tau_{\vg}$ can be bounded by \cite[Fact C.3]{AMEL2014}:
	\begin{equation} \label{eq: variance of tau_g}
	  	\big(\var(\tau_{\vg})\big)^{1/2} \le \frac{\|\fxs\|_2}{f_0(\fxs)} = \frac{1}{f_0(\fxs/ \|\fxs\|_2)}.
	\end{equation}

	Then, let us compute the variance of $J_{\vg}'(\vtau)$ as a function of $\vg$. For this purpose, note that Lemma \ref{lem: distance to sum of sets} already shows that $J_{\vg}'(\tau)$ is a Lipschitz function of $\vg$ with the Lipschitz constant $2 \sup_{\vs \in \partial f_0(\fxs)} \|\vs\|_2$. Again, \cite[Fact C.3]{AMEL2014} delivers the bound
	\begin{equation} \label{eq: variance of J'_g}
	  	\big(\var[J_{\vg}'(\tau)]\big)^{1/2} \le 2 \sup_{\vs \in \partial f_0(\fxs)} \|\vs\|_2.
	\end{equation}

	At last, combining \eqref{eq: error by variance}, \eqref{eq: last term nonnegative}, \eqref{eq: variance of tau_g}, and \eqref{eq: variance of J'_g}, we obtain Proposition \ref{prop: error bound for calc}.

\subsection{Statistical dimension of the prior feasible descent cone of the $\ell_1$ minimization with nonnegative constraints} \label{app: calc_l1_nonnegative}

Without loss of generality, we assume that the first $s$ coordinates of $\fxs$ are positive, and the last $n-s$ coordinates are zero. Note that the subdifferential of $\|\cdot\|_1$ at $\fxs$ is
\begin{displaymath}
	\vu \in \partial \|\fxs\|_1 \Leftrightarrow \left\{
	\begin{array}{ll}
	  	\vu_i = 1, & \textnormal{when} \ \fxs_i > 0, \\
		-1 \le \vu_i \le 1, & \textnormal{when} \ \fxs_i = 0.
	\end{array}
\	\right.
\end{displaymath}
Therefore, for any $\tau \ge 0$, we have
$$
S(\tau) = N + \tau \cdot \partial \|\fxs\|_1 = \big\{ \fx \in \R^n: \fx_i = \tau \ \textnormal{for} \ 1 \le i \le s, \ \textnormal{and} \ \fx_i \le \tau \ \textnormal{for} \ s < i \le n \big\}.
$$
It follows that
$$
\dist^2 \big(\vg, S(\tau)\big) = \sum_{i=1}^s (\vg_i - \tau)^2 + \sum_{i=s+1}^n [ \max(\vg_i - \tau, 0) ]^2.
$$
Hence, the function $J_3(\tau)$ is
$$
J_3(\tau) = \E \dist^2 \big(\vg, S(\tau)\big) = \sum_{i=1}^s \E (\vg_i - \tau)^2 + \sum_{i=s+1}^n \E [ \max(\vg_i - \tau, 0) ]^2 = s(1+\tau^2) + \frac{1}{2}(n-s) \int_{\tau}^{\infty} (u - \tau)^2 \varphi(u) \mathrm{d}u,
$$
where the function $\varphi(u) = \sqrt{\frac{2}{\pi}} e^{-u^2/2}$. Now, denote $\psi_2: [0,1] \rightarrow [0, 1]$ the following function:
\begin{equation*} \label{eq: function psi 2}
  \psi_2(\rho) = \inf_{\tau \ge 0} \Big\{ \rho(1+\tau^2) + \frac{1}{2}(1-\rho) \int_{\tau}^{\infty} (u - \tau)^2 \varphi(u) \mathrm{d}u \Big\}.
\end{equation*}
By Corollary \ref{coro: bound_l1_nonnegative}, we reach the following relation:
$$
\delta(\fdc_3) \le n \cdot \psi_2(s/n).
$$
For the lower bound, we need to bound the term 
$$
\frac{2 \sup \{\|\vs\|_2: \vs \in \partial \|\fxs\|_1\}}{\|\fxs\|_1 / \|\fxs\|_2}.
$$
To this end, first note that 
$$
2\sup_{\vs \in \partial \|\fxs\|_1} \|\vs\|_2 = 2 \sqrt{n}.
$$
Moreover, since all non-negative vectors with exactly $s$ positive entries generate the same subdifferential, and hence, the same prior restricted cone, so we may select each of the positive entries to be $1$, and obtain that $\|\fxs\|_1 / \|\fxs\|_2 = \sqrt{s}$. The lower bound follows immediately.

Next, let us check the infimum in \eqref{eq: function psi} is attained at the unique solution of the stationary equation \eqref{eq: stationary equation}. Recall that Lemma \ref{lem: expected distance to the sum of multiple sets} shows that the infimum of $J(\tau)$ must be attained at a unique point. Moreover, we can compute the right derivative of $J(\tau)$ at the origin, and find that it is negative. Therefore, the infimum of the function $J(\tau)$ must be attained when $J'(\tau) = 0$. Simplifying $J'(\tau) = 0$ leads to the stationary equation \eqref{eq: stationary equation}.

\section{Proof of Fact \ref{fact: statistical dimension of sum}} \label{sec: proof of statistical dimension of sum}

We treat the case when $\left< \va, \vb \right> = 0$ for any $\va \in \K_1$ and $\vb \in \K_2$. The other two cases are similar. The statistical dimension of a convex cone can be expressed via its polar \cite[Proposition 3.1 (4)]{AMEL2014}, so we have
\begin{align*}
  	\delta \big( (\K_1 + \K_2)^{\circ} \big) 
	&= \E \dist^2 (\vg, \K_1 + \K_2) = \E \inf_{\va \in K_1, \vb \in \K_2} \| \vg - \va - \vb \|_2^2 \\
	&= \E \inf_{\va \in \K_1, \vb \in \K_2} \big( \|\vg\|_2^2 + \|\va\|_2^2 + \|\vb\|_2^2 -2\left< \vg, \va\right> - 2 \left< \vg, \vb \right> + 2\left< \va, \vb\right> \big) \\
	&= \E \inf_{\va \in \K_1, \vb \in \K_2} \big( \|\vg - \va\|_2^2 + \|\vg - \vb\|_2^2 - \|\vg\|_2^2 \big) 
	= \E \inf_{\va \in \K_1} \|\vg-\va\|_2^2 + \E \inf_{\vb \in \K_2} \|\vg - \vb\|_2^2 - \E \|\vg\|_2^2 \\
	&= \E \dist^2 (\vg, \K_1) + \E \dist^2 (\vg, \K_2) - n \\
	&= \delta(K_1^{\circ}) + \delta(\K_2^{\circ}) - n.
\end{align*}
The sum of the statistical dimension of a convex cone and that of its polar equals the ambient dimension \cite[Proposition 3.1 (8)]{AMEL2014}. It follows that
\begin{align*}
	\delta(\K_1 + \K_2) 
	= n - \delta \big( (\K_1 + \K_2)^{\circ} \big) = n - \big[ \delta(K_1^{\circ}) + \delta(\K_2^{\circ}) - n \big]
	= [ n - \delta(\K_1^{\circ}) ] + [ n - \delta(\K_2^{\circ}) ] = \delta(\K_1) + \delta(\K_2).
\end{align*}

\bibliographystyle{IEEEtran}
\bibliography{IEEEabrv,references}

\end{document}